%% file: main.tex
\documentclass[sigconf,nonacm]{aamas}
\input{includes}
\input{aamas2025_header.tex}
\input{definitions}

\title{Game Theory with Simulation\texorpdfstring{\\}{ }in the Presence of Unpredictable Randomisation}

    \author{Vojtěch Kovařík}
    \affiliation{
      \institution{Carnegie-Mellon University}
      \city{Pittsburgh}
      \country{United States}}
    \affiliation{
      \institution{Czech Technical University}
      \city{Prague}
      \country{Czech Republic}}
    \email{vojta.kovarik@gmail.com}

    \author{Nathaniel Sauerberg}
    \affiliation{
      \institution{University of Texas}
      \city{Austin}
      \country{United States}}
    \email{njsauerberg@gmail.com}

    \author{Lewis Hammond}
    \affiliation{
      \institution{University of Oxford}
      \city{Oxford}
      \country{United Kingdom}}
    \email{lewis.hammond@cs.ox.ac.uk}

    \author{Vincent Conitzer}
    \affiliation{
      \institution{Carnegie-Mellon University}
      \city{Pittsburgh}
      \country{United States}}
    \email{conitzer@cs.cmu.edu}
    

\begin{abstract}
    \input{abstract}
\end{abstract}

\begin{document}

% AAMAS STUFF HERE >>>>>>>
%%% The following commands remove the headers in your paper. For final 
%%% papers, these will be inserted during the pagination process.

\pagestyle{fancy}
\fancyhead{}
% <<<<< END OF AAMAS STUFF

\maketitle

% % AAMAS DOES NOT HAVE ABSTRACT HERE, BUT IN PREAMBLE
% \begin{abstract}
%     \input{abstract}
% \end{abstract}

\input{intro}
\input{background}
\input{the_actual_new_stuff}
\input{literature}        
\input{conclusions}

\anonymise{
    \input{acknowledgements}
}{}

% \clearpage
\balance
    % Required by AAMAS. Use somewhere on the first column of the last page.
    
\bibliographystyle{ACM-Reference-Format}
\bibliography{main}

\clearpage
\appendix
\input{efgs.tex}
\input{proofs.tex}

\end{document}

%% file: includes.tex
\usepackage{amsmath,amsfonts,amsthm,mathtools}
	\numberwithin{equation}{section} %number equations by section
\usepackage{enumitem}[shortlabels]
\usepackage{thm-restate}
\usepackage{hyperref}       % hyperlinks
\usepackage{cleveref}       % for not having to write Proposition~\ref{}
\usepackage{microtype}      % microtypography	
\usepackage{ifthen}         % for logic
\usepackage{nicematrix}     % for game matrices
\usepackage{csquotes} %Added by NS, for quotes
\usepackage{comment}

%% file: aamas2025_header.tex
\usepackage{balance} % for balancing columns on the final page

\setcopyright{ifaamas}
\acmConference[AAMAS '25]{Proc.\@ of the 24th International Conference
on Autonomous Agents and Multiagent Systems (AAMAS 2025)}{May 19 -- 23, 2025}
{Detroit, Michigan, USA}{A.~El~Fallah~Seghrouchni, Y.~Vorobeychik, S.~Das, A.~Nowe (eds.)}
\copyrightyear{2025}
\acmYear{2025}
\acmDOI{}
\acmPrice{}
\acmISBN{}

\acmSubmissionID{45}

%% file: definitions.tex
\newboolean{commentsactivated}
\setboolean{commentsactivated}{false}
\newcommand{\vc}[1]{\ifthenelse{\boolean{commentsactivated}}{{\color{blue} {\em VC: #1 }}}{}}
\newcommand{\co}[1]{\ifthenelse{\boolean{commentsactivated}}{{\color{magenta} {\em CO: #1 }}}{}}
\newcommand{\vk}[1]{\ifthenelse{\boolean{commentsactivated}}{{\color{red} {\em VK: #1 }}}{}}
\newcommand{\lh}[1]{\ifthenelse{\boolean{commentsactivated}}{{\color{teal} {\em LH: #1 }}}{}}
\newcommand{\ns}[1]{\ifthenelse{\boolean{commentsactivated}}{{\color{olive} {\em NS: #1 }}}{}}
\newcommand{\hideableComments}[1]{\ifthenelse{\boolean{commentsactivated}}{{\edit{#1}}}{}}

\newboolean{highlightEdits}
\setboolean{highlightEdits}{true}
\newcommand{\edit}[1]{\ifthenelse{\boolean{highlightEdits}}{{\color{purple}{#1}}}{#1}}

\newboolean{showArxivStuff}
\setboolean{showArxivStuff}{true}
\newcommand{\highlightArxiv}[1]{\ifthenelse{\boolean{highlightEdits}}{{\color{teal}{#1}}}{#1}}
\newcommand{\arxivOnly}[1]{\ifthenelse{\boolean{showArxivStuff}}{{\highlightArxiv{#1}}}{}}
\newcommand{\longOrShort}[2]{\ifthenelse{\boolean{showArxivStuff}}{{\highlightArxiv{#1}}}{#2}}

\newboolean{anonymous}
\setboolean{anonymous}{true}
    \setboolean{anonymous}{false}
\newcommand{\anonymise}[2]{\ifthenelse{\boolean{anonymous}}{{{#2}}}{#1}}
\newtheorem{definition}{\protect\definitionname}
    \numberwithin{definition}{section}      % Definition Sec.Num         format for definitions etc
\newtheorem{lemma}[definition]{\protect\lemmaname}

\newtheorem{corollary}[definition]{\protect\corollaryname}
\newtheorem{claim}[definition]{\protect\claimname}

\newtheorem{observation}[definition]{\protect\observationname}

\newtheorem{example}[definition]{\protect\examplename}
\newtheorem{notation}[definition]{\protect\notationname}

\theoremstyle{definition}
    \newtheorem{remark}[definition]{\protect\remarkname}

\newtheorem*{claim*}{\protect\claimname}

\providecommand{\corollaryname}{Corollary}
\providecommand{\claimname}{Claim}
\providecommand{\definitionname}{Definition}
\providecommand{\lemmaname}{Lemma}
\providecommand{\notationname}{Notation}
\providecommand{\remarkname}{Remark}
\providecommand{\problemname}{Problem}
\providecommand{\propositionname}{Proposition}
\providecommand{\examplename}{Example}
\providecommand{\theoremname}{Theorem}
\providecommand{\conjecturename}{Conjecture}
\providecommand{\observationname}{Observation}

\DeclareMathSymbol{\shortminus}{\mathbin}{AMSa}{"39}

\newcommand{\defword}[1]{\textbf{\boldmath{#1}}}

\newcommand{\N}{\mathbb{N}}
\newcommand{\R}{\mathbb{R}}
\newcommand{\E}{\mathbb{E}}
\newcommand{\mc}{\mathcal}

\newcommand{\symbolPlaceholder}{{\, \cdot \, }}

\newcommand{\actions}{\mc A}
\newcommand{\action}{a}
\newcommand{\actionPl}{a_1}
\newcommand{\actionOpp}{a_2}
\newcommand{\policy}{\pi}

\newcommand{\policies}{\Pi}

\newcommand{\playerSet}{\mc N}
\newcommand{\playerFunction}{p}
\newcommand{\pl}{i}

\newcommand{\opp}{\others}          % opponent
\newcommand{\others}{{\textnormal{-}\pl}}

\newcommand{\Plone}{\text{P1}}
\newcommand{\Pltwo}{\text{P2}}

\newcommand{\histories}{\mc H}
\newcommand{\history}{h}

\newcommand{\leaf}{z}
\newcommand{\leaves}{\mc Z}

\newcommand{\utility}{u}
\newcommand{\reward}{r}
\newcommand{\br}{\textnormal{br}}                     % best response
\newcommand{\NE}{\textnormal{NE}}                     % Nash equilibrium
\newcommand{\SE}{\textnormal{SE}}                     % Stackelberg equilibrium

\newcommand{\supp}{\textnormal{supp}\,}                % Support
\newcommand{\game}{{\mc{G}}}

\newcommand{\EFG}{E}

\renewcommand{\emptyset}{\varnothing}

\DeclareMathOperator*{\argmax}{arg\,max}

\DeclareMathSymbol{\shortminus}{\mathbin}{AMSa}{"39}

\newcommand{\NP}[0]{$\mathsf{NP}$}
\newcommand{\NPH}[0]{\NP-hard}

\newcommand{\actionSubset}{K}
\newcommand{\nOfActions}{n}
    \newcommand{\simcost}{{\textnormal{c}_\simSubscript}}
\newcommand{\simSubscript}{{\textnormal{sim}}}
\newcommand{\simSubscriptPure}{{\textnormal{p-sim}}}
\newcommand{\simSubscriptMixed}{{\textnormal{m-sim}}}
\newcommand{\psimgame}{\game_\simSubscriptPure}
\newcommand{\msimgame}{\game_\simSubscriptMixed}
    \newcommand{\SimulateInformal}{\textnormal{Simulate}}
    \newcommand{\SimInformal}{\textnormal{sim}}
\newcommand{\mSim}{\textnormal{m-sim}}
\newcommand{\pSim}{\textnormal{p-sim}}
\newcommand{\fbr}{\textnormal{f-br}}        % ``friendly'' best response

\newcommand{\simProb}{p_\simSubscript}
\newcommand{\bigO}{O}

\newcommand{\iPiN}{introduces Pareto-improving NE}
\newcommand{\iPiNwithoutS}{introduce Pareto-improving NE}

\newcommand{\Trust}{\textnormal{Trust}}
\newcommand{\PartialTrust}{\textnormal{Partial Trust}}
\newcommand{\FullTrust}{\textnormal{Full Trust}}
\newcommand{\WalkOut}{\textnormal{Walk Out}}
\newcommand{\Cooperate}{\textnormal{Cooperate}}
\newcommand{\Defect}{\textnormal{Defect}}
\newcommand{\T}{\textnormal{T}}
\newcommand{\PT}{\textnormal{PT}}
\newcommand{\FT}{\textnormal{FT}}
\newcommand{\WO}{\textnormal{WO}}
\newcommand{\C}{\textnormal{C}}
\newcommand{\D}{\textnormal{D}}
\newcommand{\incentivise}[1]{i^*[#1]}

\newcommand{\pureStrategy}{s}
    \newcommand{\pureStrategyAlt}{t}
    \newcommand{\PureStrategies}{S}
    \newcommand{\strategySubset}{\PureStrategies'}
\newcommand{\mixedStrategy}{\sigma}
    \newcommand{\mixedStrategyAlt}{\rho}
    \newcommand{\MixedStrategies}{\Sigma}
\newcommand{\pureMetaStrategy}{m}
    \newcommand{\PureMetaStrategies}{M}
\newcommand{\mixedMetaStrategy}{\mu}
    \newcommand{\mixedMetaStrategyAlt}{\nu}

\newcommand{\extremise}[1]{\left[#1\right]'}    % Replacing strategy by an equivalent convex combination of extremal points
\newcommand{\meta}[1]{\widehat{#1}}    % the operator for treating a strategy in the base game as a meta strategy
\newcommand{\mesa}[1]{\check{#1}}
\newcommand{\extremalPoints}{\textnormal{Ext}}
\newcommand{\inverseBR}{{\br}^{-1}}
\newcommand{\inverseFBR}{{\fbr}^{-1}}
\newcommand{\FinitePMS}{\PureMetaStrategies^{\textnormal{fin}}}
\newcommand{\ReducedPMS}{\PureMetaStrategies^{\textnormal{inf}}}
\newcommand{\specialPolytope}{X}
\newcommand{\closure}[1]{\overline{#1}}

\newcommand{\Bad}{\textnormal{B}}
\newcommand{\Neutral}{\textnormal{N}}
\newcommand{\Good}{\textnormal{G}}
\newcommand{\Awesome}{\textnormal{A}}

\newcommand{\Full}{^\textnormal{F}}
\newcommand{\Partial}{^\textnormal{P}}

\newcommand{\defectProbability}[1]{\delta}
\newcommand{\lowerD}[1]{\underline{\delta}_{#1}}
\newcommand{\upperD}[1]{\delta^*_{#1}}
\newcommand{\dFT}{\upperD{\FT}}

\newcommand{\evilProb}{p_{\D}}
\newcommand{\Tless}{\T_\textnormal{less}}
\newcommand{\Tmid}{\T_\textnormal{mid}}
\newcommand{\Tmore}{\T_\textnormal{more}}

\newcommand{\TG}{\texttt{TG}}
\newcommand{\PTG}{\texttt{PTG}}

\newcommand{\maxminV}{\underline v_1}
\newcommand{\maxminStrategy}{\underline{\pureStrategy}_1}
\newcommand{\allowedResponses}{\actions_2}
\newcommand{\optimalResponse}[1]{\pureStrategy_2^*[#1]}
\newcommand{\optimalCommitment}[1]{\overline{\mixedStrategy}_2[#1]}
\newcommand{\responseStrategy}{\varphi_2}
\newcommand{\responseStrategyAlt}{\psi_2}

\newcommand{\stackelberg}[1]{\mixedStrategy^{\textnormal{SE,\,}#1}_2}
\newcommand{\baselineIndex}{{k_1}}
\newcommand{\deviateIndex}{{k_2}}
\newcommand{\SEvalue}[2]{v^{\textnormal{SE}}_{#2}(\game_{#1})}
\newcommand{\metaNE}{\mixedMetaStrategy^*}
\newcommand{\Horrible}{\textnormal{H}}
\newcommand{\OO}{\textnormal{OO}}
\newcommand{\OptOut}{\textnormal{Opt Out}}

\newcommand{\Baseline}{\textnormal{B}}

\newcommand{\password}{p}
\newcommand{\guess}{g}
\newcommand{\passwordGuessing}{\game^\textnormal{PG}}
\newcommand{\PG}[2]{\textnormal{PG}_{#1}(#2)}
\newcommand{\dontGuess}{\neg \textnormal{g}}

\newcommand{\Rock}{\textnormal{R}}
\newcommand{\Paper}{\textnormal{P}}
\newcommand{\Scissors}{\textnormal{S}}

\newcommand{\graph}{\Gamma}

%% file: abstract.tex
AI agents will be predictable in certain ways that traditional agents are not.
Where and how can we leverage this predictability in order to improve social welfare?
We study this question in a game-theoretic setting where one agent can pay a fixed cost to simulate the other in order to learn its mixed strategy.
As a negative result, we prove that, in contrast to prior work on pure-strategy simulation, enabling mixed-strategy simulation may no longer lead to improved outcomes for both players in all so-called “generalised trust games”.
    In fact, mixed-strategy simulation does not help in any game where the simulatee’s action can depend on that of the simulator.
We also show that, in general, deciding whether simulation introduces Pareto-improving Nash equilibria in a given game is \NPH{}.
As positive results, we establish that mixed-strategy simulation can improve social welfare if the simulator has the option to scale their level of trust, if the players face challenges with both trust and coordination, or if maintaining some level of privacy is essential for enabling cooperation.

%% file: intro.tex
\section{Introduction}\label{sec:intro}

% General motivation
With the current pace of progress in AI,
    we are likely to increasingly see important interactions take place not only between humans,
    but also with and between AI agents \cite{dafoe2021a,dafoe2020open}.
To ensure that the societal impact of these interactions is positive, it is important to understand the ways in which AI agents differ from humans \citep{conitzer2023}.
This in turn can help us design interventions that promote socially desirable outcomes \citep{Shavit2023a}.

% Simulation
One important distinction between human and AI agents
    is that the behaviour of AI agents is determined by their source code, and can therefore -- in certain cases -- be reliably predicted \cite{Tennenholtz2004}.
This could be achieved, for example,
    by inspecting the AI's source code and reasoning about it,
    or by creating a copy of the AI and running it in a simulated environment.
As these examples suggest,
    we will assume that predicting the AI's actions requires non-trivial effort,
    and is therefore associated with some \textit{cost} \cite{kovarik2023game,halpern2008game}.
(Readers familiar with Stackelberg games \cite{von1934marktform} can think of this setting as one where the follower has to choose to pay some cost before they are allowed to see the leader's mixed commitment.
For a more detailed discussion of related work, see \Cref{sec:sub:related-work}.)
For concreteness, this paper will discuss this general topic in terms of \textit{simulating} the AI agent, though our results also apply to other forms of prediction.

% Simulation is plausible and matters
As we will show, the ability to simulate agents before interacting with them can (provably) lead to increased trust and cooperation.
More than a topic of merely theoretical interest, however, the availability of black-box access to the latest AI models and high-fidelity simulators could lead to simulation being a key tool in the safe and beneficial deployment of advanced AI agents \cite{balesni2024evaluationsbasedsafetycasesai,shevlane2023model}.
Importantly, in domains ranging from financial markets to public infrastructure, these agents will face \textit{strategic incentives}, making it critical to understand the implications of simulation in \textit{game-theoretic} settings.

% Studying simulation in \textit{game-theoretic} settings is thus critical for understanding and 
% its impact and developing effective interventions (ranging from software tools to legislation) 
%     gaining realistic expectations of its impact and
%     developing effective interventions that make this impact more beneficial
%         (e.g., tools, norms, and legislation).

% Pure- vs mixed-strategy simulation
An idealised variant of this setting was studied by \citet{kovarik2023game},
    who assumed that the simulator is able to predict the AI's action perfectly.
However, this assumption might often be unrealistic,
    not least because the AI might have access to a source of randomness that cannot be predicted by the simulator \cite{wang20134,jacak2021quantum}.
As the following extended example shows, this difference has far-reaching consequences, which we explore in the remainder of the paper.

\subsection{Illustrative Example}\label{sec:sub:illustrative_example}

\paragraph{Alice, Bob, and his robots.}
Consider a setting in which Alice (player one) and Bob (player two) are due to interact in some particular situation, corresponding formally to some arbitrary two-player game $\game$.
Instead of interacting directly, however, Bob will deploy a robot\footnote{It is in no way important to the paper's results that the ``robot'' is physically embodied; we just use it here as evocative language.} that will act on his behalf.
Moreover, Alice will have the option to analyse the robot,
    at some cost $\simcost > 0$. We will refer to this analysis as \textit{simulation}.

In general, the robot can use randomness to determine which action to take.
Correspondingly, we will distinguish between two types of simulation,
    depending on whether Alice is able to predict the robot's source of randomness or not.
To capture this distinction formally,
    we can assume that the robot corresponds to
    some probability distribution $\mixedStrategy_2$ over pure strategies in $\game$.
In \textit{mixed-strategy simulation},
    Alice learns the robot's mixed strategy $\mixedStrategy_2$.
In \textit{pure-strategy simulation},
    Alice learns which pure strategy $\pureStrategy_2$ will be sampled by the robot when playing $\game$.

In this work, we assume that if Alice decides to simulate the robot,
    she will then best-respond to the revealed strategy,
    breaking ties in Bob's favour.

\paragraph{The simulation meta-game.}
When Alice has access to mixed-strategy simulation,
    Alice and Bob need to reason not only about the ``base-game'' $\game$,
    but also about the ``simulation meta-game'' $\game_\simSubscript$.
In $\game_\simSubscript$,
    Alice must decide whether to simulate, and which strategy to use if she does not.
Correspondingly, Bob must decide which robot -- i.e., which mixed strategy -- to select as his representative.
When Bob randomises his choice, he is thus \textit{mixing over mixed strategies}.
And because Alice can simulate the robot but not Bob,
    Bob's overall strategy is not necessarily reducible to a single mixed strategy.\footnotemark{}
    \footnotetext{
        If Bob was certain that Alice has access to pure-strategy simulation,
            he might decide to use robots that play deterministically,
            since doing otherwise confers no benefit.
        This would make the additional level of randomisation unnecessary.   
    }

% Controlling simcost - and what that implies
\paragraph{Focusing on low simulation costs and strict Pareto improvements.}
Throughout the paper, we will assume that the simulation cost is not under the control of either of the players.
However, for the purpose of interpreting the results, note that Bob in particular might be able to influence $\simcost$ to some degree.
For example, when Bob has a fully adversarial relationship with Alice,
    he might not share any details about the robot
    and even intentionally obfuscate its design to make the analysis harder.
In contrast, when Bob wants Alice to trust him,
    he might make the robot easier to understand
    or even subsidise the simulation cost.\footnotemark{}
    \footnotetext{
        However, in practice, it seems unlikely that Bob could achieve $\simcost = 0$ or even $\simcost < 0$.
        This is because even if Bob makes simulation as simple as possible and pays Alice to simulate, she will always be tempted to put lower than maximum effort into her analysis
        (to save effort or keep some of the subsidy for herself).
    }
Because of this, in this paper we will primarily investigate the case where
    $\simcost$ is \textit{low but positive}.
We will also focus on settings where
    simulation holds the promise of improving the outcome for \textit{both} players.
The primary candidates for such settings are games that revolve around trust, coordination, or both.

\paragraph{The trust game \emph{TG}}
The central example of a trust game is $\TG$ (\Cref{fig:TG}):
    Bob -- or rather, his robot -- approaches Alice with an investment opportunity.
    If she lends him \$100k, he will make \$40k in profit. %which he can then split 50:50 with Alice.
    Alice can $\WalkOut$ ($\WO$) on Bob,
        terminating the game with payoffs $\utility_A = \utility_B = 0$.
    If Alice instead $\Trust$s ($\T$) Bob,
        he can either $\Defect$ ($\D$) and steal Alice's money ($\utility_A = -100$, $\utility_B = 100$)
        or $\Cooperate$ ($\C$), splitting the profits 50:50 with Alice ($\utility_A = \utility_B = 20$).
Unfortunately, without simulation, $\Defect$ is a dominant strategy for Bob,
    so the only Nash equilibrium (\NE{}) of $\TG$ is for Alice to $\WalkOut$.    

\paragraph{Pure-strategy simulation in \emph{TG}}
In the simulation variant $\TG_\simSubscript$ of $\TG$,
    Bob attempts to earn Alice's trust by sharing the robot's specification with her.
She then has the option to $\SimulateInformal$ ($\SimInformal$) the robot for $\simcost = $ \$2k in order to learn which strategy it will employ.
One might hope that this would reliably allow cooperation between Alice and Bob,
    yielding payoffs $\utility_A = 20 \shortminus 2 = 18$, $\utility_B = 20$.
Unfortunately, such an outcome would not be stable, because
    Alice would be tempted to increase her profits by $\Trust$ing Bob blindly
        (and thus saving the simulation cost).
    This, in turn, creates a temptation for Bob to submit a robot that will $\Defect$ on Alice.

\paragraph{Pure-strategy simulation in generalised trust games.}
\citet{kovarik2023game} show that in the pure-strategy simulation game $\TG_\simSubscriptPure$,
    it is an equilibrium
        for Alice to mix between $\Trust$ and $\SimulateInformal$,
        and for Bob to mix between (robots that) $\Cooperate$ and $\Defect$.
    While this simulation equilibrium is not optimal in terms of social welfare
        (compared to what could be obtained in a world without strategic constraints),
        it constitutes a strict improvement over the original equilibrium for both Alice and Bob.
The authors then show that a similar result holds in any \textit{generalised} trust game
    (which they define as a game where
        giving Bob the ability to make credible pure-strategy commitments
        is guaranteed to strictly improve the utility for both players
        compared to any \NE{} of the original game).

\begin{figure}[tb]
    \centering
    \begin{NiceTabular}{rccc}[cell-space-limits=3pt]
                        & $\Cooperate$     & $\Defect$          & $\ \ \ $\\
        $\Trust$    & \Block[hvlines]{2-2}{}
                          $20, 20$          & $-100, 100$       & $\ \ \ $ \\
        $\WalkOut$      & $0, 0$            & $\ \ \ 0, 0 \ $   & $\ \ \ $ 
    \end{NiceTabular}
    \begin{NiceTabular}{rccc}[cell-space-limits=3pt]
                                & $\Cooperate$     & $\Defect$ & $\ \ \ $\\
        $\FullTrust$ & \Block[hvlines]{3-2}{}
                                $20, 20$    & $-100, 100$           & $\ \ \ \ \ \ \ $ \\
        $\PartialTrust$ &       $10, 10$    & $\ \, -25, 25 \ $     & $\ \ \ \ \ \ \ $ \\
        $\WalkOut$ &            $0, 0$      & $\ \ \ 0, 0 \ $       & $\ \ \ \ \ \ \ $ 
    \end{NiceTabular}
    \caption{Trust game $\TG$ and its partial-trust extension $\PTG$.}
    \Description{Trust game TG and its partial-trust extension PTG.}
    \label{fig:TG}\label{fig:PTG}
\end{figure}

\paragraph{Mixed-strategy simulation is useless in \emph{TG}}
Unfortunately, this result no longer holds in the mixed-strategy simulation variant $\TG_\simSubscriptMixed$.
To see why, imagine that Bob is about to submit a robot which cooperates with Alice 100\% of the time.
He will then be tempted to replace this robot by one that only cooperates 99\% of the time.
    If Alice simulates the robot, surely she will not $\WalkOut$ on him just because of the 1\% defection chance; after all, her expected utility for $\Trust$ing is still positive.
In fact, this reasoning shows that
    Bob can safely set the cooperation rate to $\frac{100+2}{100+20} = 85\%$,
        the lowest possible value where Alice still recoups all of her simulation costs.
However,
    % \edit{because we are assuming that the interaction is not repeated
         % (and there are no reputation effects, etc.),}
    Bob can go even further.
    He can reason that once Alice has already simulated,
        she will treat the simulation cost as a sunk cost.
    Taken all together, this means that no matter what Alice does,
        he might as well replace all of his ``100\% cooperate robots'' by ``83.5\% cooperate robots''
        (since a cooperation rate of $\frac{100}{100+20}$ is the lowest he can go while still giving $\Trust$ a positive expected value).
Unfortunately, at this point,
    Alice no longer makes enough profit to recoup $\simcost$,
    so her only sensible options are to $\Trust$ Bob blindly or $\WalkOut$.
This effectively brings Alice and Bob back to $\TG$ \textit{without simulation},
    and the only Nash equilibrium of that game is for Alice to $\WalkOut$.

\paragraph{Mixed-strategy simulation can help if Alice has good alternatives.}
In light of this negative result,
    one might wonder whether there are \textit{any} games where mixed-strategy simulation is useful.
Fortunately, it turns out that there are.
To see this, consider an extension $\PTG$ (\Cref{fig:PTG}) of the earlier trust game scenario $\TG$
    where Alice has the additional option to trust Bob only \textit{partially} ($\PT$),
        \textit{with the robot being unable to differentiate between $\PartialTrust$ and $\FullTrust$}.
    For example, she could secretly register her business in jurisdictions with higher taxes but more secure banking infrastructure,
        which would decrease the overall profits ($\utility_A(\PT, \C) = \utility_B(\PT, \C) = 10$)
        but reduce the robot's ability to steal from her
            ($\utility_A(\PT, \D) = \shortminus 25$, $\utility_B(\PT, \D) = 25$).
Bob could now repeat the same reasoning as before,
    concluding that any ``100\% cooperate robot'' can be safely replaced by a ``99\% cooperate robot''.
However, he will now have to stop at $\frac{100-25}{20-10} \doteq 88.3\%$.
This is because below this value, Alice's best response will switch from $\FullTrust$ to $\PartialTrust$,
    which will decrease Bob's utility.
    (This shows the importance of $\utility_2(\FT, \C)$ being higher than $\utility_2(\PT, \C)$.)
We can verify that
    the mixed-strategy simulation version of $\PTG$ has an \NE{} where
        Bob mostly submits a ``88.3\% cooperate robot''
            but sometimes replaces it by a ``100\% defect robot'',
        while Alice mixes between Simulating  (after which, depending on the result, she plays either $\FullTrust$ or $\WalkOut$) and blindly playing $\PartialTrust$.
Fortunately,
    the frequency of Bob using the ``100\% defect robot'' is proportional to $\simcost$,
    so for any sufficiently low $\simcost$ (e.g., for $\simcost = 2$),
    both players end up making a profit.

\subsection{Outline and Contributions}\label{sec:sub:outline}

\Cref{sec:background}
    describes the standard notation for normal-form games and covers some classic game-theoretic concepts.
    % such as opponent-friendly tie-breaking and Stackelberg equilibria.
% Extensive-form games are described in \Cref{sec:app:EFGs}.

In \Cref{sec:new_defs},
    we formally define mixed-strategy simulation games $\msimgame$,
    contrast them with pure-strategy simulation games $\psimgame$,
    and establish their basic properties.
    In particular, we show that
        while $\msimgame$ is technically an infinite game,
        it can be reduced to a normal-form game
            whose size is at most exponential in the size of the base-game
            (Prop.\,\ref{prop:finite_str_space}).

\Cref{sec:computational} explores the computational aspects of mixed-strategy simulation.
    First, it establishes an upper bound on the complexity of finding an \NE{} of $\msimgame$
        (Prop.\,\ref{prop:solve_upper_bound}).
    Second,
        while determining the exact complexity of finding an \NE{} of $\msimgame$ is left as an open problem,
        we observe that any \NE{} of the original game still exists as an equilibrium of the simulation game,
            so finding \textit{all} \NE{} of a simulation game is \NPH{}.
    Third, from the design point of view, a crucial question is whether enabling mixed-strategy simulation
        in a particular game introduces beneficial Nash equilibria
        -- e.g., ones that result in a Pareto-improvement, an increase in social welfare, or an improvement in the utility of a particular player
        (relative to the equilibria of the original game $\game$).
        We show that answering any variant of this question is, in general, \NPH{} (Thm.~\ref{thm:helps_is_NPH_v3}).
This implies that one should not expect to be able to find a simple description of simulation's effects in \textit{general} games.
    Consequently, we find it more promising to focus on
        identifying specific \textit{classes} of games where simulation has easily describable effects.

In \Cref{sec:welfare}, we investigate the effects of simulation on the players' welfare.
First, we extend the negative result for $\TG$ from \Cref{sec:sub:illustrative_example}, by showing that
    mixed-strategy simulation cannot help in any game where
        the robot observes Alice's base-game strategy before acting
        (Thm.\,\ref{thm:negative}).
We then extend the positive result for $\PTG$ from \Cref{sec:sub:illustrative_example},
    by describing a general class of games where
        Alice can vary her level of trust
        and mixed-strategy simulation allows her to profitably use the second-highest level of trust
        (Thm.\,\ref{thm:generalised_PTG}).
We also prove that mixed-strategy simulation is beneficial in a class of games involving elements of both trust and coordination (Thm.\,\ref{thm:generalised_DTG}).
Finally, we show that there are situations
    -- those where the simulated agent finds it important to maintain their privacy --
    where mixed-strategy simulation is more socially beneficial than pure-strategy simulation
    (Thm.\,\ref{theorem:password_guessing}).

We conclude by
    reviewing the most closely related work (\Cref{sec:sub:related-work})
    and summarising the paper's findings (\Cref{sec:sub:conclusion}).
Detailed proofs are in the appendix.

%% file: background.tex
\section{Background}\label{sec:background}

For a finite set $X$, $\Delta(X)$ denotes the \textbf{set of all probability distributions} over $X$.
For a probability distribution $\rho$, $\supp(\rho)$ denotes the \textbf{support} of $\rho$.
We use \Plone{} and \Pltwo{} as shorthands for ``player one'' and ``player two''.
When there is risk of confusion about which game a given object belongs to, we add superscript notation
    (e.g., $\utility^\game$ for utility in $\game$).

% \paragraph{Normal-form games.}
A two-player \defword{normal-form game} (NFG) $\game$
    is a triplet $(\PureStrategies_1, \PureStrategies_2, \utility)$ where:
        $\PureStrategies := \PureStrategies_1 \times \PureStrategies_2 \neq \emptyset$ is a set of \textbf{pure strategy profiles}
            (finite, unless specified otherwise)
        and $\utility = (\utility_1, \utility_2) : \PureStrategies \to \R^2$ is the \textbf{utility function}.
We will typically denote the elements of $\PureStrategies_\pl$ (pure strategies) as $\pureStrategy_\pl$.
A \textbf{mixed strategy} $\mixedStrategy_\pl$ is a probability distribution over pure strategies.
$\MixedStrategies_\pl := \Delta(\PureStrategies_\pl)$ denotes the set of all mixed strategies.
% \vk{We might need special notation for ``the mixed strategy which puts all probability on a given pure strategy''.}
Since
    any pure strategy $\pureStrategy_\pl$ can be identified with the mixed strategy $\sigma_\pl^{\pureStrategy_\pl}$ that selects $\pureStrategy_\pl$ with probability $1$,
    we sometimes view pure strategies as a subset of mixed strategies.
A \textbf{subgame} of $\game$ is any game of the form
    $\game' = (\PureStrategies'_1, \PureStrategies'_2, \utility^\game)$, where
    $\PureStrategies'_\pl \subseteq \PureStrategies^\game_\pl$.

% \paragraph{Nash equilibria:}
With a light abuse of notation, we will overload the symbol $\utility$ to also denote the \textit{expected} utilities corresponding to mixed strategies.
A strategy $\mixedStrategy_1$ is said to be a \defword{best response} to a strategy $\mixedStrategy_2$ if
    $\mixedStrategy_1 \in \argmax_{\mixedStrategy_1' \in \MixedStrategies_1} \utility_1(\mixedStrategy_1', \mixedStrategy_2)$.
We use $\br(\mixedStrategy_2)$ to denote the (non-empty) set of all \textit{pure} best responses to $\mixedStrategy_2$.
Since the utility of the best-responding player is determined by the other player's strategy, we sometimes denote it as
    $\utility_1(\br(\mixedStrategy_2), \mixedStrategy_2)$.
    % (resp. $\utility_2(\mixedStrategy_1, \br)$).
(The analogous definitions apply when the roles of \Plone{} and \Pltwo{} are reversed.)
A \defword{Nash equilibrium}
    % (NE)
    is a strategy profile $\mixedStrategy = (\mixedStrategy_1, \mixedStrategy_2)$ under which each player's strategy is a best response to the strategy of the other player.
    $\NE(\game)$ denotes the set of all Nash equilibria in $\game$.

% \paragraph{Favourable best responses:}
A strategy $\pureStrategy_1$ is said to be an (opponent-)\textbf{favourable best response} to $\mixedStrategy_2$ if
    $\pureStrategy_1 \in \argmax_{\pureStrategyAlt_1 \in \br(\mixedStrategy_2)} \utility_2(\pureStrategyAlt_1, \mixedStrategy_2)$.
We use $\fbr(\mixedStrategy_2)$ to denote the (non-empty) set of all (pure) favourable best responses to $\mixedStrategy_2$.
When one player uses a favourable best response, the utilities of \textit{both} players are determined by the other player's strategy;
    this allows us to denote these utilities as
        $\utility_1(\fbr(\mixedStrategy_2), \mixedStrategy_2)$
        and
        $\utility_2(\fbr(\mixedStrategy_2), \mixedStrategy_2)$.

A \textbf{Stackelberg game} \cite{von1934marktform} is a setting where
    one player, the leader, commits to a mixed strategy
    to which the other player, the follower, best-responds.
\textit{In this paper, we assume that \Pltwo{} is the Stackelberg leader and \Plone{} is the follower,}
    which better fits the assumption that \Plone{} is the simulator.
Formally, a Stackelberg game $\meta{\game}$ corresponding to a base-game $\game$ works as follows.
First, the leader selects a mixed strategy $\mixedStrategy_2 \in \MixedStrategies^\game_2$.
Afterwards,
    the follower selects a favourable best response $\pureStrategy_1 \in \fbr^\game(\mixedStrategy_1)$,
    (i.e., \textit{breaking ties in the leader's favour}).
The players then receive payoffs
    $
        \utility^{\meta{\game}}(\pureStrategy_1, \mixedStrategy_2)
        :=
        \utility^{\game}(\pureStrategy_1, \mixedStrategy_2)
    $.
By \textbf{Stackelberg equilibrium} (\SE{}) of $\game$, we mean any \NE{} of the Stackelberg game $\meta{\game}$.

We also also consider ``pure Stackelberg games''
    where the leader is limited to committing to a \textit{pure} strategy $\pureStrategy_2 \in \PureStrategies^\game_2$.
However, to avoid the ambiguity of ``pure Stackelberg equilibrium'',
    we will refer to these games as \textbf{pure-commitment games}
    and to their NE as \textbf{pure-commitment equilibria}.

A strategy profile $\mixedStrategy$ is said to be a \textbf{strict Pareto improvement} over $\mixedStrategyAlt$
    if it satisfies
        $\utility_1(\mixedStrategy) > \utility_1(\mixedStrategyAlt)$
        and
        $\utility_2(\mixedStrategy) > \utility_2(\mixedStrategyAlt)$.
A two-player game $\game$ is said to be a \textbf{generalised trust game} \cite{kovarik2023game}
    if any pure-commitment equilibrium of $\game$ (with \Pltwo{} as the leader)
    is a strict Pareto improvement over any Nash equilibrium of $\game$.
(For a prototypical example of such a $\game$, see \Cref{fig:TG}.)

%% file: the_actual_new_stuff.tex
\section{Pure- vs. Mixed-Strategy Simulation}
    \label{sec:new_defs}\label{SEC:NEW_DEFS}

In this section, we formally define mixed-strategy simulation,
    contrast it with pure-strategy simulation,
    and survey the basic properties of the corresponding games.

\subsection{Definitions of Simulation Games}\label{sec:sub:sim_game_defs}

In \Cref{sec:sub:illustrative_example}, we informally described simulation games through a scenario
    in which Bob (\Pltwo{}) selects a robot that acts on his behalf
    and Alice (\Plone{}) has an option to pay a fixed cost
        to analyse the robot prior to interacting with it.
If Alice takes advantage of this option,
    she learns the (pure or mixed) strategy that the robot is going to employ
    and best-responds to it, breaking ties in Bob's favour.\arxivOnly{\footnotemark{}
        \footnotetext{
            \ The assumption that Alice breaks ties in Bob's favour is not necessarily realistic,
            as she might intentionally break ties to \textit{hurt} Bob
            to incentivise him selecting more favourable strategies.
            However, Bob generally could counter this by only adjusting his probabilities by some small $\epsilon$.
            Allowing more realistic tiebreaking strategies would thus
                significantly complicate all formal statements
                without necessarily allowing us to derive more meaningful results.
        }}
Otherwise, the game proceeds as usual.
We now give a formal counterpart to this description, the first part of which is a reformulation of that of \citet{kovarik2023game}.

\begin{definition}[Pure- and mixed-strategy simulation]\label{def:mixed-simgame}
    The \textbf{mixed-} and \textbf{pure-strategy simulation games} $\msimgame^\simcost$ and $\psimgame^\simcost$
        (or simply $\msimgame$ and $\psimgame$)
        corresponding to a two-player NFG $\game$ and \textbf{simulation cost} $\simcost > 0$
        are defined as the (infinite) NFGs given by:
        \begin{align*}
            \PureStrategies^{\msimgame}_1
                & :=
                \PureStrategies^\game_1 \cup \{ \mSim \},
            \ \ 
            \PureStrategies^{\msimgame}_2
                :=
                \MixedStrategies^\game_2
            ,
            \\
            \PureStrategies^{\psimgame}_1
                & :=
                \PureStrategies^\game_1 \cup \{ \pSim \},
            \ \ 
            \PureStrategies^{\psimgame}_2
                :=
                \MixedStrategies^\game_2
            ,
        \end{align*}
    where $\mSim$ and $\pSim$ are new strategies of \Plone{},
        called \textbf{mixed-} and \textbf{pure-strategy simulation}\arxivOnly{\footnotemark{}} in $\game$,
    defined by
        \begin{align*}
            \utility_1(\mSim, \mixedStrategy_2)
                & :=
                \utility^\game_1(\br^\game(\mixedStrategy_2), \mixedStrategy_2) - \simcost
            \\
            \utility_2(\mSim, \mixedStrategy_2)
            & :=
                \utility^\game_2(\fbr^\game(\mixedStrategy_2), \mixedStrategy_2)
                ,
                \\
                \text{resp. \phantom{blaahblaaahbla}} &
            \\
            \utility_1(\pSim, \mixedStrategy_2)
                & :=
                    \E_{\pureStrategy_2 \sim \mixedStrategy_2}
                        \utility_1(\br^\game(\pureStrategy_2), \pureStrategy_2)
                    -
                    \simcost
            \\
            \utility_2(\pSim, \mixedStrategy_2)
            & :=
                \E_{\pureStrategy_2 \sim \mixedStrategy_2}
                    \utility_2(\fbr^\game(\pureStrategy_2), \pureStrategy_2)
            .
        \end{align*}
        A \textbf{simulation equilibrium} is an \NE{} in which \Plone{} simulates with non-zero probability.
\end{definition}
    \footnotetext{
        The phrases pure-and mixed-strategy simulation might suggest that
            the object that the simulation is being applied to is a pure, resp. mixed strategy.
        In fact,
            the two types of simulation are applicable to the same objects,
            but they differ in the information they reveal.
        In other words, pure and mixed-strategy simulation
            differ in the type of \textit{output} they produce,
            not in the type of input they accept.
        % We are open to suggestions for terminology that avoid this issue while remaining succinct.
        }

To illustrate the distinction between $\pSim$ and $\mSim$,
    imagine that Alice and Bob play a game of rock-paper-scissors.
If Bob employs a single robot, \texttt{Uniform-bot},
       which uses an internal random number generator
        to play each of the three actions with probability $\frac{1}{3}$,
    applying mixed-strategy simulation $\mSim$ will only tell Alice that the robot's strategy
        is $(\frac{1}{3}, \frac{1}{3}, \frac{1}{3})$
        (which is entirely unhelpful).
In contrast, applying pure-strategy simulation $\pSim$ to \texttt{Uniform-bot} will predict the robot's exact action,
    allowing Alice to win the game every single time.
If Bob instead randomises between \texttt{Rock-bot}, \texttt{Paper-bot}, and \texttt{Scissors-bot},
    each of which can only use a single action,
    both $\pSim$ and $\mSim$ will reveal the robot's exact action.

\textit{When the distinction between pure- and mixed-strategy simulation does not matter,
    we will use the colloquial term \emph{simulation}.}
    % (and write $\Sim$ and $\simgame$).}

This example also shows that in a pure-strategy simulation game $\psimgame$,
    Bob cannot gain anything by randomising over multiple mixed strategies (since all utilities only depend on the \textit{overall} distribution over $\PureStrategies^\game_2$).
For the purposes of formal analysis of pure-strategy simulation games,
    this allows us to assume that Bob's space of pure strategies in $\psimgame$
    is limited to
    $
        \PureStrategies^{\psimgame}_2
        =
        \PureStrategies^\game_2
    $.
    \footnote{
        Note that the same simplification
            -- replacing $\PureStrategies^{\msimgame}_2 \coloneqq \MixedStrategies^\game_2$ by $\PureStrategies^\game_2$ --
        cannot be valid in $\msimgame$.
        This follows from the Trust Game example in \Cref{sec:sub:illustrative_example}
            (which illustrates that
            $\psimgame$ and $\msimgame$ can have different properties).
    }

\subsection{Randomising over Mixed Strategies}\label{sec:sub:meta_strategies}

Throughout the paper, and in particular in some of the proofs,
    it will be crucial to be able to treat ``mixtures over mixed strategies'' differently from a standard mixed strategies
    % it will sometimes be necessary to distinguish between
        % convex combinations of mixed strategies in $\game$
        % and
        % probability distributions over mixed strategies in $\game$
        (since the two respond differently to mixed-strategy simulation, as we saw earlier).
To address this issue,
    we will refer to probability distributions over $\MixedStrategies^\game_2$ as \textbf{meta-strategies}
    and denote them by symbols such as $\mixedMetaStrategy_2$
    (or $\pureMetaStrategy_2$ when the meta-strategy is pure, i.e., when it puts all probability mass on a single $\mixedStrategy_2 \in \MixedStrategies^\game_2$).
% In this context,
%     we will refer to strategies in $\game$ as strategies in the \textbf{base-game}.
We will use the hat symbol
    to indicate that a given mixed strategy is being used as a (pure) meta-strategy.
For example, 
    Bob's above-mentioned mixed meta-strategy of uniformly randomising between a \texttt{Rock-bot}, Paper-bot, and Scissors-bot could be formally written as
    $
        \mixedMetaStrategy_1
        :=
            \frac{1}{3} \meta{\Rock}
            \!+\!
            \frac{1}{3} \meta{\Paper}
            \!+\!
            \frac{1}{3} \meta{\Scissors}
    $,
    while the pure meta-strategy of always using the \texttt{Uniform-bot} would correspond to
    $
        \pureMetaStrategy_2
        :=
        \meta{
            \frac{1}{3} \Rock
            \!+\!
            \frac{1}{3} \Paper
            \!+\!
            \frac{1}{3} \Scissors
        }
    $.

For a mixed meta-strategy $\mixedMetaStrategy_2 \in \PureStrategies^{\msimgame}_2$
    and pure base-game strategy $\pureStrategy_2 \in \PureStrategies^\game_2$,
    we will use
    $
        \mesa{\mixedMetaStrategy}_2(\pureStrategy_2)
    $
    to denote the total probability that $\mixedMetaStrategy_2$ puts on $\pureStrategy_2$.
For example,
    if $\mixedMetaStrategy_2$ represents
        Bob using \texttt{Uniform-bot} with probability 30\%
        and \texttt{Rock-bot} $\meta{\Rock}$ with the remaining 70\% probability,
    we have
        $
            \mesa{\mixedMetaStrategy}_2(\Rock)
            =
                0.3 \cdot \frac{1}{3}
                +
                0.7 \cdot 1
            =
            0.8
        $.

\subsection{Basic Properties of Simulation Games}\label{sec:sub:basic_properties}

% Since pure-strategy simulation was studied by \citet{kovarik2023game},
While this paper focuses on the implications of mixed-strategy simulation $\mSim$, 
    pure-strategy simulation $\pSim$ will be relevant for two reasons.
First, it serves as an important baseline for comparison.
Second, it can be a useful source of intuitions for the properties of mixed-strategy simulation
    ---
    this is because $\mSim$ in $\game$
    can be also be understood as $\pSim$ in the infinite game
        $
            % \meta{\game}
            % :=
            (\PureStrategies^\game_1, \MixedStrategies^\game_2, \utility)
        $.\footnotemark{}
    \footnotetext{
        In light of the equivalence between $\msimgame$ and $(\PureStrategies^\game_1, \MixedStrategies^\game_2, \utility)_\simSubscriptPure$, we might hope to answer questions about mixed-strategy simulation
            by applying existing pure-strategy simulation theory to $(\PureStrategies^\game_1, \MixedStrategies^\game_2, \utility)$.
        However, this strategy turns out to be inapplicable because the prior work requires the base game $\game$ to be finite.
    }
% Moreover, as we saw in \Cref{sec:sub:illustrative_example},
%     $\msimgame$ and $\psimgame$ can behave differently,
% implying that
%     transforming $\game$ as above can change some of its simulation-relevant properties.

The following lemma shows that
    -- unlike in pure-strategy simulation games --
    any NE of $\game$ is also an NE of the simulation game $\msimgame$.
This is because if \Pltwo{} puts all probability mass on a single mixed strategy $\mixedStrategy_2$,
    costly simulation results in lower utility for \Plone{} than directly best-responding to $\mixedStrategy_2$.

\begin{restatable}{lemma}{lemmaNEofGAreNEofMsim}\label{lem:NE_of_G_are_NE_of_msim}
    Identifying
        $
            \mixedStrategy \in \MixedStrategies^\game
        $
        with
        $
            (\mixedStrategy_1, \meta{\mixedStrategy}_2) \in \MixedStrategies^{\msimgame}
        $,
    we have $\NE(\game) \subseteq \NE(\msimgame)$ for any $\game$.
\end{restatable}

Our primary interest in simulation is to use it as a tool for improving the outcomes for the players.
What, however, ought to be our metric of success,
    particularly in light of \Cref{lem:NE_of_G_are_NE_of_msim}?
When discussing the impacts of simulation on social welfare,
    we will primarily be concerned with whether
    enabling simulation \textbf{introduces Pareto-improving Nash equilibria}.
    Formally, this means that
        \emph{there is some $c_0 > 0$
        s.t. for every $\simcost < c_0$,
        the game
            % $\msimgame^\simcost$ (resp. $\psimgame^\simcost$)
            $\msimgame^\simcost$
        has a Nash equilibrium $\metaNE$
            % with $\metaNE_1(\mSim) > 0$
        in which $\utility(\metaNE)$ is \emph{strictly} higher,
            for both players,
        than the utility achievable in any NE of $\game$}
        (and similarly for $\pSim$).

\begin{remark}\label{rem:IPIN_doesnt_mean_there_arent_new_bad_sim_eq}
    Note that the fact that enabling simulation introduces Pareto-improving NE
        does not necessarily preclude simulation from also introducing new NE
        whose utility is \emph{lower} than some, or even all, NE of the original game.
    While it is worth analysing when this occurs, such equilibrium selection problems are largely beyond the scope of the present work.
\end{remark}

Strictly speaking,
    $\msimgame$ is defined as a game where \Pltwo{} has infinitely many pure strategies,
    which could greatly complicate its analysis.
Fortunately, the following result shows that it is always enough to consider a finite number of strategies for \Pltwo{}.

\begin{restatable}[Reduction to a finite strategy space]{proposition}{propFiniteStrategySpace}\label{prop:finite_str_space}
For any finite $\game$,
there exists a finite subgame $\msimgame'$ of $\msimgame$ s.t.:
    \begin{align*}
        & \text{ (i) } 
            \NE(\msimgame') \subseteq \NE(\msimgame)
            ,
            \\
        & \text{(ii) }
            \forall \mixedMetaStrategy \in \NE(\msimgame)
            \,
            \exists \mixedMetaStrategy' \!\! \in \NE(\msimgame') 
            :
            \utility(\mixedMetaStrategy') = \utility(\mixedMetaStrategy)
            .
    \end{align*}
\end{restatable}
    
% \arxivOnly{
\begin{proof}[Proof sketch]
    $\MixedStrategies^\game_2$ can be expressed as a union of (non-closed) polytopes
        % \begin{align*}
        $
            \inverseFBR(\pureStrategy_1)
            :=
        $
        $
            \left\{
                \mixedStrategy_2 \in \MixedStrategies_2^\game
                \mid
                \fbr(\mixedStrategy_2) \ni \pureStrategy_1
            \right\}
        $,
        % for
        $
            \pureStrategy_1 \in \PureStrategies_1^\game
        $.
            % .
        % \end{align*}
    We then have $\utility^{\msimgame}_2(\mSim, \mixedStrategy_2) = \utility^\game_2(\pureStrategy_1, \mixedStrategy_2)$
        for any strategy $\mixedStrategy_2 \in \inverseFBR(\pureStrategy_1)$.
    \Pltwo{} can then recover all relevant strategies
        by mixing over the vertices of the closure $\closure{\inverseFBR(\pureStrategy_1)}$.
\end{proof}
% }

\section{Computational Results}
    \label{sec:computational}\label{SEC:COMPUTATIONAL}

In this section, we investigate the difficulty of analysing mixed-strategy simulation games.
From \Cref{prop:finite_str_space}, it follows that even though $\msimgame$ is defined as an infinite game,
it can be solved in finite time. 
By ``solving'' a game, we mean any of: 
    (a) finding one \NE{};
    (b) finding an \NE{} that maximises social welfare or the utility of one of the players;
    or (c) finding all \NE{} payoff profiles and some \NE{} corresponding to each.

\begin{restatable}[Upper bound on solving $\msimgame$]{proposition}{solvingUpperBound}\label{prop:solve_upper_bound}
    For any $\game$, solving $\msimgame$ is
        at most as difficult as solving a game $\meta{\game}$ with
            $| \PureStrategies^{\meta{\game}} | = 
                \bigO(|\PureStrategies^\game_1|^2
                \cdot
                2^{|\PureStrategies^\game_2|})$.
\end{restatable}

\begin{proof}[Proof sketch]
    The non-trivial part is the size of $\PureStrategies^{\meta{\game}}_2$.
    \Cref{prop:finite_str_space} shows that a suitable $\PureStrategies^{\meta{\game}}_2$ can be obtained by
        splitting the $(\vert \PureStrategies^\game_2 \vert \shortminus 1)$-dimensional simplex $\MixedStrategies^\game_2$
        into $|\PureStrategies^\game_1|$ convex polytopes
    and only considering the vertices of these polytopes.
    Estimating the number of these vertices yields the result.
\end{proof}

The fact that $\NE(\game) \subseteq \NE(\msimgame)$
    trivially implies that finding \textit{all} \NE{} of $\msimgame$
    is at least as difficult as finding all NE of $\game$.
The difficulty of finding simulation equilibria of $\msimgame$ depends on $\simcost$.
    When $\simcost$ is prohibitively high,
        solving $\msimgame$ is equivalent to solving $\game$
        (since Alice never simulates).
    For general $\simcost$, we leave determining the exact complexity of finding simulation equilibria of $\msimgame$ as an open problem.

From the perspective of a designer, arguably the most important question is
    whether enabling simulation is likely to lead to more socially beneficial outcomes in a given game.
    The following result shows that this is, in general, hard to determine.

\begin{restatable}[Determining whether simulation helps is hard]{theorem}{helpsIsNPHtheorem}\label{thm:helps_is_NPH_v3}
    Denote by $P_\textnormal{a}$, \dots, $P_\textnormal{e}$ the problems of determining whether enabling $\mSim$ introduces an \NE{} which is strictly better than all \NE{} of $\game$ in terms of
        (a)~both players' utilities, %~Pareto improvement,
        (b)~\Plone{}'s utility,
            (c)~\Pltwo{}'s utility,
        (d)~any strictly monotonic social welfare function (such as $\utility_1 + \utility_2$ or $\utility_1 \cdot \utility_2$), or
        (e)~the egalitarian social welfare function $\min\{ u_1, u_2\}$.
        
    For general games $\game$, each of the problems $P_\textnormal{a}$, \dots, $P_\textnormal{e}$ is \NPH{}.
\end{restatable}

\begin{proof}[Proof sketch]
    The proof has two main ingredients.
    First, we create a game where the only \NE{} payoffs are $(0, 0)$,
        but enabling simulation introduces a simulation equilibrium with payoffs $(1-\simcost, 1)$.
        (This is easily achieved in a scenario where the players need to coordinate between two identical trust games; cf. \Cref{fig:hardness_game}.)
    Second, we use a pre-existing method \cite{sauerberg2024computing} for constructing a class $\mc C$ of games such that
        (a) any $\game \in \mc C$ has a \NE{} with payoffs $(0, 0)$;
        (b) a game $\game \in \mc C$ may or may not have a \NE{} with payoffs $(1, 1)$, but definitely has no other \NE{};
        (c) determining whether $\game \in \mc C$ does or does not have the \NE{} with payoffs $(1, 1)$ is \NPH{}.
    We then show that enabling $\mSim$ does not affect the \NE{} of $\game \in \mc C$.
    
    By putting these two ingredients together, we obtain a class of games where
        enabling $\mSim$ is guaranteed to yield a good outcome,
        but such outcome may or may not have been possible even without simulation
        -- and determining whether this is the case or not is equivalent to solving a problem that is known to be \NPH{}.
\end{proof}

For the purpose of this paper,
    \Cref{thm:helps_is_NPH_v3} suggests that we should not expect to be able to find a concise description of the effects of enabling $\mSim$ in general games.
    We will, therefore, instead focus on identifying particular classes of games where simulation has predictable effects.

\section{Effects of Simulation on Players' Welfare}\label{sec:welfare}

In this section, we describe specific classes of games where enabling mixed-strategy simulation does, and does not, lead to socially beneficial outcomes.

\subsection{Drawbacks of an Overly Informed Co-Player}\label{sec:sub:negative}

In \Cref{sec:sub:illustrative_example}, we saw that in the simple case of a $2 \times 2$ trust game\footnotemark{}, enabling $\pSim$ \iPiN{}{} but enabling $\mSim$ does not.
The following theorem is a generalisation of this negative result.
    \footnotetext{
        Strictly speaking, \Cref{sec:sub:illustrative_example} describes the trust game as a simultaneous-move game.
        However, note that this game is strategically equivalent to the game where
            \Plone{} acts first
                (deciding between \Trust{} and \WalkOut{}),
            after which \Pltwo{} observes \Plone{}'s action
            and chooses whether to \Cooperate{} or \Defect{}.
        In other words, this $2 \times 2$ trust game is a normal-form representation of a game which satisfies the assumptions of \Cref{thm:negative}.
    }

\begin{restatable}[Simulating a perfectly informed player]{theorem}{thmNegativeResult}
    \label{thm:negative}\label{THM:NEGATIVE}
    Let $\game_0$ be a finite two-player game.
    Denote by $\game$ the game where:
    \begin{enumerate}[label=(\roman*)]
        \item First, \Plone{} selects $\pureStrategy_1 \in \PureStrategies^{\game_0}_1$ and \Pltwo{} observes \Plone{}'s choice.
        \item Next, \Pltwo{} selects a pure strategy $\pureStrategy_2 \in \PureStrategies^{\game_0}_2$.
            We assume that \Pltwo{} must select a Pareto-optimal response
                (but they are not required to best-respond).\footnotemark{}
                \footnotetext{
                    Formally, we require that
                    if \Pltwo{} selects $\pureStrategy_2$ in response to $\pureStrategy_1$, then any $\pureStrategy'_2 \in \PureStrategies^{\game_0}_2$ must have
                    either
                    $
                        \utility_1(\pureStrategy_1, \pureStrategy'_2)
                        \leq
                        \utility_1(\pureStrategy_1, \pureStrategy_2)
                    $
                    or
                    $
                        \utility_2(\pureStrategy_1, \pureStrategy'_2)
                        \leq
                        \utility_2(\pureStrategy_1, \pureStrategy_2)
                    $.
                }
        \item The players receive utilities $\utility_\pl^{\game_0} (\pureStrategy_1, \pureStrategy_2)$.
    \end{enumerate}

    \noindent
    Then enabling $\mSim$ does not \iPiNwithoutS{} in $\game$.
\end{restatable}

% \arxivOnly{
\begin{proof}[Proof sketch]
    When \Pltwo{} can select $\pureStrategy_2$ as a function of \Plone{}'s choice of $\pureStrategy_1$,
        they can increase the relative attractiveness of any fixed $\pureStrategy_1^* \in \PureStrategies_1$
            by being maximally aggressive against any $\pureStrategy_1 \neq \pureStrategy_1^*$.
        Crucially, \Pltwo{} can do this without lowering \Plone{}'s utility of $\pureStrategy^*_1$.
    Moreover, \Pltwo{} can then bring \Plone{}'s utility for $\pureStrategy^*_1$ all the way to their
        maxmin value
        -- and any NE of $(\game_0)_\simSubscriptMixed$ will require \Pltwo{} to do so.
    However, once \Plone{} only gets their maxmin value, they have no reason to simulate,
        destroying any potential for simulation-based cooperation.
\end{proof}
% }

\subsection{Partial Trust}
    \label{sec:sub:positive_partial_trust}\label{SEC:SUB:POSITIVE_PARTIAL_TRUST}

The following definition captures settings where
    Alice can vary the degree to which she trusts Bob,
    with more trust enabling better outcomes for both,
    but also making Alice more vulnerable to exploitation
    (for illustration, see \Cref{fig:gPTG}).
The purpose of this section is to show that settings where such modulation of trust is possible
    can benefit from mixed-strategy simulation.

\begin{figure}[tb]
    \centering
    \begin{NiceTabular}{rccc}[cell-space-limits=3pt]
                                & $\Cooperate$     & $\Defect$ & $\ \ \ $\\
        $\T_1$ & \Block[hvlines]{4-2}{}
                                $20, 20$    & $-100, 100$           & $\ \ \ \ \ \ \ $ \\
        $\T_2$ &       $10, 10$    & $\ \, -20, 20 \ $     & $\ \ \ \ \ \ \ $ \\
        $\T_3$ &            $5, 3$      & $\ \ \ -1, 6 \ $       & $\ \ \ \ \ \ \ $ \\
        $\WO$ &            $0, 0$      & $\ \ \ 0, 0 \ $       & $\ \ \ \ \ \ \ $ \\
        \\
        $\T'_1$ & \Block[hvlines]{3-2}{}
                                $20, 20$    & $-100, 100$           & $\ \ \ \ \ \ \ $ \\
        $\T'_2$ &       $10+1, 10 \shortminus 1$    & $\ \, -30, 30 \ $     & $\ \ \ \ \ \ \ $ \\
        $\T'_{1.5}$ &       $\frac{10+20}{2}, \frac{10+20}{2}+1$    & $\ \, \frac{(\shortminus 20) + (\shortminus 100)}{2}, \frac{20+100}{2} \ $     & $\ \ \ \ \ \ \ $ \\
        \\
        $\T_{1.9}$ & \Block[hvlines]{1-2}{}
                                $11, 11$    & $-99, 99$           & $\ \ \ \ \ \ \ $
    \end{NiceTabular}
    \caption{\textit{Top:} An illustration of a generalised partial-trust game $\game$ from \Cref{def:gPTG}.
        \textit{Middle:} Examples of strategies that would invalidate the technical conditions in the definition if we added them to $\game$.
        $\T'_1$ fails (4a), since it does not have a unique value $\utility_1(\T, \C)$.
        $\T'_2$ fails the requirement (4b), that any increase in $\utility_1(\T, \C)$
            -- here caused by going from $\T_2$ to $\T'_2$ --
            must also increase $\utility_2(\T, \C)$
            (and $\utility_2(\T, \D)$, and decrease $\utility_1 (\T, \D)$).
        $\T'_{1.5}$ fails (5),
            since $\T'_{1.5}$ yields the same $\utility_1$ as the convex combination $\frac{1}{2} \cdot \T_1 + \frac{1}{2} \cdot \T_2$
            without also having the same $\utility_2$.
        \textit{Bottom:} $\T_{1.9}$ is an example of a strategy 
        that
            would not invalidate any of the technical conditions from the definition,
            but adding it to $\game$ would break the assumption of ``sufficiently high $\utility_2(\FT, \C)$'' that is required for \Cref{thm:generalised_PTG}.
    }
    \Description{Top: An illustration of a generalised partial-trust game.
    Middle: An examples of actions that would not be allowed by a definition of a generalised PTG.
    Bottom: An example of an action that would be allowed, but would cause the next theorem to fail.}
    \label{fig:gPTG}
\end{figure}

\begin{restatable}[Generalised Partial-Trust Game]{definition}{defGPTG}\label{def:gPTG}
    By a \textbf{generalised partial-trust game} (PTG), we mean any
        $
            \game
            =
            (\PureStrategies_1, \PureStrategies_2, \utility)
        $
    that satisfies the conditions
    \begin{enumerate}[label={(\arabic*)}]
        \item \textbf{\Pltwo{} has two strategies:}
            \Pltwo{} only has only two pure strategies,
            which we label $\Cooperate$ ($\C$) and $\Defect$ ($\D$);
        \item \textbf{\Plone{} has a dedicated strategy for opting out of the game:}
            \Plone{} has a strategy, which we label $\WalkOut$ ($\WO$), for which
            $
                \utility(\WO, \C)
                =
                \utility(\WO, \D)
                =
                (0, 0)
            $;
        \item \textbf{Trust enables profits but is exploitable:}\\
            Any \Plone{}'s strategy $\T \neq \WO$ (``trust'') satisfies
            \begin{align*}
                \utility_1(\T, \C)
                & \ > \ 
                \utility_1(\WO, \symbolPlaceholder) = 0
                \ > \ 
                \utility_1(\T, \D)
                \\
                \utility_2(\T, \D)
                & \ > \ 
                \utility_2(\T, \C)
                \ > \ 
                \utility_2(\WO, \symbolPlaceholder) = 0
                ;
            \end{align*}
            % \begin{align*}
            %     \utility_1(\T, \C)
            %     & \ > \ 
            %     \ \ \ \ \ \ 0 \ \ \ \ \ \ 
            %     \ > \ 
            %     \utility_1(\T, \D)
            %     \\
            %     \utility_2(\T, \D)
            %     & \ > \ 
            %     \utility_2(\T, \C)
            %     \ > \ 
            %     0
            %     .
            % \end{align*}
    \end{enumerate}
    and the technical assumptions
    \begin{enumerate}[label={(\arabic*)}]
        \setcounter{enumi}{3}
            \item \textbf{There is a straightforward hierarchy of trust:}\\
            (a) For any two strategies $\T \neq \T'$, we have $\utility_1(\T, \C) \neq \utility_1(\T', \C)$.\\
            (b) When
                $
                    \utility_1(\T, \C) > \utility_1(\T', \C)
                $,
            we also have
                $
                    \utility_2(\T, \C) > \utility_2(\T', \C)
                $,
                $
                    \utility_1(\T, \D) < \utility_1(\T', \D)
                $,
                $
                    \utility_2(\T, \D) > \utility_2(\T', \D)
                $;
        \item \textbf{\Plone{} cannot use convex combinations for tie-breaking:}\\
            For any $\T$,
            if a convex combination $\mixedStrategy_1 = \lambda \pureStrategy_1 + (1 \shortminus \lambda) \pureStrategyAlt_1$
            satisfies
                $\utility_1(\T, \mixedStrategy_2) = \utility_1(\mixedStrategy_1, \mixedStrategy_2)$
                for all $\mixedStrategy_2$,
            it must also satisfy
                $\utility_2(\T, \mixedStrategy_2) = \utility_2(\mixedStrategy_1, \mixedStrategy_2)$
                for all $\mixedStrategy_2$.
            \label{ass:gPTG_definition_no_tiebreaking}
    \end{enumerate}
\end{restatable}

\noindent
To give an intuition for the conditions used in \Cref{def:gPTG}, note that 
    (3) ensures that non-zero payoffs can only be achieved when \Plone{} $\Trust$s \Pltwo{},
        but \Pltwo{} is always tempted to $\Defect$,
        which makes \Plone{} strictly worse off than if they $\WalkOut$.
    The technical conditions (4a) and (5) ensure that
        once \Plone{} decides on the tradeoff between potential gains from cooperation and exploitability,
        they have no room left for varying \Pltwo{}'s payoffs.
    The technical condition (4b) ensures that
        higher cooperative gains for \Plone{} go hand in hand with higher cooperative gains for \Pltwo{}
        (but also increase \Plone{}'s exploitability and \Pltwo{}'s gains from defection).

The concept of a game with a gradation of trust can be extended in many ways,
    such as not having the default outcome be zero,
    not requiring that a higher degree of trust means that $\utility_2(\T, \D)$ is higher,
    giving \Pltwo{} a hierarchy of cooperative and defective strategies, etc.
However, to simplify the exposition,
    this paper will only consider the basic setup described in \Cref{def:gPTG}.
The following lemma
    summarises the basic properties of generalised partial-trust games.   

\begin{restatable}{lemma}{lemmaBasicPropertiesOfGPTG}\label{lem:gPTG_properties}
    Let $\game$ be
        a generalised partial-trust game. Then:
        % as in \Cref{def:gPTG}.
    \begin{enumerate}[label=(\roman*)]
        \item For any $\mixedStrategy \in \NE(\game)$, $\mixedStrategy_1(\WO) = 1$;
        \item The unique pure-commitment equilibrium of $\game$ is $(\FT, \C)$, where
            $\{ \FT \}  = \argmax \left\{ \utility_1(\T, \C) \mid \T \in \PureStrategies^\game_1 \right\}$.\\
            In particular, $\game$ is a generalised trust game.
    \end{enumerate}
    If $\utility_2(\FT, \C)$ is sufficiently high relative to other payoffs, then:
    \begin{enumerate}[label=(\roman*)]
        \setcounter{enumi}{2}
        \item The unique SE of $\game$ has the form $(\FT, \incentivise{\FT})$,
            where
                \begin{align*}
                \incentivise{\FT} & = \dFT \cdot \D + (1-\dFT) \cdot \C
                ,
                \\
                \dFT
                    & =
                    \max \, \{
                        \delta \in [0, 1]
                    \mid
                        \FT \in \br(\delta \D + (1-\delta) \C)
                    \}
                ,
                \end{align*}
                is the ``optimal commitment that still incentivises $\FT$'';
        \item The SE of $\game$ is a strict Pareto-improvement over NE of $\game$
            if and only if
            there is $\T \in \PureStrategies^\game_1$ s.t.
                $
                    \frac{
                        \utility_1(\T, \C)
                    }{
                        \shortminus \utility_1(\T, \D)
                    }
                    >
                    \frac{
                        \utility_1(\FT, \C)
                    }{
                        \shortminus \utility_1(\FT, \D)
                    }
                $.
    \end{enumerate}
\end{restatable}

\begin{proof}[Proof sketch]
The difficult part of \Cref{lem:gPTG_properties} is (iv),
which relies on the fact that
\Plone{} can trivially scale their level of trust
    by interpolating between any two actions,
    including $\FT$ and $\WalkOut$.
This lets \Plone{} disregard any $\T$
    that has a worse risk-benefit ratio than $\FT$,
making (iv) equivalent to the claim that:
    ``giving \Pltwo{} the ability to make mixed commitments results in a Pareto-improvement
    if and only if
    disregarding these redundant actions leaves \Plone{} with more options than just $\FT$ and $\WO$.''
    This follows from (iii).
\end{proof}

\noindent
In light of \Cref{lem:gPTG_properties}, a generalised PTG is said to be \textbf{non-trivial} when it satisfies the condition (iv).
We can now prove a generalisation of the positive result from \Cref{sec:sub:illustrative_example}.

\begin{restatable}[Simulation helps with partial trust]{theorem}{thmGenPTG}\label{thm:generalised_PTG}
    Let $\game$ be a non-trivial generalised partial-trust game.
    If $\utility_2(\FT, \C)$ is sufficiently high relative to other payoffs in $\game$,
        enabling $\mSim$ \iPiN{} in $\game$.
\end{restatable}

% \arxivOnly{
\begin{proof}[Proof sketch for \Cref{thm:generalised_PTG}]
    The key insight is that when \Pltwo{} plays the Stackelberg equilibrium $\mixedStrategy^\SE_2$ of $\game$,
        \Plone{} will be indifferent between $\FT$ and some other strategy $\PT$,
    and the non-triviality condition ensures that $\PT \neq \WO$.
    The proof then consists of showing that there is an equilibrium where
        \Pltwo{} mixes between $\mixedStrategy^\SE_2$ of $\game$ and defecting with probability 100\%
        and
        \Plone{} mixes between $\PT$ and $\mSim$.
    The assumption on $\utility_2(\FT, \C)$ ensures that \Pltwo{} cannot improve their utility by switching to some intermediate level of defection.
\end{proof}
% }

% \begin{restatable}{corollary}{corollaryPTG}\label{cor:PTG}
%     (i) Mixed-strategy simulation \iPiN{} in \PTG{} (\Cref{fig:PTG}).\\
% % \end{restatable}
% % 
% % \begin{restatable}{corollary}{mixedSimCanHelp}\label{cor:mixedSimCanHelp}
%     (ii)
%         Pure-strategy simulation \iPiN{}{} in \emph{all} generalised trust games
%         while mixed-strategy simulation only does so in \emph{some}.
% \end{restatable}

\subsection{Trust and Coordination}
    \label{sec:sub:positive_coordination}\label{SEC:SUB:POSITIVE_COORDINATION}

We now investigate simulation in coordination games.

\begin{definition}[Generalised coordination game]\label{def:generalised_coordination_game}
    By a \textbf{generalised coordination game},
    we will mean a finite two-player $\game$ game where:
    \begin{itemize}
        \item $\PureStrategies_\pl = \left\{ \action_\pl^1, \dots, \action_\pl^\nOfActions \right\}$,
            % $\PureStrategies_2 = \left\{ \actionOpp^1, \dots, \actionOpp^\nOfActions \right\}$
            for some $\nOfActions \geq 2$;
        \item
            $
                \utility_1(\actionPl^k, \actionOpp^l)
                =
                \utility_2(\actionPl^k, \actionOpp^l)
                =
                0
            $
            for $k \neq l$;
            and
        \item
            $
                \utility_1(\actionPl^k, \actionOpp^k), \ 
                \utility_2(\actionPl^k, \actionOpp^k)
                > 0
            $
            for any $k$.
    \end{itemize}
\end{definition}

% COMMENTING OUT THE MORE GENERAL FORMULATION IN FAVOUR OF THE ONE WHERE MISCOORDINATION GETS PAYOFF 0
% \begin{definition}[Generalised coordination game]\label{def:generalised_coordination_game}
%     Let $\game$ be a two-player normal-form game where $\PureStrategies_1 = \left\{ \actionPl^1, \dots, \actionPl^\nOfActions \right\}$ and $\PureStrategies_2 = \left\{ \actionOpp^1, \dots, \actionOpp^\nOfActions \right\}$ for some $\nOfActions \in \N$.
%     If
%         for each player~$\pl$
%         and each pair of actions $(\actionPl^k, \actionOpp^l)$, $k \neq l$,
%         we have
%             $\
%                 \utility_\pl(\actionPl^k, \actionOpp^l)
%                 <
%                 \min \left\{
%                     \utility_\pl(\actionPl^1, \actionOpp^1), \dots, \utility_\pl(\actionPl^\nOfActions, \actionOpp^\nOfActions)
%                 \right\}
%             $,
%         we say that $\game$ is a \textbf{generalised coordination game}. 
%         % \lh{Is it worth extending this to the $n \times m$ case or does that break things?}
%         % \vk{It's more that that wouldn't add anything interesting.}
% \end{definition}

\noindent
As a standard property of coordination games, we get that:

\begin{restatable}{lemma}{lemmaCoordinationGameProperties}\label{lem:coordination_game_properties}
    For any generalised coordination game,
        $
            \NE(\game)
            =
            \left\{
                \mixedStrategy^{\actionSubset}
                \mid
                \actionSubset \subseteq \left\{1, \dots, \nOfActions \right\}
            \right\}
        $
    for some $\mixedStrategy^{\actionSubset}$ which satisfy:
    (i)~$
            \supp (\mixedStrategy^{\actionSubset}_\pl)
            =
            \{ \action_\pl^k \mid k \in \actionSubset \}
        $.
    (ii)~NE that mix over fewer actions yield higher payoffs.
    (That is,
        $\mixedStrategy^{\actionSubset'}$ is a strict Pareto improvement over $\mixedStrategy^{\actionSubset}$
        whenever $\actionSubset' \subsetneq \actionSubset$.)
\end{restatable}

Recall that by \Cref{lem:NE_of_G_are_NE_of_msim}, any \NE{} of the original game also exists as an \NE{} of the simulation game
 -- in particular, enabling $\mSim$ cannot prevent the existence of the (undesirable) \NE{} where Bob only uses a single \textit{mixed} ``robot''.
However, enabling $\mSim$ does have the potential to prevent miscoordination when Bob randomises over multiple robots that use incompatible strategies
    (e.g., when $\mixedMetaStrategy_2 = \frac{1}{2} \meta{\action}_2^1 + \frac{1}{2} \meta{\action}_2^2$).
In addition to this fact, the following result shows that
    mixed-strategy simulation also introduces simulation equilibria
    that are better than miscoordination,
    but not as good as successful coordination at the players' favourite outcome.

\begin{restatable}[Simulation in coordination games]{proposition}{propCoordinationGames}\label{prop:coordination_games}
    Let $\game$ be a generalised coordination game
        and denote by $\mixedStrategy^{\{1, \dots, \nOfActions\}}$ its fully mixed NE.
    Then, for sufficiently low $\simcost$, we have:
    \begin{enumerate}[label=(\roman*)]
        \item $\msimgame$ has some simulation equilibrium $\metaNE$;
        \item Any simulation equilibrium $\metaNE \in \NE(\msimgame)$ satisfies
            \begin{align*}
                \utility_1(\mixedStrategy^{\{1, \dots, \nOfActions\}})
                    & <
                    \utility_1(\metaNE)
                    <
                    \max\nolimits_{k} \, \utility_1(\actionPl^k, \actionOpp^k)
                \\
                \utility_2(\mixedStrategy^{\{1, \dots, \nOfActions\}})
                    & <
                    \utility_2(\metaNE)
                    \leq
                    \max\nolimits_{k} \, \utility_2(\actionPl^k, \actionOpp^k)
                ;
            \end{align*}
        \item Unless $\game$ has multiple optimal pure commitments for \Pltwo{},
            any such $\metaNE$ satisfies
            $\utility_2(\metaNE) < \max_{k} \utility_2(\actionPl^k, \actionOpp^k)$.
    \end{enumerate}
\end{restatable}

\Cref{prop:coordination_games} shows that mixed-strategy simulation is able to prevent the worst equilibria, but does not \iPiNwithoutS{} in the stronger sense of allowing for an outcome that wouldn't be achievable through other means (i.e., by successful selection of a pure equilibrium).
The following example and theorem show that the usefulness of mixed-strategy simulation increases
    when the players need to deal not only with coordination but also with issues of trust.

\begin{figure}
    \centering
    \begin{NiceTabular}{rccc}[cell-space-limits=3pt]
        & $\actionOpp^1$    & $\actionOpp^2$    & $\OO$ \\
        $\actionPl^1$ & \Block[hvlines]{3-3}{}
            \begin{tabular}{@{}c@{}} $20, 20\phantom{\shortminus}$ \ \ $\shortminus99, 40$ \\ $\phantom{\shortminus9}9, \shortminus99$ \ \ $\phantom{\shortminus9}9, \shortminus99$ \end{tabular}
            & $0, 0$
            & $0, 1$ \\
        $\actionPl^2$
            & $0, 0$
            & \begin{tabular}{@{}c@{}} $20, 20\phantom{\shortminus}$ \ \ $\shortminus99, 40$ \\ $\phantom{\shortminus}10, \shortminus99$ \ \ $\phantom{\shortminus}10, \shortminus99$ \end{tabular}
            & $0, 1$ \\
        $\OO$
            & $1, 0$
            & $1, 0$
            & $1, 1$
    \end{NiceTabular}
    \caption{Trust-and-coordination game,
        where coordinating on a joint action $(\actionPl^k, \actionOpp^k)$
        leads the players to a trust subgame.
    }
    \Description{A trust-and-coordination game.}
    \label{fig:DTG}
\end{figure}

\begin{definition}[Trust-and-coordination game]\label{ex:generalised_trust_and_coordination}
    By a \textbf{trust-and-coordination game}, we mean a game $\game$ which works as follows
        (for examples, see \Cref{fig:DTG} and \ref{fig:gDTG}).
    \begin{itemize}
        \item In the first stage, the players simultaneously select an action from the set
            $\{ \action_\pl^1, \dots, \action_\pl^n, \OO \}$.
        \item If the players select $(\actionPl^k, \actionPl^l)$ for $k \neq l$,
            they receive ``\textbf{b}ad'' miscoordination payoffs $(\Bad_1, \Bad_2)$.
        \item \textbf{O}pting \textbf{O}ut of the game via $\OO$ yields $(\Bad_1, \Bad_2)$,
            with an additional reward $\epsilon$ for the player(s) who used $\OO$.
        \item If the players coordinate on some $(\actionPl^k, \actionOpp^k)$,
            they enter the second stage of the game,
            where they play a subgame $\game_k$.
        \item Each $\game_k$ is a $2 \! \times \! 2$ trust game with actions
                $\{ \Trust, \WalkOut \}$, resp. $\{ \Cooperate, \Defect \}$.
        \item We denote the payoffs in $\game_k$ as
            \begin{align*}
                \utility^{\game_k}(\T, \C) & := (\Good^k_1, \Good^k_2),
                \ \ \ 
                \utility^{\game_k}(\T, \D) := (\Horrible^k_1, \Awesome^k_2), \\
                \utility^{\game_k}(\WO, \C) & = \utility(\WO, \D) := (\Neutral^k_1, \Horrible^k_2).
            \end{align*}
            (The naming is meant to be suggestive of
                \textbf{a}wesome,
                \textbf{g}ood,
                \textbf{n}eutral,
                and \textbf{h}orrible
                .)
            We assume that:
                \begin{align*}
                    \Horrible^k_1 < \Bad_1 < \Bad_1 + \epsilon < \Neutral^k_1 < \Good^k_1
                    \\
                    \Horrible^k_2 < \Bad_2 < \Bad_2 + \epsilon < \Good^k_2 < \Awesome^k_2
                    .
                \end{align*}
    \end{itemize}
\end{definition}

The only NE of $\game$ is
    for both players to $\OptOut$, yielding utilities $\Bad_\pl + \epsilon$.
In contrast, the SE of $\game$ (with \Pltwo{} as the leader) would consist of
    coordinating on one of the subgames $\game_k$
    and then playing its SE,
    yielding utilities $\utility_1 = \Neutral^k_1$, $\utility_2 > \Good^k_2$.
We will use $\SEvalue{k}{\pl}$ to denote
    the expected utility corresponding to the SE (with \Pltwo{} as the leader)
    of the subgame $\game_k$.

\begin{restatable}[Simulation helps in trust-and-coordination games]{theorem}{thmGenDTG}\label{thm:generalised_DTG}
    Let $\game$ be a trust-and-coordination game.
    If
    % \vspace{-0.35em}
        \begin{enumerate}[label=(\alph*)]
            \item
                $
                    \argmax \SEvalue{k}{2}
                    \cap
                    \argmax \SEvalue{k}{1}
                    =
                    \emptyset
                $; and
            \item The NE payoffs $\Horrible^k_2$ of \Pltwo{} in the subgames $\game_k$
        are sufficiently low relative to the other payoffs in $\game$; 
        \end{enumerate}
    % \vspace{-0.35em}
    then enabling $\mSim$ \iPiN{}.
\end{restatable}

% \arxivOnly{
\begin{proof}[Proof sketch]
    The proof consists of showing that $\msimgame$ admits an \NE{} where
        \Pltwo{} mostly plays the Stackelberg equilibrium of ``\Plone{}'s favourite subgame'' $\game_{k_1}$,
            but sometimes deviates to playing the SE ``\Pltwo{}'s favourite subgame'' $\game_{k_2}$,
        while
        \Plone{} mixes between their part of the SE of $\game_{k_1}$ and simulating.
    Because the only NE of $\game$ is the undesirable outcome $(\OptOut, \OptOut)$,
        this simulation equilibrium will constitute a Pareto improvement over any NE of $\game$.
\end{proof}
% }

\begin{restatable}{corollary}{corrollaryDTG}\label{ex:DTG}
    Enabling $\mSim$ \iPiN{} in the trust-and-coordination game from \Cref{fig:DTG}.
\end{restatable}

\subsection{Mixed-Strategy Simulation and Privacy}
    \label{sec:sub:privacy}\label{SEC:SUB:PRIVACY}

\citet{kovarik2023game} show that
    pure-strategy simulation can sometimes be harmful to \textit{both} players.
An example that illustrates this dynamic is a scenario where Bob,
    after successfully cooperating with Alice, has to put all his profits into a password-protected account.
While Alice could always attempt to guess Bob's password, doing so would typically be futile.
However, if she had access to pure-strategy simulation,
    she would be able to predict Bob's password and steal his profits,
    so Bob would chose to not cooperate with Alice in the first place.
In contrast, if Alice only had access to mixed-strategy simulation,
    Bob could protect his profits by using a randomly-generated password,
    thus preserving the possibility of cooperation with Alice.

In \Cref{ex:password_guessing}, we give a general construction which adds this ``password-guessing'' dynamic into any base-game, 
allowing us to derive the following result.

\begin{restatable}{theorem}{thmMsimCanBeBetterThanPsim}\label{theorem:password_guessing}
    There are games where enabling $\mSim$ \iPiN{}, but pure-strategy simulation does not.
\end{restatable}

% \arxivOnly{
\begin{proof}[Proof sketch]
    A suitable $\game$ can be constructed by
        starting with the partial-trust game $\PTG$ (Fig.\,\ref{fig:PTG}),
        giving Bob the ability to $\OptOut$ of playing,
        and introducing the password-guessing dynamic.
    The only equilibrium of $\psimgame$ will be for Bob to $\OptOut$,
        which is strictly worse (for both) than Alice using $\WalkOut$
            (the only NE of $\game$).
    However, $\msimgame$ will have a simulation equilibrium which Pareto-improves on the $\WalkOut$ outcome.
\end{proof}
% }

%% file: literature.tex
% \section{Discussion}\label{sec:discussion}

% We conclude with an overview of the most closely related work
%     and summarising our contributions and promising future work.

\section{Related Work}\label{sec:sub:related-work}

% Pure strategy simulation
% Several earlier works study simulation in games, often with an emphasis on predicting artificial agents.
\citet{kovarik2023game} study a setting which is the closest to ours,
    but focus exclusively on the much stronger assumption of using pure-strategy simulation.
\citet{Vardy2004} study the same setting (i.e., what we refer to as pure-strategy simulation),
    but approach it from a more traditional economics angle,
    focusing on the simulated agent's value of commitment rather than on Pareto-improvements.
Simulation in the context of AI agents is also studied by \citet{chen2024screening},
    who distinguish between the ``screening'' and ``disciplining'' effects on the AI's behaviour.
    Unlike the present paper, this work assumes that the simulated agent is drawn from a fixed population (rather than being strategic).

% Simulation with uncertainty
Alternative formulations of games with simulation incorporate unpredictable randomisation in different ways.
In \textit{games with espionage}  \citep{Solan2004}, for example, \Plone{} can pay to gain a \textit{probabilistic} signal based on \Pltwo{}'s realised pure strategy.
    % use an information device $q : \PureStrategies_2 \to \Delta(\Omega)$ that produces a random, informative signal $\omega \sim q(\pureStrategy_2)$ based on \Pltwo{}'s realised strategy.
In \emph{oracle games} \citep{Young2020}, \Plone{} can attempt to learn \Pltwo{}'s pure strategy,
    but the success chance depends on the payment made by \Plone{}.
% In \emph{oracle games} \citep{Young2020}, \Plone{} can pay an oracle any amount $x$ for an $I(x)$ probability of learning \Pltwo{}'s pure strategy (and otherwise learning nothing), where $I$ is non-decreasing in $x$.
% \lh{Round off.}

% Sequential games
% Other authors consider simulation in the context of \emph{Stackelberg games}.
\citet{Harris2023}, study the problem of Bayesian persuasion under the assumption that the \textit{leader} can simulate the follower a number of times in order to learn what their response will be.
% \lh{Round off.}

% Incomplete information
Other related work exists in \textit{incomplete} information games, where a player pays to learn what others \textit{know} (rather than \textit{do}).
In mechanism design with \textit{(partially) verifiable information}, an agent's strategy might be restricted by their private information \citep{Green1986}, or this information might be revealed by the designer paying a cost \citep{Townsend1979,Ball2019}.
In some models of costly \textit{information acquisition},
    % players pay not (only) to gain information about the hidden state of the world,
    a player can pay to learn the others' observations of the hidden state \citep{Hellwig2009,Myatt2011,Denti2023}.
% \lh{Round off.}

% Mutual simulation
Finally, some other works consider the possibility of \textit{mutual} simulation,
    though they typically assume that simulation is not costly.
\citet{kovarik2024recursive} study a setting where
    the players observe the result of simulation jointly,
    but are uncertain as to whether they themselves might be in a simulation.
\citet{RobustProgramEquilibrium} shows that in games played between AI agents with mutual access to each other's code \citep{Tennenholtz2004}, simulation can lead to cooperation.
Another related approach is game theory with translucent players \cite{halpern2018translucent},
    which assumes that the players tentatively settle on some strategy from which they can deviate, but doing so has some chance of being visible to the other player.
    In our terminology, this corresponds to a setting where each player always performs free but unreliable simulation of the other player.
\citet{Capraro2019} show how translucency can lead to increased cooperation.

% \arxivOnly{
% The literature on Newcomb's problem \cite{nozick1969} in decision theory also studies problems in which one agent can predict another's choices. For example, Newcomb's problem itself and \textit{Parfit's hitchhiker}  \cite{Parfit1984,Barnes1997} can be viewed as a variant of the trust game with simulation and various other scenarios in that literature can be viewed as zero-sum games with the ability to predict \cite[Section 11]{Gibbard1981}\cite{Spencer2017,ExtractingMoneyFromCDT}. However, in this literature the predictor is generally not strategic. Instead, the predictor's behavior is assumed to be fixed (e.g, always simulate and best respond to the observed strategy). The game-theoretic, mutually strategic nature of our simulation games is essential to the present paper. The literature on Newcomb's problem instead focuses on more philosophical issues of how an agent should choose when being the subject of prediction.}

%% file: conclusions.tex
\section{Conclusion}\label{sec:sub:conclusion}

Strategic interactions involving AI agents are likely to become increasingly frequent and important. They may also be fundamentally different from more familiar interactions, as AI agents may -- in theory -- be more transparent than humans or human institutions. For example, it may be possible to simulate an AI agent, likely at a small cost.

In this paper, we studied the implications of costly simulation in the presence of unpredictable randomisation, showing that it is neither strictly weaker nor stronger than pure-strategy simulation (in terms of improving social welfare).
While determining whether enabling mixed-strategy simulation is beneficial turns out to be \NPH{} in general,
    we identified several classes of games where the effects of simulation can be predicted.
Concretely, we showed that mixed-strategy simulation can lead to increased cooperation
    in games where the simulator needs to decide whether to trust the other player, when either the simulator has a more nuanced set of options than just full trust and no trust,
        or the players face challenges with both trust and coordination.
% We saw that mixed-strategy simulation can lead to increased cooperation
%     in trust games where the simulator has a more nuanced set of options than just full trust and no trust,
%     and when the players face challenges with both trust and coordination.
However, we also saw that mixed-strategy simulation fails to foster cooperation if \Pltwo{} observes \Plone{}'s action before moving.

While mixed-strategy simulation is arguably a more realistic model than pure-strategy simulation, it is still limited in several important ways.
Future work should explore generalisations such as dynamic games, games with incomplete information, and other forms of simulation (such as by generating multiple, stochastic samples).
Other important directions include identifying further classes of game in which different forms of simulation can help, and in matching those to potential real-world domains of application.
In so doing, we will be better equipped to reap the benefits and avoid the risks once AI agents become widespread.

%% file: acknowledgements.tex
\section*{Acknowledgments}

We are grateful to Caspar Oesterheld, Ratip Emin Berker, and Jiayuan Liu for discussions regarding the paper,
and to Ivan Geffner and Drake Thomas for insights related to \Cref{prop:solve_upper_bound}.
We thank the Cooperative AI Foundation and Polaris Ventures (formerly the Center for Emerging Risk Research) for financial support.
Vojtech Kovarik received partial financial support from Czech Science Foundation grant no. GA22-26655S.
Lewis Hammond acknowledges the support of an EPSRC Doctoral Training Partnership studentship (reference: 2218880).

%% file: efgs.tex
\section{Extensive-Form Games}\label{sec:app:EFGs}

\begin{definition}[EFG]\label{def:EFG}
An \defword{extensive form game} is a tuple of the form $\EFG = \left< \playerSet, \actions, \histories, p, \pi_c, \mc I, u \right>$ for which
\begin{itemize}
 \item $\playerSet = \{1,\dots,N\}$ for some $N \in \mathbb{N}$,
 \item $\histories$ is a tree\footnote{Recall that a ``tree on $X$'' refers to a subset of $X^*$ that is closed under initial segments.} on $\actions$,
 \item $\actions$ and all sets $\actions(h) := \{ a\in \actions \mid ha \in \histories \}$, for $h\in \histories$, are compact,
 \item $p : \histories \setminus \leaves \to \playerSet \cup \{c\}$ (where $\leaves$ denotes the leaves of $\histories$),
 \item $\pi_c(h) \in \Delta (\actions(h))$ for $p(h)=c$,
 \item $u : \leaves \to \R^{N}$, and
 \item $\mc I = (\mc I_1,\dots, \mc I_N)$ is a collection of partitions of $\histories$, where each $\mc I_i$
    provides enough information to identify $i$'s legal actions.\footnote{
        That is, for every $I\in \mc I_i$, $p(h)$ is either equal to $i$ for all $h\in I$ or for no $h\in I$. If for all, then $\actions(h)$ doesn't depend on the choice of $h\in I$.
    }
\end{itemize}
\end{definition}

\noindent
Recall that every normal-form game can be represented as an EFG
    (by assuming that players select actions one by one, but only observe the actions of others after taking their own action).

A \defword{behavioural strategy} for player $\pl$ is a mapping
    \begin{align*}
        \policy_\pl
        :
        \left\{ \history \in \histories \mid \playerFunction(\history) = \pl \right\}
        \mapsto
        \policy_\pl(\history) \in \Delta(\actions(\history))
        .
    \end{align*}
By $\policies = \bigtimes_{\pl \in \playerSet} \policies_\pl$, we denote the set of all behavioural strategy profiles $\policy = (\policy_\pl)_{\pl \in \playerSet}$.
A behavioural strategy $\policy_\pl$ is said to be \textbf{pure} when
    for every $\history$ with $\playerFunction(\history) = \pl$ and every $\action \in \actions_\pl(\history)$,
    we have $\policy_\pl(\action) \in \{0, 1\}$.
Each $\policy \in \policies$ induces a probability distribution over the set $\leaves$ of leaves the EFG.
This allows us to define the \defword{expected utility} in $\EFG$ as
    $\utility_\pl(\policy)
        := \E_{\leaf \sim \policy} \utility_\pl(\leaf)
    $.
A behavioural strategy $\policy_\pl$ is said to be a \defword{best response} to $\policy_\opp$ if
    $\policy_\pl \in \arg \max_{\policy'_\pl \in \policies_\pl} \utility_\pl(\policy'_\pl, \policy_\opp)$.
% A \defword{Nash equilibrium} (NE)
%     % (or just ``an equilibrium'')
%     is a strategy profile $\policy$ under which each player's strategy $\policy_\pl$ is a best response to $\policy_\opp$.
%     We use $\NE(\game)$ to denote the set of all Nash equilibria of $\game$.

All of the definitions straightforwardly generalise to the case where rewards are received during the game
    (i.e., we have some reward function $\reward_\pl$ that assigns $\reward_\pl(\history, \action) \in \R$ to every $\action \in \actions(\history)$ and $\utility_\pl(\leaf) := \sum_{\history \action \sqsubset \leaf} \reward_\pl(\history, \action)$.
Moreover, they also generalise to the case of infinite games.
    While this might sometimes cause the expectations to be undefined or infinite, we will only work with games where this is not an issue.

A \textbf{normal-form representation} $\game$ of an extensive-form game $\EFG$
    is defined as the normal-form game
        $(\PureStrategies^\game_1, \dots, \PureStrategies^\game_N, \utility^\game)$
    given by
        \begin{align*}
            \PureStrategies_\pl^{\game}
            & :=
            \left\{
                \policy_\pl \in \policies_\pl
            \mid
                \policy_\pl
                \textnormal{ is pure}
            \right\}
            ,
            \\
            \utility^\game(\policy_1, \dots, \policy_N)
            & :=
            \utility^\EFG(\policy_1, \dots, \policy_N)
            .
        \end{align*}

%% file: proofs.tex
% === SYNTAX for the restatable environment stuff. Don't forget the asterisk. =====
% \restatableTheoremLabelGoesHere*
% \begin{proof}
%     \vk{TODO}
% \end{proof}
% ========== xxx ============

\section{Proofs \texorpdfstring{for \Cref{sec:new_defs} (Basic Properties)}{of Basic Properties of Simulation Games}}\label{sec:app:basic_properties}

In the remainder of the appendix, we present the proofs of the theoretical results described in the main text.
To make some of the lengthier proofs easier to navigate,
    we state their main steps as separate claims and prove these claims using a separate proof environment.

The remainder of this section gives the proofs related to the basic properties of simulation games,
    and in particular the reduction of $\msimgame$ to a strategically-equivalent finite subgame.

\lemmaNEofGAreNEofMsim*

\begin{proof}
    Let $\mixedStrategy \in \NE(\game)$.
    To prove that $(\mixedStrategy_1, \meta{\mixedStrategy}_2)$ is an \NE{} of $\msimgame$, we need to show that neither player has a profitable deviation.
    Recall that
        $\PureStrategies^{\msimgame}_1 = \PureStrategies^\game_1 \cup \{ \mSim \}$
        and
        $\PureStrategies^{\msimgame}_2 = \MixedStrategies^\game_2$.
    Since $\mixedStrategy$ is an \NE{},
        no $\pureStrategy_1 \in \PureStrategies^\game_1$ or $\mixedStrategy_2 \in \MixedStrategies^\game_2$
        can constitute a profitable deviation from $\mixedStrategy$
        --- and hence from $(\mixedStrategy_1, \meta{\mixedStrategy}_2)$ either.
    Moreover, for $\mSim$, we have
        \begin{align*}
            \utility^{\msimgame}_1(\mSim, \meta{\mixedStrategy}_2)
            & =
            \utility^\game_1(\br^\game(\mixedStrategy_2), \mixedStrategy_2)
                - \simcost
            \\
            & =
             \utility^\game_1(\mixedStrategy_1, \mixedStrategy_2)
                - \simcost
            \leq
            \utility^\game_1(\mixedStrategy_1, \mixedStrategy_2)
        .
        \end{align*}
    This shows that $(\mixedStrategy_1, \meta{\mixedStrategy}_2)$ admits no profitable deviation.
\end{proof}

\subsection{Reduction of Strategy Space in NFGs}
    \label{sec:app:general_NFG_reduction}

The following standard result ensures that when investigating the Nash equilibria of $\game$, it suffices to focus on the subset of pure strategies that are a best response to some mixed strategy.
(This is a weaker notion than rationalisability, which iteratively removes all strictly dominated strategies and all strategies which are not a best response to some strategy of the opponent.)

\begin{lemma}[Keeping only best responses]\label{lem:best_responses_only}
    Let $\game' = \left( \strategySubset_1, \strategySubset_2, \utility \right)$ be a subgame
    of a (possibly infinite) NFG $\game = \left( \PureStrategies_1, \PureStrategies_2, \utility \right)$
    and suppose that
        \begin{align*}
            \strategySubset_\pl
            \supset
            \left\{
                    \pureStrategy_\pl \in \PureStrategies_\pl
                \mid
                    \exists \mixedStrategy_\opp \in \MixedStrategies_\opp :
                    \pureStrategy_\pl \in \br(\mixedStrategy_\opp)
            \right\}
            .
        \end{align*}
    Then $\NE(\game) = \NE(\game')$.
\end{lemma}

\noindent
The proof is standard, but we provide it for completeness.

\begin{proof}
$\NE(\game') \subseteq \NE(\game)$:
    Suppose that $\mixedStrategy \in \MixedStrategies^\game$ is not a Nash equilibrium of $\game$.
    First,
        if $\mixedStrategy$ uses strategies that do not appear in $\game'$,
        then $\mixedStrategy$ cannot be an \NE{} of $\game'$
        and there is nothing to prove.
    Second,
        if $\supp(\mixedStrategy_\pl) \subseteq \strategySubset_\pl$ for both players,
        then we can use the fact that since $\mixedStrategy$ is not an \NE{} in $\game$
        to get that for (at least) one of the players, there exists some pure best-response $\pureStrategy_\pl \in \br^\game(\mixedStrategy_\opp)$ for which $\utility_\pl^\game(\pureStrategy_\pl, \mixedStrategy_\opp) > \utility_\pl^\game(\mixedStrategy)$.
        However, since $\strategySubset_\pl$ contains all pure best responses, $\pureStrategy_\pl$ is also available in $\game'$.
        This shows that player $\pl$ can unilaterally improve their utility by deviating from $\mixedStrategy_\pl$, so $\mixedStrategy$ is not an \NE{} of $\game'$.

$\NE(\game) \subseteq \NE(\game')$:
This part of the proof is analogous, except that we use the fact that when
    $\mixedStrategy_\pl$ is a \textit{possibly mixed} best-response to $\mixedStrategy_\opp$,
    it can be expressed as a convex combination of \textit{pure} best responses to $\mixedStrategy_\opp$.
And since all of these pure best responses are also available in $\game'$
    (and $\game'$ does not introduce any new strategies that could provide an even better best response),
    any NE in $\game$
    --- i.e., a pair of mutual best responses $(\mixedStrategy_\pl, \mixedStrategy_\opp)$ ---
    will also be an \NE{} in $\game'$.
\end{proof}

The next result shows that a further reduction can be achieved by removing strategies that can be expressed as convex combinations of other strategies.
Note that because the condition \eqref{eq:extremisation_pure} requires that the opponent's utility remains unchanged, this reduction removes fewer strategies than the removal of weakly dominated strategies, which in turn results in the removal of fewer Nash equilibria.

\begin{lemma}[Keeping only extremal points]\label{lem:extremal_points_only}
    Let $\game' = \left( \strategySubset_1, \strategySubset_2, \utility \right)$ be a subgame
    of a (possibly infinite) NFG $\game = \left( \PureStrategies_1, \PureStrategies_2, \utility \right)$.
    If $\game'$ satisfies the following condition
        \begin{align}
            &
            \left( \forall \pl \in \{1, 2\} \right)
            \left( \forall \pureStrategy_\pl \in \PureStrategies_\pl^\game  \right)
            ( \exists \mixedStrategy'_\pl \in \MixedStrategies_\pl^{\game'} )
                : 
                \forall \pureStrategy_\opp \in \PureStrategies_\opp^\game
            \label{eq:extremisation_pure}
            \\
            & \ \ \ 
                \utility_\pl(\mixedStrategy'_\pl, \pureStrategy_\opp)
                    \geq
                    \utility_\pl(\pureStrategy_\pl, \pureStrategy_\opp)
                \ \ \& \ \ 
                \utility_\opp(\mixedStrategy'_\pl, \pureStrategy_\opp)
                    =
                    \utility_\opp(\pureStrategy_\pl, \pureStrategy_\opp)
                ,
                \nonumber
        \end{align}
    then we have
        (i) $\NE(\game') \subseteq \NE(\game)$
        and 
            \begin{align*}
                \text{(ii) }
                    \left( \forall \mixedStrategy \in \NE(\game) \right)
                    \left( \exists \mixedStrategy' \in \NE(\game') \right)
                    :
                    \utility(\mixedStrategy') = \utility(\mixedStrategy)
                    .
            \end{align*}
\end{lemma}

\noindent
In other words,
    as long as we are primarily interested in equilibria for their expected values,
    keeping only the extremal points comes without any loss of generality.

\begin{proof}
Before we proceed with the proof,
note first that the assumption \eqref{eq:extremisation_pure} automatically implies the following stronger condition:
    \begin{align}
        &
        \left( \forall \pl \in \{1, 2\} \right)
        \left( \forall \mixedStrategy_\pl \in \MixedStrategies_\pl^\game  \right)
        ( \exists \mixedStrategy'_\pl \in \MixedStrategies_\pl^{\game'} )
            : 
            \forall \mixedStrategy_\opp \in \MixedStrategies_\opp^\game
        \label{eq:extremisation_mixed}
        \\
        & \ \ \ 
            \utility_\pl(\mixedStrategy'_\pl, \mixedStrategy_\opp)
                \geq
                \utility_\pl(\mixedStrategy_\pl, \mixedStrategy_\opp)
            \ \ \& \ \ 
            \utility_\opp(\mixedStrategy'_\pl, \mixedStrategy_\opp)
                =
                \utility_\opp(\mixedStrategy_\pl, \mixedStrategy_\opp)
        \nonumber
    \end{align}
(this follows from the fact that expected utility is a linear function).
For the purpose of this proof, we will use $\extremise{\mixedStrategy_\pl}$ to denote some strategy $\mixedStrategy'_\pl$ which serves as a witness that \eqref{eq:extremisation_mixed} holds for $\mixedStrategy_\pl$.
    \vk{This jump from \eqref{eq:extremisation_pure} to \eqref{eq:extremisation_mixed} hides some implicit assumptions about the set $\PureStrategies_\pl$ being finite, or the allowed probability distributions being well-behaved.
    But I think it's fair to ignore this.}

Proof of (i):
    Let $\mixedStrategy' \in \NE(\game')$.
    Suppose that $\mixedStrategy'$ wasn't an \NE{} in $\game$,
        that is,
            suppose there was some $\mixedStrategyAlt_\pl \in \MixedStrategies_\pl^\game$ such that
            $
                \utility_\pl(\mixedStrategyAlt_\pl, \mixedStrategy'_\opp)
                >
                \utility_\pl(\mixedStrategy'_\pl, \mixedStrategy'_\opp)
            $.
    Then by \eqref{eq:extremisation_mixed}, we have
        $
            \utility_\pl(\extremise{\mixedStrategyAlt_\pl}, \mixedStrategy'_\opp)
            \geq
            \utility_\pl(\mixedStrategyAlt_\pl, \mixedStrategy'_\opp)
        $,
    implying that
        $
            \utility_\pl(\extremise{\mixedStrategyAlt_\pl}, \mixedStrategy'_\opp)
            >
            \utility_\pl(\mixedStrategy'_\pl, \mixedStrategy'_\opp)
        $.
    This would mean that player $i$ could increase their utility in $\game'$ by unilaterally deviating to $\extremise{\mixedStrategyAlt_\pl}$, contradicting the assumption that $\mixedStrategy
    $ is an \NE{} of $\game'$.

Proof of (ii):
    Let $\mixedStrategy \in \NE(\game)$.
    We will show that the strategy $\mixedStrategy' := (\extremise{\mixedStrategy_1}, \extremise{\mixedStrategy_2})$ satisfies the conclusion of (ii).
    By applying \eqref{eq:extremisation_pure} twice, we get
        \begin{align*}
            \utility_\pl(\mixedStrategy_\pl, \mixedStrategy_\opp)
            \leq
            \utility_\pl(\extremise{\mixedStrategy_\pl}, \mixedStrategy_\opp)
            =
            \utility_\pl(\extremise{\mixedStrategy_\pl}, \extremise{\mixedStrategy_\opp})
            .
        \end{align*}
    However,
        since $\mixedStrategy$ is an \NE{} of $\game$ and doesn't allow for any profitable deviations,
        the inequality
            $
                \utility_\pl(\mixedStrategy_\pl, \mixedStrategy_\opp)
                \leq
                \utility_\pl(\extremise{\mixedStrategy_\pl}, \mixedStrategy_\opp)
            $
        cannot be strict,
        implying that
            $
                \utility_\pl(\mixedStrategy_\pl, \mixedStrategy_\opp)
                =
                \utility_\pl(\extremise{\mixedStrategy_\pl}, \extremise{\mixedStrategy_\opp})
            $.
        As a result, we have:
        \begin{align*}
            \utility(\mixedStrategy')
            =
            \utility(\mixedStrategy)
            .
        \end{align*}
    
    If $(\extremise{\mixedStrategy_1}  \! , \, \extremise{\mixedStrategy_2})$ wasn't an \NE{} of $\game'$,
        there would have to be some strategy $\rho'_\pl \in \mixedStrategy_\pl^{\game'}$ for which
        $
            \utility_\pl(\rho'_\pl, \extremise{\mixedStrategy_\opp})
            >
            \utility_\pl(\extremise{\mixedStrategy_\pl} \! , \, \extremise{\mixedStrategy_\opp})
        $.
        However, by \eqref{eq:extremisation_mixed}, we would then have
        \begin{align*}
            \utility_\pl(\rho'_\pl, \mixedStrategy_\opp)
            =
            \utility_\pl(\rho'_\pl, \extremise{\mixedStrategy_\opp})
            >
            \utility_\pl(\extremise{\mixedStrategy_\pl} \! , \, \extremise{\mixedStrategy_\opp})
            =
            \utility_\pl(\mixedStrategy_\pl, \mixedStrategy_\opp)
            ,
        \end{align*}
        contradicting the assumption that $\mixedStrategy$ is an \NE{} of $\game$.
\end{proof}

\vk{We could say something more here, about \Cref{lem:extremal_points_only}, because the proof actually shows more. In our actual setting, the strategy space $\PureStrategies^\game_\pl$ is convex and the mixed strategy $\mixedStrategy'_\pl$ is in fact such that the ``center'' of $\mixedStrategy'_\pl$ is $\pureStrategy'_\pl$.}

\subsection{Application of the Reduction to Mixed-Strategy Simulation Games}\label{sec:app:simgame_reduction}

By a (bounded convex) \textbf{polytope}, we will mean
    a convex closure of a finite number of points in a Euclidean space.
By a \textbf{face} of a polytope, we will mean
    any intersection of the polytope with a half-space such that none of the interior points of the polytope lie on the boundary of the half-space.
We will make use of the following standard result:

\begin{lemma}\label{lem:argmax_of_linear_on_polytope}
    Suppose that $d \in \N$,  $f : \R^d \to \R$ is a linear function, and $C \subseteq \R^d$ is a polytope.
    Then $\argmax \left\{ f(x) \mid x \in C \right\}$ is a polytope
        (more specifically, a face of $C$).
\end{lemma}

\begin{proof}
    This follows from the definition of a face, and the fact that a face of a polytope is a polytope.
\end{proof}

\propFiniteStrategySpace*

\noindent
The proof of this result includes Claims~\ref{cl:finiteMone}-\ref{cl:finiteMtwo} and their proofs.

\begin{proof}
    For the purpose of this proof,
        we will take the phrase ``the game $A$ can be reduced to a subgame $B$'' to mean that
        $B$ is a subgame of $A$ and they satisfy the conditions (i) and (ii) from the statement of the proposition
        --- except that this will not necessarily require $B$ to be finite.
    Note that these reductions are transitive:
        if $A$ can be reduced to $B$ and $B$ can be reduced to $C$, then $A$ can be reduced to $C$.
    We will use $(\symbolPlaceholder)'$ to denote that a particular object belongs to a subgame
    -- e.g., $\mixedMetaStrategy'_1$ a subgame strategy of \Plone{}.

    We will actually prove a stronger result than \Cref{prop:finite_str_space},
        by showing that the conclusion holds even if \Plone{} were allowed to mix over mixed strategies in $\game$
        (i.e., that it would hold even for $\PureStrategies^{\msimgame}_1$ were defined as $\MixedStrategies^\game_1 \cup \{ \mSim \}$).
    The same argument would apply to showing that \Pltwo{} cannot benefit from additional mixing (over mixtures of mixtures),
        justifying the current form of the definition of mixed-strategy simulation games.
    
    Let
        $\game$ be the finite base-game and
        $\msimgame = (\PureMetaStrategies_1, \PureMetaStrategies_2, \utility)$,
            where
                $\PureMetaStrategies_1 = \PureStrategies_1^\game \cup \{ \mSim \}$ and
                $\PureMetaStrategies_2 = \MixedStrategies_2^\game$,
            be the corresponding mixed-strategy simulation game.
    To find the reduction, we will construct sets
        $\FinitePMS_1 \subseteq \PureMetaStrategies_1$
        and
        $\FinitePMS_2 \subseteq \ReducedPMS_2 \subseteq \PureMetaStrategies_2$ such that:
        \begin{itemize}
            \item[(1)]  (a) $(\PureMetaStrategies_1, \PureMetaStrategies_2, \utility)$
                            can be reduced to
                            $(\FinitePMS_1, \PureMetaStrategies_2, \utility)$.\\
                        (b) $\FinitePMS_1$ is finite.
            \item[(2)] $(\FinitePMS_1, \PureMetaStrategies_2, \utility)$
                            can be reduced to
                            $(\FinitePMS_1, \ReducedPMS_2, \utility)$.
            \item[(3)]  (a) $(\FinitePMS_1, \ReducedPMS_2, \utility)$
                            can be reduced to
                            $(\FinitePMS_1, \FinitePMS_2, \utility)$.\\
                        (b) $\FinitePMS_2$ is finite.
        \end{itemize}
    Since the reductions are transitive, this will suffice to prove the main result.
    In remainder of the proof, we will find the above sets and show that they satisfy the above conditions.

    \begin{claim}\label{cl:finiteMone}
        The conditions (1a) and (1b) hold for
        \begin{align}\label{eq:MOnePrime}
                \FinitePMS_1
                :=
                \PureStrategies_1^\game
                \cup
                \left\{ \mSim \right\}
            .
            \end{align}
    \end{claim}

    \noindent
    \Cref{cl:finiteMone} is an essentially trivial application of \Cref{lem:extremal_points_only}
        --- however, it is somewhat complicated\footnotemark{} by the need to deal with the notation around the simulation ``meta'' games:
        \footnotetext{
            \ The reader might take solace in the fact that after this proof, we will essentially always work with the reduced strategy spaces, which will obviate most of the technical and notational difficulties.
        }
    We need to verify that the condition \eqref{eq:extremisation_pure} holds.
    This is trivial for $\pl = 2$, since \Pltwo{}'s strategy space remains unchanged.
    For $\pl = 1$, the condition requires us to
        start with an arbitrary to pure (meta) strategy $\pureMetaStrategy_1$ in $\msimgame$
        and find a corresponding mixed (meta) strategy $\mixedMetaStrategy'_1$
            in the subgame $(\FinitePMS_1, \PureMetaStrategies_2, \utility)$
        (in such a way that $\pureMetaStrategy_1$ and $\mixedMetaStrategy'_1$ yield identical utilities against every pure strategy $\pureMetaStrategy_2$ of \Pltwo{}).
    In practice, this is simple to achieve:
        since $\PureStrategies_1^{\msimgame} = \MixedStrategies_1^\game \cup \{ \mSim \}$
            (by definition),
        a pure (meta) strategy of \Plone{} in $\msimgame$
            is either the action $\mSim$
            or a some mixed strategy $\mixedStrategy_1 \in \MixedStrategies_1^\game$.
        In the first case,
        the simulation action $\mSim$ is available in the subgame $(\FinitePMS_1, \PureMetaStrategies_2, \utility)$,
            so we can simply set $\mixedMetaStrategy'_1 := \pureMetaStrategy_1 := \mSim$.
        In the case when $\pureMetaStrategy_1 = \mixedStrategy_1 \in \MixedStrategies_1^\game$,
            we simply set
                $
                    \mixedMetaStrategy'_1
                    :=
                    \sum_{\pureStrategy_1 \in \PureStrategies_1^\game}
                        \mixedStrategy_1(\pureStrategy_1) \cdot \meta{\pureStrategy_1}
                $
            (where $\meta{\pureStrategy_1}$ denotes the meta-strategy of putting all probability mass on the base-game strategy $\pureStrategy_1$).
        By definition of $\msimgame$,
            all pure meta-strategies of \Pltwo{} are of the form $\pureMetaStrategy_2 = \meta{\mixedStrategy_2}$ for some $\mixedStrategy_2 \in \MixedStrategies_2^\game$,
            and we have
                $\utility(\meta{\pureStrategy_1}, \meta{\mixedStrategy_2}) = \utility(\pureStrategy_1, \mixedStrategy_2)$
                for every $\pureStrategy_1 \in \PureStrategies_1^\game$, $\mixedStrategy_2 \in \MixedStrategies_2^\game$.
        This implies that
            $\mixedStrategy_1$ and $\mixedMetaStrategy'_1$ yield the same expected utilities against every $\pureMetaStrategy_2 \in \PureMetaStrategies_2$.
    This in turn implies that $(\FinitePMS_1, \PureMetaStrategies_2, \utility)$ satisfies the condition \eqref{eq:extremisation_pure},
        which concludes the proof of \Cref{cl:finiteMone}.

    \begin{claim}\label{cl:reducedMtwo}
        The condition (2) holds for
        \begin{align*}
            \ReducedPMS_2
            :=
            \left\{
                    \mixedStrategy_2 \in \MixedStrategies_2^\game
                \mid
                    \exists \pureMetaStrategy_1 \in \PureStrategies_1^\game \cup \{ \mSim \}
                    :
                    \mixedStrategy_2 \in \br(\pureMetaStrategy_1)
            \right\}
            .
        \end{align*}
    \end{claim}

    \noindent
    \Cref{cl:reducedMtwo} is a trivial application of \Cref{lem:best_responses_only} to
        ``$\game$'' $ = (\FinitePMS_1, \PureMetaStrategies_2, \utility)$
        and
        ``$\, \game'\,$'' $= (\FinitePMS_1, \ReducedPMS_2, \utility)$.
    
    \begin{claim}\label{cl:convexFbr}
        For any $\pureStrategy_1 \in \PureStrategies_1$, denote
        \begin{align*}
            \inverseBR(\pureStrategy_1)
            & :=
            \left\{
                \mixedStrategy_2 \in \MixedStrategies_2^\game
                \mid
                \br(\mixedStrategy_2) \ni \pureStrategy_1
            \right\}
            \\
            \inverseFBR(\pureStrategy_1)
            & :=
            \left\{
                \mixedStrategy_2 \in \MixedStrategies_2^\game
                \mid
                \fbr(\mixedStrategy_2) \ni \pureStrategy_1
            \right\}
            .
        \end{align*}
        \begin{enumerate}[label=(\roman*)]
            \item The set $\inverseBR(\pureStrategy_1)$ is closed and convex.
            \item The set $\inverseFBR(\pureStrategy_1)$ is convex.
        \end{enumerate}
    \end{claim}
    
    \begin{proof}[Proof of \Cref{cl:convexFbr}]
        (i): This is standard.
            (The ``closed'' part holds because utility functions are continuous.
            The ``convex'' part follows from the fact that utility functions are linear.)
        
        (ii): First, recall that
            \begin{align*}
                \fbr(\mixedStrategy_2)
                :=
                \argmax \left\{
                        \utility_2(\pureStrategyAlt_1, \mixedStrategy_2)
                    \mid
                        \pureStrategyAlt_1 \in \br(\mixedStrategy_2)
                    \right\}
                .
            \end{align*}
        
        To prove convexity of $\inverseFBR(\pureStrategy_1)$,
            let $\mixedStrategy_2 := \lambda \cdot \mixedStrategy^1_2 + (1-\lambda) \cdot \mixedStrategy^2_2$
            for some $\mixedStrategy_2^1, \mixedStrategy_2^2 \in \inverseFBR(\pureStrategy_1)$ and $\lambda \in (0, 1)$.
            Since expected utilities in $\game$ are convex, we have
                \begin{align}\label{eq:convexity_fbr}
                    \utility_\pl(\pureStrategyAlt_1, \mixedStrategy_2)
                    =
                        \lambda \cdot \utility_\pl(\pureStrategyAlt_1, \mixedStrategy^1_2)
                        + 
                        (1-\lambda) \cdot \utility_\pl(\pureStrategyAlt_1, \mixedStrategy^2_2)
                \end{align}
            for any $\pureStrategyAlt_1 \in \PureStrategies^\game_1$ and $\pl \in \{1, 2\}$.

        Note that any favourable best response is, definitionally a best response too,
            so we have
            $\pureStrategy_1 \in \br(\mixedStrategy^1_2)$ and $\pureStrategy_1 \in  \br(\mixedStrategy^2_2)$.
        This implies that we must also have $\pureStrategy_1 \in \br(\mixedStrategy_2)$.
            (If this did not hold, then there would be some $\pureStrategy'_1$
                such that $\utility_1(\pureStrategy'_1, \mixedStrategy_2) > \utility_1(\pureStrategy_1, \mixedStrategy_2)$.
            By \eqref{eq:convexity_fbr}, this would imply that
                either $\lambda \cdot \utility_1(\pureStrategy'_1, \mixedStrategy^1_2) > \lambda \cdot \utility_1(\pureStrategy_1, \mixedStrategy^1_2)$
                or $(1-\lambda) \cdot \utility_1(\pureStrategy'_1, \mixedStrategy^2_2) > (1-\lambda) \cdot \utility_1(\pureStrategy_1, \mixedStrategy^2_2)$,
            which would in turn contradict the assumption that
                $\pureStrategy_1$ is a best response to $\mixedStrategy^1_2$ and $\mixedStrategy^2_2$.)
    
        All that remains to show is that $\mixedStrategy_2$ is a \textit{favourable} best response to $\pureStrategy_1$ -- that is, that
            there cannot exist some $\pureStrategyAlt_1 \in \br(\mixedStrategy_2)$ such that
            $
                \utility_2(\pureStrategyAlt_1, \mixedStrategy_2)
                >
                \utility_2(\pureStrategy_1, \mixedStrategy_2)
            $.
        Suppose, for a contradiction, that some such $\pureStrategyAlt_1$ did exist.
        By \eqref{eq:convexity_fbr}, we would then have
            $\utility_2(\pureStrategyAlt_1, \mixedStrategy^1_2) > \utility_2(\pureStrategy_1, \mixedStrategy^1_2)$
            or $\utility_2(\pureStrategyAlt_1, \mixedStrategy^2_2) > \utility_2(\pureStrategy_1, \mixedStrategy^2_2)$.
        
        Without loss of generality, assume that
            $\utility_2(\pureStrategyAlt_1, \mixedStrategy^1_2) > \utility_2(\pureStrategy_1, \mixedStrategy^1_2)$.
        Because $\pureStrategy_1$ is a \textit{favourable} best response to $\mixedStrategy^1_2$,
            this would mean that $\pureStrategyAlt_1$ cannot be a best response to $\mixedStrategy^1_2$.
        In other words, we would have
            $
                \utility_1(\pureStrategyAlt_1, \mixedStrategy^1_2)
                <
                \utility_1(\pureStrategy_1, \mixedStrategy^1_2)
            $.
        However,
            since both $\pureStrategyAlt_1$ and $\pureStrategy_1$ are best responses to $\mixedStrategy_2$,
            we also have
            $
                \utility_1(\pureStrategyAlt_1, \mixedStrategy_2)
                =
                \utility_1(\pureStrategy_1, \mixedStrategy_2)
            $.
        And because
            utilities are linear
            and
            $\mixedStrategy_2$ is a convex combination of $\mixedStrategy^1_2$ and $\mixedStrategy^2_2$,
        this would imply that
            $
                \utility_1(\pureStrategyAlt_1, \mixedStrategy^2_2)
                >
                \utility_1(\pureStrategy_1, \mixedStrategy^2_2)
            $.
        However, this would contradict the assumption that $\pureStrategy_1$ is a best response to $\mixedStrategy^2_2$.
        % The first inequality implies that $\pureStrategyAlt_1 \notin \br(\mixedStrategy^1_2)$, as otherwise we would not have $\pureStrategy_1 \in \fbr(\mixedStrategy^1_2)$, and similarly for second inequality.
        % On the other hand, because $\pureStrategy_1, \pureStrategyAlt_1 \in \br(\mixedStrategy_2)$ we know that:
        % \begin{align*}
        %     \utility_1(\pureStrategyAlt_1,\mixedStrategy_2)
        %     &= \utility_1(\pureStrategyAlt_1,\mixedStrategy_2)\\
        %     &= \lambda \cdot \utility_1(\pureStrategy_1, \mixedStrategy^1_2)
        %     + (1-\lambda) \cdot \utility_1(\pureStrategy_1, \mixedStrategy^2_2)\\
        %     &= \lambda \cdot \max_{\pureStrategy'_1} \utility_1(\pureStrategy'_1, \mixedStrategy^1_2)
        %     + (1-\lambda) \cdot \max_{\pureStrategy'_1} \utility_1(\pureStrategy'_1, \mixedStrategy^2_2),
        % \end{align*}
        % where the final line follows from the fact that $\pureStrategy_1 \in \br(\mixedStrategy^1_2)$ and $\pureStrategy_1 \in  \br(\mixedStrategy^2_2)$.
        % But then clearly we must also have $\utility_1(\pureStrategyAlt_1, \mixedStrategy^1_2) = \max_{\pureStrategy'_1} \utility_1(\pureStrategy'_1, \mixedStrategy^1_2)$ and $\utility_1(\pureStrategyAlt_1, \mixedStrategy^2_2) = \max_{\pureStrategy'_1} \utility_1(\pureStrategy'_1, \mixedStrategy^2_2)$, as otherwise \Plone{} could deviate to $\pureStrategy_1$ in order to increase their expected utility.
        % This contradicts our earlier result that $\pureStrategyAlt_1 \notin \br(\mixedStrategy^1_2)$ or $\pureStrategyAlt_1 \notin \br(\mixedStrategy^2_2)$, concluding the proof.
    \end{proof}

    % \begin{claim}\label{cl:finiteMtwo}
    %     The conditions (3a) and (3b) hold for
    %     \begin{align*}
    %         \FinitePMS_2
    %         & :=
    %         \ReducedPMS_2
    %         \cap
    %         \bigcup_{\strategySubset_1 \subseteq \PureStrategies_1}
    %         \bigcup_{\pureStrategy_1 \in \strategySubset_1}
    %             \extremalPoints\left(
    %                 \closure{\specialPolytope(\pureStrategy_1, \strategySubset_1)}
    %             \right),
    %         \\
    %         \specialPolytope(\pureStrategy_1, \strategySubset_1)
    %         & :=
    %         \left\{
    %             \mixedStrategy_2 \in \MixedStrategies_2
    %             \mid
    %                 \br(\mixedStrategy_2) = \strategySubset_1, \ 
    %                 \fbr(\mixedStrategy_2) \ni \pureStrategy_1
    %         \right\}
    %         ,
    %     \end{align*}
    %     where
    %         $\closure{Y}$ denotes the closure of $Y$,
    %         $\extremalPoints(Z)$ denotes the set of all extremal points of $Z$.
    % \end{claim}
    \begin{claim}\label{cl:finiteMtwo}
        The conditions (3a) and (3b) hold for
        \begin{align*}
            \FinitePMS_2
            & :=
            \ReducedPMS_2
            \cap
            \bigcup_{\pureStrategy_1 \in \PureStrategies^\game_1}
                \extremalPoints\left(
                    \closure{
                        \inverseFBR(\pureStrategy_1)
                    }
                \right)
            ,
        \end{align*}
        where
            $\closure{Y}$ denotes the topological closure of $Y$
            and
            $\extremalPoints(Z)$ denotes the set of all extremal points of $Z$.
    \end{claim}
        % \begin{align*}
        %             \inverseBR(\pureStrategy_1)
        %             & := \left\{
        %                 \mixedStrategy_2 \in \MixedStrategies_2^\game
        %                 \mid
        %                 \br(\mixedStrategy_2) \ni \pureStrategy_1
        %             \right\}
        %             \\
        %             \inverseFBR(\pureStrategy_1)
        %             & := \left\{
        %                 \mixedStrategy_2 \in \MixedStrategies_2^\game
        %                 \mid
        %                 \fbr(\mixedStrategy_2) \ni \pureStrategy_1
        %             \right\}
        %             .
        % \end{align*}

    \begin{proof}[Proof of \Cref{cl:finiteMone}]
    % To prove (3b)
    %     --- i.e., that $\FinitePMS_2$ is finite ---
    %     it is enough to observe that
    %     for every $\pureStrategy_1$ and $\strategySubset_1$,
    %         $\closure{\specialPolytope(\pureStrategy_1, \strategySubset_1)}$ is a polytope
    %             (i.e., a convex closure of a finite number of points
    %             or, equivalently,
    %             an intersection of a finite number of closed half-spaces in $\R^{\PureStrategies_2}$).
    To show that $\FinitePMS_2$ satisfies (3b)
        --- i.e., that $\FinitePMS_2$ is finite ---
        we will show that each set
                $\closure{\inverseFBR(\pureStrategy_1)}$
            is a polytope
                (i.e., a convex closure of a finite number of points
                or, equivalently,
                an intersection of a finite number of closed half-spaces in $\R^{\PureStrategies_2}$).
        We will do this by
            writing $\closure{\inverseFBR(\pureStrategy_1)}$ as a union of a finite number of polytopes.
            In conjunction with the fact that $\closure{\inverseFBR(\pureStrategy_1)}$ is convex (\Cref{cl:convexFbr}),
                this will imply that $\closure{\inverseFBR(\pureStrategy_1)}$ must itself be a polytope.
    To show this, note that $\inverseFBR(\pureStrategy_1)$ can be trivially rewritten as
        \begin{align*}
            & \inverseFBR(\pureStrategy_1) =
            \\
            & =
                \left\{
                    \mixedStrategy_2 \in \MixedStrategies^\game_2
                    \mid
                    \pureStrategy_1 \in \fbr(\mixedStrategy_2)
                \right\}
            \\
            & =
                \bigcup\nolimits_{\strategySubset_1 \subseteq \PureStrategies^\game_1, \, \strategySubset_1 \ni \pureStrategy_1}
                    \Big\{
                        \mixedStrategy_2 \in \MixedStrategies^\game_2
                    \mid
                        \pureStrategy_1 \in \fbr(\mixedStrategy_2)
                        \ \& \ 
                        \\
                        &
                        \phantom{aasdfasfasdfasaaaaaaaaaa}
                        \ \& \ 
                        \forall \pureStrategy'_1 \in \strategySubset_1 : 
                            \pureStrategy'_1 \in \br(\mixedStrategy_2)
                    \Big\}
            \\
            & =:
                \bigcup\nolimits_{\strategySubset_1 \subseteq \PureStrategies^\game_1, \, \strategySubset_1 \ni \pureStrategy_1}
                    \specialPolytope(\pureStrategy_1, \strategySubset_1)
            .
        \end{align*}

    We will show that each of the sets $\closure{\specialPolytope(\pureStrategy_1, \strategySubset_1)}$ is a polytope.
    To see this, observe first that
        a strategy $\mixedStrategy_2$ belongs to $\specialPolytope(\pureStrategy_1, \strategySubset_1)$ if and only if it satisfies:
    \begin{enumerate}[label={(\roman*)}]
        \item Every element of $\strategySubset_1$ is a best response to $\mixedStrategy_2$:
            \begin{align*}
                (\forall \pureStrategy'_1 \in \strategySubset_1)
                (\forall \pureStrategyAlt_1 \in \PureStrategies_1^\game) :
                    \utility_1(\pureStrategy'_1, \mixedStrategy_2)
                    \geq
                    \utility_1(\pureStrategyAlt_1, \mixedStrategy_2)
                .
            \end{align*}
        \item There are no other best responses to $\mixedStrategy_2$:
            \begin{align*}
                (\forall \pureStrategy'_1 \in \strategySubset_1)
                (\forall \pureStrategyAlt_1 \in \PureStrategies_1^\game \setminus \strategySubset_1) :
                    \utility_1(\pureStrategy'_1, \mixedStrategy_2)
                    <
                    \utility_1(\pureStrategyAlt_1, \mixedStrategy_2)
                .
            \end{align*}
        \item There is no best response \textit{in the subset $\strategySubset_1$} that gives \Pltwo{} more utility than $\pureStrategy_1$:
            \begin{align*}
                (\forall \pureStrategy'_1 \in \strategySubset_1) :
                    \utility_2(\pureStrategy_1, \mixedStrategy_2)
                    \geq
                    \utility_2(\pureStrategy'_1, \mixedStrategy_2)
                .
            \end{align*}
    \end{enumerate}
    When we take the closure of $\specialPolytope(\pureStrategy_1, \strategySubset_1)$,
        the only change is that the strict inequalities change to non-strict ones,
        which is equivalent to removing the condition (ii).
    % (This in particular implies that for every
    %     $
    %         \mixedStrategyAlt_2
    %         \in
    %         \closure{\specialPolytope(\pureStrategy_1, \strategySubset_1)}
    %     $,
    %     $\pureStrategy_1$ is a best response to $\mixedStrategyAlt_2$.
    %     This observation will be useful later, in the proof of (3a).)
    Since $\utility$ is linear,
        these conditions specify an intersection of closed half-planes,
        which in turn implies that $\closure{\specialPolytope(\pureStrategy_1, \strategySubset_1)}$ is a polytope.\footnotemark{}
            \footnotetext{
                In fact, together with the fact that $\closure{\inverseFBR(\pureStrategy_1)}$ is a polytope,
                this proof implies that
                    each of the sets $\closure{\specialPolytope(\pureStrategy_1, \strategySubset_1)}$
                    is a face of $\closure{\inverseFBR(\pureStrategy_1)}$.
            }
    This shows that $\FinitePMS_2$ satisfies (3b).

    To show that $\FinitePMS_2$ satisfies (3a),
        we need to verify the condition \eqref{eq:extremisation_pure} for ``$\game$'' $=(\FinitePMS_1, \ReducedPMS_2, \utility)$ and `` $\game'$ '' $=(\FinitePMS_1, \FinitePMS_2, \utility)$.
    Since the strategy space of \Plone{} is the same in both games, this comes down to verifying that
        for every pure strategy of \Pltwo{} in the game $(\FinitePMS_1, \ReducedPMS_2, \utility)$,
        we can find a mixed strategy of \Pltwo{} in the subgame $(\FinitePMS_1, \FinitePMS_2, \utility)$
        which yields
            the same $\utility_1(\symbolPlaceholder)$,
            and the same or higher $\utility_2(\symbolPlaceholder)$,
            against every pure strategy of \Plone{} in the ``full'' game $(\FinitePMS_1, \ReducedPMS_2, \utility)$.
    Into our specific setting, it suffices to show that
        for every
            $
                \mixedStrategy_2
                \in
                \MixedStrategies_2^\game
            $
        (that is a best response to some $\mixedMetaStrategy_1 \in \Delta(\PureStrategies^\game_1 \cup \{\mSim\})$)
        we can find a convex combination
            \begin{align}
                \mixedMetaStrategy_2
                =
                    \sum_j
                        \lambda_j \meta{\mixedStrategyAlt}_2^j
                \in
                    \Delta(\FinitePMS_2)
            \end{align}
        such that
            \begin{align}
                \utility(\pureStrategy_1, \mixedStrategy_2)
                & =
                    \utility(\pureStrategy_1, \mixedMetaStrategy_2)
                    \ \ \textnormal{ for every } \pureStrategy_1 \in \PureStrategies_1
                    \label{eq:threeB_nonsim_equality}
                \\
                \utility_1(\mSim, \mixedStrategy_2)
                & =
                    \utility_1(\mSim, \mixedMetaStrategy_2)
                    \label{eq:threeB_plone_sim}
                \\
                \utility_2(\mSim, \mixedStrategy_2)
                & \geq
                    \utility_2(\mSim, \mixedMetaStrategy_2)
                    \label{eq:threeB_pltwo_sim}
                .
            \end{align}
    \vk{I strongly suspect that the remainder of the proof can be simplified somehow. (If we didn't need to verify that the vertices are best responses to some $\mixedMetaStrategy_1$, it could be, like, two sentences.) But this is not a priority at this point.}

    As a first step towards verifying that this condition holds,
        let $\mixedStrategy_2 \in \MixedStrategies^\game_2$
            be a best response to some $\mixedMetaStrategy_1 \in \Delta(\PureStrategies_1^\game \cup \{ \mSim \})$.
        We can write this $\mixedMetaStrategy_1$ as
            $
                \mixedMetaStrategy_1
                =:
                \simProb \cdot \mSim
                +
                (1-\simProb) \cdot \mixedStrategy_1
            $
            for some $\simProb \in [0, 1]$, $\mixedStrategy_1 \in \MixedStrategies_1^\game$.
        Denote by $\pureStrategy_1$ some element of $\fbr(\mixedStrategy_2)$.
        As an auxiliary step, we will show that
            $\mixedStrategy_2$ can be written as a convex combination of vertices $\pureStrategyAlt_2^j$ of $\closure{\inverseFBR(\pureStrategy_1}$
            that are also best responses to $\mixedMetaStrategy_1$.
            (Because $\closure{\inverseFBR(\pureStrategy_1}$ is a convex polytope,
                expressing $\mixedStrategy_2$ as a convex combination of its vertices is trivial.
                However, some care will need to be taken to verify that the specific vertices $\pureStrategyAlt_2^j$ belong to $\FinitePMS_2$
                    -- i.e., that $\pureStrategyAlt_2^j$ is a best response to some strategy of \Plone{}, in this case $\mixedMetaStrategy_1$.)
        
    To find prepare for finding the suitable vertices $\pureStrategyAlt_2^j$,
        consider the first the following two functions defined on $\MixedStrategies_2^\game$:
    \begin{align*}
        f_{\pureStrategy_1}(\mixedStrategyAlt_2)
            := \, &
            \utility_2(\simProb \cdot \pureStrategy_1 \ \, + \ \, (1-\simProb) \cdot \mixedStrategy_1, \mixedStrategyAlt_2)
        \\
        f_{\mSim}(\mixedStrategyAlt_2)
            := \, &
            \utility_2(\simProb \cdot \mSim + (1-\simProb) \cdot \mixedStrategy_1, \mixedStrategyAlt_2)
            \\
            = \, &
            \utility_2(\mixedMetaStrategy_1, \mixedStrategyAlt_2)
        .
    \end{align*}
    As we will note below, $f_{\pureStrategy_1}$ has the following properties:
    \begin{enumerate}[label={(\alph*)}]
        \item $f_{\pureStrategy_1}$ is linear (unlike $f_{\mSim}$).
        \item $f_{\pureStrategy_1}(\mixedStrategyAlt_2) = \utility_2(\mixedMetaStrategy_1, \mixedStrategyAlt_2)$
            for any $\mixedStrategyAlt_2$ s.t. $\fbr(\mixedStrategyAlt_2) \ni \pureStrategy_1$.
        \item $f_{\pureStrategy_1}(\mixedStrategyAlt_2) \leq \utility_2(\mixedMetaStrategy_1, \mixedStrategyAlt_2)$
            for any $\mixedStrategyAlt_2$ such that $\br(\mixedStrategyAlt_2) \subseteq \PureStrategies'_1$,
            and in particular for any $\mixedStrategyAlt_2 \in \closure{\inverseFBR(\pureStrategy_1)}$.
            % and in particular for any $\mixedStrategyAlt_2 \in \closure{\specialPolytope(\pureStrategy_1, \strategySubset_1)}$.
        \item $
                A
                :=
                \argmax \left\{
                    f_{\pureStrategy_1}(\mixedStrategyAlt_2)
                    \mid
                    \mixedStrategyAlt_2 \in \closure{\inverseFBR(\pureStrategy_1)}
                    % \mixedStrategyAlt_2 \in \closure{\specialPolytope(\pureStrategy_1, \strategySubset_1)}
                \right\}
            $
            is a face of the polytope $\closure{\inverseFBR(\pureStrategy_1)}$.
            % is a face of the polytope $\closure{\specialPolytope(\pureStrategy_1, \strategySubset_1)}$.
        \item $\mixedStrategy_2 \in A$.
        \item Any element of $A$ is a best response to $\mixedMetaStrategy_1$.
        \item $
                \extremalPoints(A)
                \subseteq
                \extremalPoints\left(\closure{\inverseFBR(\pureStrategy_1)}\right)
                % \extremalPoints\left(\closure{\specialPolytope(\pureStrategy_1, \strategySubset_1)}\right)
                \cap \ReducedPMS_2
                \subseteq
                \FinitePMS_2
            $.
    \end{enumerate}
        We will now show the conditions (a-g).
            Note that the proofs of the later properties implicitly use the previous properties.
        The condition (a) holds because $\utility_2 = \utility^\game_2$ is linear.
        (b) follows from the definition of $\utility(\mSim, \mixedStrategyAlt_2)$.
        (c) holds because if the set of best responses is strictly larger than $\strategySubset_1$, it might contain some best response which gives higher utility to \Pltwo{} than $\pureStrategy_1$.
            (So if $\mixedMetaStrategy_1$ puts positive probability on $\mSim$,
                \Pltwo{} could get \textit{strictly} higher utility $\utility_2(\mixedMetaStrategy_1, \mixedStrategyAlt_2)$ than $f_{\pureStrategy_1}(\mixedStrategyAlt_2)$.)
        (d) follows from \Cref{lem:argmax_of_linear_on_polytope}.
        To show (e), suppose that
                $
                    f_{\pureStrategy_1}(\mixedStrategy_2)
                    <
                    f_{\pureStrategy_1}(\mixedStrategyAlt_2)
                $
                for some $\mixedStrategy_2 \in A$.
            We would then have
                $
                    \utility_2(\mixedMetaStrategy_1, \mixedStrategy_2)
                    =
                    f_{\pureStrategy_1}(\mixedStrategy_2)
                    <
                    f_{\pureStrategy_1}(\mixedStrategyAlt_2)
                    \leq
                    \utility_2(\mixedMetaStrategy_1, \mixedStrategyAlt_2)
                $,
                contradicting the assumption that $\mixedStrategy_2$ is a best response to $\mixedMetaStrategy_1$.
        (f) holds because for any $\mixedStrategyAlt_2 \in A$, we have
            $
                \utility_2(\mixedMetaStrategy_1, \mixedStrategyAlt_2)
                \geq
                f_{\pureStrategy_1}(\mixedStrategyAlt_2)
                =
                \max_{\mixedStrategyAlt'_2 \in A}
                    f_{\pureStrategy_1}(\mixedStrategyAlt'_2)
                =
                f_{\pureStrategy_1}(\mixedStrategy_2)
                =
                \max_{\mixedStrategy'_2 \in \MixedStrategies_2}
                    \utility_2(\mixedMetaStrategy_1, \mixedStrategy'_2)
            $.
        (g) follows from the combination of (d) and (f), with the second inclusion being just the definition of $\FinitePMS_2$.
    
    We are now ready to construct the advertised mixed (meta-) strategy $\mixedMetaStrategy_2$.
    From (d) and (e), it follows that we can construct $\mixedStrategy_2$ as a convex combination
        $
            \mixedStrategy_2
            =
            \sum_j
                \lambda_j \mixedStrategyAlt_2^j
        $
        of some strategies $\mixedStrategyAlt_2^j \in \extremalPoints(A)$.
    We define $\mixedMetaStrategy_2$ as
        \begin{align*}
            \mixedMetaStrategy_2
            := 
            \sum\limits_j
                \lambda_j
                \cdot
                \meta{\mixedStrategyAlt}_2^j
            .
        \end{align*}
    By (g), $\mixedMetaStrategy_2$ belongs to $\Delta(\FinitePMS_2)$.
    As a result, it remains to verify that $\mixedMetaStrategy_2$ satisfies
        the conditions \eqref{eq:threeB_nonsim_equality}, \eqref{eq:threeB_plone_sim}, and \eqref{eq:threeB_pltwo_sim}.
    To show \eqref{eq:threeB_plone_sim}, note that by linearity of the function
        $\mixedMetaStrategyAlt_2 \in \Delta(\FinitePMS_2) \mapsto \utility_\pl(\pureStrategyAlt_1, \mixedMetaStrategyAlt_2)$
        for any $\pureStrategyAlt_1 \in \PureStrategies_1^\game$,
        we have
            $
                \utility_\pl(\pureStrategyAlt_1, \mixedMetaStrategy_2)
                =
                \utility_\pl(\pureStrategyAlt_1, \mixedStrategy_2)
            $
        for both $\pl$.
    To show \eqref{eq:threeB_plone_sim},
        note that both $\mixedStrategy_2$ and all $\mixedStrategyAlt_2^j$ have $\pureStrategy_1$ as a best response.
            (This holds because any \textit{favourable} best response is in particular a best response,
                and, unlike $\inverseFBR(\pureStrategy_1)$,
                the set $\inverseBR(\pureStrategy_1)$ is closed.)
        This implies that
            \begin{align*}
                \utility_1(\mSim, \mixedMetaStrategy_2)
                &
                =
                \sum\nolimits_j
                    \lambda_j \utility_1(\mSim, \meta{\mixedStrategyAlt}_2^j)
                \\
                &
                =
                \sum\nolimits_j
                    \lambda_j \utility_1(\br^\game, \mixedStrategyAlt_2^j)
                \\
                &
                =
                \sum\nolimits_j
                    \lambda_j \utility_1(\pureStrategy_1, \mixedStrategyAlt_2^j)
                \\
                &
                =
                \utility_1(\pureStrategy_1, \sum\nolimits_j \lambda_j \mixedStrategyAlt_2^j)
                \\
                &
                =
                \utility_1(\pureStrategy_1, \mixedStrategy_2)
                =
                \utility_1(\br^\game, \mixedStrategy_2)
                =
                \utility_1(\mSim, \mixedStrategy_2)
                .
            \end{align*}
    The proof of \eqref{eq:threeB_plone_sim} is analogous,
        except that \Pltwo{}'s utility might sometimes increase,
        because $\pureStrategy_1$ might not necessarily be the \textit{favourable} best response to all $\mixedStrategyAlt_2^j$:
        \begin{align*}
            \utility_2(\mSim, \mixedMetaStrategy_2)
            &
            =
            \sum\nolimits_j
                \lambda_j \utility_2(\mSim, \meta{\mixedStrategyAlt}_2^j)
            \\
            &
            =
            \sum\nolimits_j
                \lambda_j \utility_2(\fbr, \mixedStrategyAlt_2^j)
            \\
            &
            \geq
            \sum\nolimits_j
                \lambda_j \utility_2(\pureStrategy_1, \mixedStrategyAlt_2^j)
            \\
            &
            =
            \utility_2(\pureStrategy_1, \sum\nolimits_j \lambda_j \mixedStrategyAlt_2^j)
            \\
            &
            =
            \utility_2(\pureStrategy_1, \mixedStrategy_2)
            \\
            &
            =
            \utility_2(\fbr, \mixedStrategy_2)
            =
            \utility_2(\mSim, \mixedStrategy_2)
            .
        \end{align*}
    This concludes the proof of \Cref{cl:finiteMtwo}.
    \end{proof}
    
    With \Cref{cl:finiteMtwo}, the proof of \Cref{prop:finite_str_space} is now complete.
\end{proof}

\vk{optional TODO: In theory, we could also formally show the claim that
        ``in $\psimgame$, WLOG, \Pltwo{} uses pure strategies only.
    This should be quite easy with the lemmas we have prepared. But, like, whatever, probably not very important?}

\section{Proofs \texorpdfstring{for \Cref{sec:computational} (Complexity Results)}{of Complexity Results}}\label{sec:app:complexity}

In this section, we present the proofs related to the complexity results given in this paper.

\solvingUpperBound*
\begin{proof}
    We rely on the reduction from $\msimgame$ to $\msimgame'$ from \Cref{prop:finite_str_space}.
    Since
        any $\msimgame'$ that satisfies the conclusion of \Cref{prop:finite_str_space} can be used for solving $\msimgame$,
        all that remains is to estimate the size of $\msimgame'$.
    Recall that by ``solving'' a game, we mean any of:  (a) finding one \NE{};
    (b) finding an \NE{} that maximises social welfare or the utility of one of the players;
    or (c) finding all \NE{} payoff profiles and some \NE{} corresponding to each.
    
    The proof of \Cref{prop:finite_str_space} shows that the reduction holds for $\msimgame' = (\PureMetaStrategies'_1, \PureMetaStrategies'_2, u)$ where:
    \begin{align*}
        \PureMetaStrategies'_1 := 
        & \ \PureStrategies_1 \cup \left\{ \mSim \right\},\\
        \PureMetaStrategies'_2 :=
        &\ \left\{
                    \mixedStrategy_2 \in \MixedStrategies_2
                \mid
                    \exists \pureMetaStrategy_1 \in \PureStrategies_1 \cup \{ \mSim \}
                    :
                    \mixedStrategy_2 \in \br(\pureMetaStrategy_1)
            \right\}\\
        & \ \ \cap
            \bigcup_{\pureStrategy_1 \in \PureStrategies_1}
                \extremalPoints(
                    \closure{\inverseFBR(\pureStrategy_1)}
                )
    \end{align*}
    where $\closure{Y}$ denotes the closure of $Y$ and $\extremalPoints(Z)$ denotes the set of all extremal points of $Z$. 
    As $\inverseFBR(\pureStrategy_1) \subseteq \inverseBR(\pureStrategy_1)$ for any $\pureStrategy_1 \in \PureStrategies_1$, then clearly the reduction also holds for $\msimgame'' = (\PureMetaStrategies'_1, \PureMetaStrategies''_2, u)$ where:
    \begin{align*}
        \PureMetaStrategies''_2 :=
        \bigcup_{\pureStrategy_1 \in \PureStrategies_1}
                \extremalPoints(
                    \closure{\inverseBR(\pureStrategy_1)}
                ).
    \end{align*}
    Given that $\vert \PureMetaStrategies'_1 \vert = \vert \PureStrategies_1 \vert + 1$, then to conclude the proof we need only bound $\vert \PureMetaStrategies''_2 \vert$.
    To do so, note that if $\mixedStrategy_2 \notin \br(\pureStrategy_1)$ then there is some $\pureStrategy_2$ such that $\utility_2(\pureStrategy_1, \pureStrategy_2) > \utility_2(\pureStrategy_1, \mixedStrategy_2)$.
    Thus, for a given strategy $\pureStrategy_1$, the set $\inverseBR(\pureStrategy_1)$ can be expressed in terms of the following set of linear inequalities:
    \begin{align*}
        \utility_2(\pureStrategy_1, \mixedStrategy_2) &\geq \utility_2(\pureStrategy_1, \pureStrategy_2) &&\forall \pureStrategy_2 \in \PureStrategies_2\\
        \mixedStrategy_2(\pureStrategy_2) &\geq 0 &&\forall \pureStrategy_2 \in \PureStrategies_2\\
        \sum_{\pureStrategy_2 \in \PureStrategies_2} \mixedStrategy_2(\pureStrategy_2) &= 1
    \end{align*}
    over the variables $\left\{ \mixedStrategy_2(\pureStrategy_2) \right\}_{\pureStrategy_2 \in \PureStrategies_2}$.
    The final equality reduces the solution of this problem to a feasible region defined by $2 \cdot \vert \PureStrategies_2 \vert$ constraints over a $(\vert \PureStrategies_2 \vert - 1)$-dimensional space.
    As a vertex of the feasible region is defined by $\vert \PureStrategies_2 \vert - 1$ hyperplanes (i.e. constraints), then the number of extremal points in the closure of $\inverseBR(\pureStrategy_1)$ is at most $\binom{2 \cdot \vert \PureStrategies_2 \vert}{\vert \PureStrategies_2 \vert - 1}$.
    Summing over $\pureStrategy_1 \in \PureStrategies_1$ we see that $\vert \PureMetaStrategies''_2 \vert \leq 
    \vert \PureStrategies_1 \vert \cdot \binom{2 \cdot \vert \PureStrategies_2 \vert}{\vert \PureStrategies_2 \vert - 1}$, and hence that $\msimgame''$ has size at most:
    $$(\vert \PureStrategies_1 \vert + 1) \cdot \vert \PureStrategies_1 \vert \cdot  \binom{2 \cdot \vert \PureStrategies_2 \vert}{\vert \PureStrategies_2 \vert - 1} = 
    O(\vert \PureStrategies_1 \vert^2 \cdot 2^{\vert \PureStrategies_2 \vert}).$$
\end{proof}

Recall that an (undirected) \textbf{graph} is a pair $\graph = (V, E)$, where $V$ is a set of \textbf{vertices} and $E \subseteq V \times V$ is a (symmetric) set of \textbf{edges}.
A \textbf{bipartite} graph is a graph of the form $\graph = (A \cup B, E)$
    such that edges only occur between vertices in $A$ and vertices in $B$
    (i.e., $A \cap B = \emptyset$ and for every $(v, v') \in E$, we have either $v \in A$, $v' \in B$ or $v \in B$, $v' \in A$).
A bipartite graph is \textbf{complete} when we have $(a, b) \in E$ for every $a \in A$, $b \in B$.
By $K_{k, l}$, we denote a complete bipartite graph with $|A| = k$, $|B| = l$.

\helpsIsNPHtheorem*

\begin{proof}
    To prove the result, we reduce from the problem \textsc{Complete Bipartite Subgraph (CBS)}, which is \NPH{} \cite{balanced-complete-bipartite-subgraph}. 
    It consists of deciding, for a given a bipartite graph $\graph = (A \cup B, E)$ and a parameter $k$, whether $\graph$ contains a complete bipartite subgraph $K_{k, k}$
        (i.e., with each partite set having $k$ vertices). 
    The reduction (specifically the construction of $\game'$) is very similar to that of Theorem 3.3 of \citet{sauerberg2024computing}.

    For the purpose of the reduction,
        let $\graph = (A \cup B, E)$ be a bipartite graph
        and let $k$ be the  parameter of the \textsc{CBS} instance.
    Our final construction will be a game $\game$, which will contain three subgames $\game'$, $\game^1$, and $\game^2$.
    For each player, the action set in $\game'$ is of the form
        $
            \actions_\pl
            :=
            A \cup B \cup \{ \OO \}
        $
        -- i.e., they have one action for each vertex and one additional \OptOut{} action.
    The corresponding utilities in $\game'$, summarised in \Cref{fig:hardness_graph_subgame} (top),
        are as follows:
    \begin{itemize}
        \item $u(\OO, \OO) = (0,0)$;
        \item $u(\OO{}, \cdot{}) = (1, -1)$ when $\cdot{} \neq \OO$;
        \item $u(\cdot{}, \OO{}) = (-1 ,1)$ when $\cdot{} \neq \OO$;
        \item For $a \in A$, $b\in B$,
            \begin{align*}
                u(a, b) =  \begin{cases}
                (1, 1) \text{ if } (a, b) \in E \\
                (0, 0) \text{ otherwise}
                ;
                \end{cases}
            \end{align*}
        \item For $a, a' \in A$,
            \begin{align*}
                u(a, a')
                =  \begin{cases}
                    (-k, k) \text{ if } a=a' \\
                    (0, 0) \text{ otherwise}
                    ;
                \end{cases}
            \end{align*}
        \item For $b, b' \in B$,
            \begin{align*}
                u(b, b')
                =  \begin{cases}
                    (k, -k) \text{ if } b=b' \\
                    (0, 0) \text{ otherwise}
                    ;
                \end{cases}
            \end{align*}
        \item For $b \in B$, $a \in A$, $u(b, a) = (0,0)$.
    \end{itemize}

    \noindent
    We refer to $A$ as \Plone{}'s partite set and $B$ as \Pltwo{}'s partite set. 
    Intuitively, the players benefit (payoffs of $(1,1)$) if they each play a vertex in their own partite set and their vertices are adjacent. 
    However, playing any vertex $v$ in one's own partite set with probability more than $1/k$
        makes the player vulnerable to \enquote{exploitation} by the other player,
        who can also play $v$ and receive a payoff greater than $1$.

    Let $\game^1$ and $\game^2$ be copies of the 2x2 trust game from \Cref{fig:TG_for_NPH_result} (bottom left), with pure strategy spaces $(\T^i, \WO^i)$ and $(\C^i, \D^i)$.
        
    Finally, let $\game$
        be a \enquote{coordination} game where the players' strategy spaces are the unions of those in $\game'$, $\game^1$, and $\game^2$.
    If they play in the same of the three subgames, their payoffs are given by the payoffs in the subgame; otherwise they receive payoffs $(0,0)$.

    Informally, we claim than $\mSim$ \enquote{helps} in $\game$ if and only if $\graph$ has a $K_{k, k}$ subgraph.
    Formally, we will show the following claims. 
    \begin{enumerate}
        \item If $\graph$ contains a $K_{k, k}$ subgraph, then
        \begin{enumerate}
            \item $\game$ admits an \NE{} with payoffs $(1,1)$;
            \item $\msimgame$ admits no \NE{} where $\utility_1 > 1$ or $\utility_2 > 1$.
        \end{enumerate}
        \item If $\graph$ does not contain a $K_{k, k}$ subgraph, then 
        \begin{enumerate}
            \item $\game$ admits no \NE{} where $\utility_1 > 0$ or $\utility_2 > 0$, but 
            \item $\msimgame$ admits an \NE{} with payoffs $(1-\simcost, 1)$.
        \end{enumerate}
    \end{enumerate}

    \noindent
    Together, these imply that
        $\graph$ contains a $K_{k,k}$ subgraph
        if and only if
        enabling \mSim{} in $\game$ introduces a \NE{} that strictly Pareto-improves over all \NE{} of $\game$.
    Given the specific payoffs involved in claims (1-2), this is equivalent to the new \NE{} constituting a strict improvement over the Nash equilibria of $\game$ in each of the senses (a-e) considered in the theorem.
    In other words, this shows that \textsc{Complete Bipartite Subgraph} can be reduced to each of the problems $P_\textnormal{a}$, \dots, $P_\textnormal{e}$.
    To finish the proof of \Cref{thm:helps_is_NPH_v3}, it thus remains to prove the claims (1) and (2).

    \begin{claim*}[1a]
        If $\graph$ contains a $K_{k, k}$ subgraph, then
            $\game$ admits an \NE{} with payoffs $(1,1)$.
    \end{claim*}

    \begin{proof}[Proof of Claim (1a)]
        Let $(A', B')$ be a $K_{k, k}$ subgraph in $\graph$.
        We claim it is an \NE{} for \Plone{} to mix uniformly over $A'$ and for \Pltwo{} to mix uniformly over $B'$. 
        This gives payoffs of $(1,1)$ because the subgraph is complete. 
    
        Deviating to any vertex within a player's own partite set gives utility at most $1$.
        Deviating to any vertex $v$ in the other player's partite set gives utility at most $(1/k) \cdot k + (1-1/k) \cdot 0 = 1$ because the other player plays $v$ with probability at most $1/k$. 
        Deviating to \OO{} gives utility $1$, and deviating to a different subgame gives utility $0$.
    \end{proof}

    Next, we prove (1b).
    Note that this part of (1) does not rely on the assumption that $\graph$ contains a $K_{k, k}$ subgraph.

    \begin{claim*}[1b]
        $\msimgame$ admits no \NE{} where $\utility_1 > 1$ or $\utility_2 > 1$.
    \end{claim*}
    \begin{proof}[Proof of Claim (1b)]
    Intuitively, this claim holds because
        in any strategy profile where one of the players had $\utility_\pl > 1$,
        the other player would be better off deviating to either $\OO$ (in $\game'$) or $\WO$ (in $\game^1$ or $\game^2$).
    We now show this formally.

    First, suppose that \Plone{} achieves $\utility_1(\mu) > 1$ for some equilibrium strategy $\mu \in \NE(\msimgame)$.
    Given the definition of $\game$, the only actions which give \Plone{} utility over $1$ are vertices $b \in B$
        for which \Pltwo{} selects $b$ with probability greater than $1/k$.
    This means that $\supp(\mu_1)$ must be a subset of $\{ b \in B \mid \mesa{\mu}_2(b) > 1/k \} \cup \{ \mSim \}$.
        (Where $\mesa{\mu}_2(y)$ denotes the overall probability that \Pltwo{} puts on $y$.)
    Suppose that $\mixedStrategy_2 \in \supp(\mu_2)$ is a strategy for which $\mixedStrategy_2(b_0) > 1/k$ holds for some $b_0 \in B$.
    Note that $\utility_2(\mSim, \mixedStrategy_2) < 0$
        (since $\br(\mixedStrategy_2) = \{b' | \mixedStrategy_2(b') = \max_{b \in B}  \mixedStrategy_2(b) \}$ and $\max_{b \in B}  \mixedStrategy_2(b) \geq \mixedStrategy_2(b_0) > 1/k$ ).
    As a result, we have $\utility_2(x, \mixedStrategy_2) < 0$ for every $x \in \supp(\mu_1)$.
    However, this means that \Pltwo{} could strictly increase their utility by replacing $\mixedStrategy_2$ by $\OO$
        (which yields $\utility_2(x, \OO) \geq 0$ for every $x \in B \cup \{ \mSim \}$).
    Since this would contradict the assumption that $\mu$ is a \NE{},
        it follows \Plone{} must not be able to achieve $\utility_1(\mu) > 1$ in any \NE{} of $\msimgame$.

    We now show that \Pltwo{}'s utility cannot exceed $1$ in any \NE{} of $\msimgame$. 
    Since the only outcomes where \Pltwo{}'s payoff might exceed $1$ are the $(a, a)$ and $(T^i, D^i)$ outcomes,
        it suffices to show that these cannot occur in any \NE{} of $\msimgame$. 

    First, we prove that the outcomes $(\T^i, \D^i)$ can never occur with positive probability.
    Suppose that $\mu_2$ puts non-zero probability on $\D^i$.
    In theory, the outcome $(\T^i, \D^i)$ could occur
        either directly (i.e., from \Plone{} using a strategy which puts positive probability on $\T^i$)
        or indirectly (i.e., from \Plone{} playing $\mSim$ and $\T^i$ being a best response to \Pltwo{}'s strategy).
    However, by the definition of the trust game $\game^i$, it follows that $\WO^i$ gives \Plone{} strictly higher utility than $\T^i$ in response to any strategy for Player $2$ (technically, any $\mesa{\mu}_2$ or $\mixedStrategy_2$) that puts positive probability on $\D^i$.
    This rules out both possibilities for the occurrence of the outcome $(\T^i, \D^i)$.

    Second, we prove that an outcome $(a, a)$ can never occur with positive probability.
    Suppose \Plone{} plays $a$ directly while $\mesa{\mu}_2(a) > 0$.
    We argue that $a$ is strictly dominated by $\OO$ against $\mu_2$:
        when \Pltwo{} plays in $\game^1$ or $\game^2$, $a$ and $\OO$ give the same utilities;
        when \Pltwo{} plays some $b \in B$, we have $\utility_1(a, b) \leq 1 = \utility_1(\OO, b)$;
        and finally when \Pltwo{} plays some $a' \in A$
            -- which happens with positive probability --
            we have $\utility_1(a, a') \leq 0 < 1 = \utility_1(\OO, a')$.
    This means that in an \NE{}, $(a, a)$ cannot be played directly
        (i.e., without \Plone{} using $\mSim$).
    The same argument shows that $a$ cannot be a best response to any $\mixedStrategy_2$ that puts positive probability on $a$,
        so the outcome $(a,a)$ cannot occur indirectly (i.e., after \Plone{} playing $\mSim$) either.
    \end{proof}

    \input{NPH_game_figures}

    \begin{claim*}[2a]
        If $\graph$ does not contain a $K_{k, k}$ subgraph, then 
            $\game$ admits no \NE{} where $\utility_1 > 0$ or $\utility_2 > 0$.
    \end{claim*}
    
    \begin{proof}[Proof of Claim (2a)]
    First, note that when players miscoordinate and play different subgames, they receive payoffs $(0,0)$.
    Payoffs higher than $0$ must thus come from playing in the same subgame.
    Moreover, when the players both play some subgame with positive probability, they must play a Nash equilibrium in that subgame.
    Payoffs higher than $0$ must thus come from playing a Nash equilibrium of one of the subgames.
    Since the only \NE{} of the $\game^i$ subgames are $(\WO^i, \D^i)$, which give payoffs $(0,0)$,
        the only remaining hope is the subgame $\game'$.
    To prove (2a)
        -- i.e., the impossibility of \NE{} payoffs higher than $0$ --
        it thus suffices to show that
            if $\graph$ does not contain a $K_{k, k}$ subgraph,
            the only \NE{} of $\game'$ is $(\OO, \OO)$
                (which gives payoffs $(0,0)$).

    Before proceeding with the proof of this claim, we observe that
        in $\game'$, no \NE{} can involve \Plone{} playing $b \in B$ or \Pltwo{} playing $a \in A$.
    To prove this claim for \Plone{}, first observe that for \Pltwo{}, $\OO$ weakly dominates $b$
        -- and it strictly dominates $b$ when $\mixedStrategy_1(b) > 0$ --
        which means that \Pltwo{} cannot put positive probability (in an equilibrium)
        on any $b \in B$ for which $\mixedStrategy_1(b) > 0$.
    For contradiction, suppose that we had $\mixedStrategy_1(b) > 0$ for some $\mixedStrategy \in \NE(\game')$.
    Since $\mixedStrategy$ is a \NE{}, we must have $\utility_1(b, \mixedStrategy_2) \geq \utility_1(\OO, \mixedStrategy_2)$.
    Inspecting the payoffs of $\game'$, we see that
        \Pltwo{}'s only action for which $\utility_1(b, \symbolPlaceholder)$ is not strictly lower than $\utility_1(\OO, \symbolPlaceholder)$ is $b$.
    This means that for \Plone{} to be able to have $\mixedStrategy_1(b) > 0$ in equilibrium,
        we must also have $\mixedStrategy_2(b) > 0$
        -- a contradiction with the observation above.
    By applying the same argument with the roles of the players switched,
        we obtain that in \NE{} of $\game'$, \Pltwo{} never plays any $a \in A$.

    We now return to the proof of (2a).
    For contradiction, suppose that there is some $\mixedStrategy \in \NE(\game')$ such that $\mixedStrategy \neq (\OO, \OO)$.
    Without loss of generality, suppose that $\mixedStrategy_1(\OO) < 1$.
    Applying the observation above,
        we have $\mixedStrategy_1(b) = 0$ for every $b \in B$,
        which implies that there must be some $a_0 \in A$ for which $\mixedStrategy_1(a_0) > 0$.
    Since $\mixedStrategy$ is an equilibrium, we must have $\utility_1(a_0, \mixedStrategy_2) \geq \utility_1(\OO, \mixedStrategy_2)$.
        However, looking at the payoffs of $\game'$ (\Cref{fig:hardness_graph_subgame}),
            we see that for \Plone{}, the only case where $a_0$ does not give strictly lower utility for \Plone{} than $\OO$
            is when $\mixedStrategy_2(\OO) = 0$.
    This means there must be some $b_0 \in B$ for which $\mixedStrategy_2(b_0) > 0$.
    We can now repeat the same argument to obtain that $\mixedStrategy_1(\OO) = 0$.
    
    From the paragraph before previous one, it follows that
        for any $\mixedStrategy \in \NE(\game')$ that is different from $(\OO, \OO)$,
        \Plone{} only uses actions from $A$ and \Pltwo{} only uses actions from $B$.
    Moreover, any vertices $a$, $b$ with $\mixedStrategy_1(a) > 0$, $\mixedStrategy_2(b) > 0$,
        must be adjacent.
        (This holds because
            for $\mixedStrategy$ as above, we have $\utility_1(\OO, \mixedStrategy_2) = \utility_2(\mixedStrategy_1, \OO) = 1$.
            This means that $\utility_1(\mixedStrategy)$ and $\utility_2(\mixedStrategy)$ must both be equal to $1$.
            By definition of $\game'$, this is only possible when the selected vertices satisfy $(a, b) \in E$.)
    Additionally, neither player can put more than $1/k$ probability on any of their vertices $v$.
        (Otherwise the other player could strictly increase their utility by also playing $v$.)
    This means that each player must randomise between at least $k$ vertices.
    Taken all together, this implies that
        there must be sets $A' \subseteq A$, $B' \subseteq B$ with $|A'|, |B'| \geq k$
        such that $(a', b') \in E$ for every $a' \in A'$, $b' \in B'$.
    However, such set $A' \cup B'$ would then give a complete subgraph $K_{m,n}$ for $m, n \geq k$,
        contradicting the assumption that $\graph$ contains no $K_{k,k}$ subgraph.

    This concludes the proof of (2a).
    \end{proof}

    We now prove the claim (2b).
    Note that this part of (2) does not rely on the assumption that $\graph$ does not contain a $K_{k, k}$ subgraph.

    \begin{claim*}[2b]
        $\msimgame$ admits an \NE{} with payoffs $(1-\simcost, 1)$.
    \end{claim*}
    
    \begin{proof}[Proof of Claim (2b)]
        We claim that it is a \NE{} of $\msimgame$
            for \Plone{} to always play \mSim{} and \Pltwo{} to mix 50:50 between
                playing $\C^1$ in $\game^1$ and playing $\C^2$ in $\game^2$.
        This strategy profile gives payoffs $(1-\simcost, 1)$ and so suffices to show the claim.

        \Plone{} clearly has no profitable deviations
            -- blindly (i.e., without simulating first) playing $\T^1$, $\T^2$, $\WO^1$, or $\WO^2$ gives expected utility $1/2$,
            blindly playing in $\game'$ gives utility $0$.

        To see that \Pltwo{} has no profitable deviations,
            note that that the only outcomes with $\utility_2 > 1$ are $(\T^i, \D^i)$ and $(a, a)$ for any $a \in A$,
            so any profitable deviation must result in such outcomes with positive probability.
        However, $\T^i$ is not a best response to any $\mixedStrategy_2$ that plays $\D^i$ with positive probability
            (since $\WO^i$ strictly dominates $\T^i$ against such strategies).
        Similarly, $a$ is not a best response to any $\mixedStrategy_2$ that plays $a$ with positive probability.
            (As $\OO$ strictly dominates $a$ against such strategies.)
        Hence, no strategy for \Pltwo{} gives utility strictly greater than $1$ against \mSim,
            so \Pltwo{} has no profitable deviations.
    \end{proof}

    This concludes the proof of \Cref{thm:helps_is_NPH_v3}.
\end{proof}

\section{Proof \texorpdfstring{of \Cref{thm:negative} (Overly Informed \Pltwo{})}{that Simulation Does Not Help Against an Overly-Informed Opponent}}

In this section, we formally prove the main negative result of this text, \Cref{thm:negative}.
Claims~\ref{cl:perfect_info_new_strategy_improves}-\ref{cl:perfect_info_new_strategy_improves_for_sim} and their proofs are a part of the proof of this theorem.

\thmNegativeResult*

\begin{proof}
    Let $\game_0$ and $\game$ be as in the statement of the theorem.
    We will heavily rely on the assumption that \Pltwo{}'s responses must be Pareto optimal
        --- in other words, that for every $\pureStrategy_1 \in \PureStrategies^{\game_0}_1$,
        \Plone{} only considers a restricted set of responses,
            which we denote $\allowedResponses(\pureStrategy_1) \subset \PureStrategies^{\game_0}_2$,
        for which we have
        \begin{align*}
                \utility_2(\pureStrategy_1, \pureStrategy_2)
                \geq
                \utility_2(\pureStrategy_1, \pureStrategyAlt_2)
            \iff
                \utility_1(\pureStrategy_1, \pureStrategy_2)
                \leq
                \utility_1(\pureStrategy_1, \pureStrategyAlt_2)
        \end{align*}
        for every $\pureStrategy_2, \pureStrategyAlt_2 \in \allowedResponses(\pureStrategy_1)$.
        Note that this implies that
        \begin{align*}
            \utility_2(\pureStrategy_1, \pureStrategy_2)
                >
                \utility_2(\pureStrategy_1, \pureStrategyAlt_2)
            \iff
            \utility_1(\pureStrategy_1, \pureStrategy_2)
                <
                \utility_1(\pureStrategy_1, \pureStrategyAlt_2)
            .
        \end{align*}

    We will use the following auxiliary notation.
        By $\maxminV$, we denote the maxmin value of $\game$ for \Plone{}:
        \begin{align*}
            \maxminV
            :=
                \max_{\pureStrategy_1 \in \PureStrategies^{\game_0}_1}
                \left(
                    \min_{\pureStrategy_2 \in \PureStrategies^{\game_0}_2}
                    \utility_1(\pureStrategy_1, \pureStrategy_2)
                \right)
            .
        \end{align*}
        For every pure strategy $\pureStrategy_1$, we pick some best response
            \begin{align*}
                \optimalResponse{\pureStrategy_1}
                \in
                \argmax \{
                    \utility_2(\pureStrategy_1, \pureStrategy_2)
                    \mid
                    \pureStrategy_2 \in \allowedResponses(\pureStrategy_1)
                \}
                .
            \end{align*}
        Since \Pltwo{} can only use Pareto-optimal responses, we have
            \begin{align*}
                \utility_2(\pureStrategy_1, \optimalResponse{\pureStrategy_1} )
                =
                \max \, \left\{
                    \utility_2(\pureStrategy_1, \pureStrategy_2)
                    \mid
                    \pureStrategy_2 \in \allowedResponses(\pureStrategy_1)
                \right\}
            \end{align*}
            and
            \begin{align*}
                \utility_1(\pureStrategy_1, \optimalResponse{\pureStrategy_1} )
                =
                \min \, \left\{
                    \utility_1(\pureStrategy_1, \pureStrategy_2)
                    \mid
                    \pureStrategy_2 \in \allowedResponses(\pureStrategy_1)
                \right\}
                \leq
                \maxminV
                .
            \end{align*}
        When $\utility_1(\pureStrategy_1, \optimalResponse{\pureStrategy_1} ) \geq \maxminV$,
        we denote by
        $
            \optimalCommitment{\pureStrategy_1}
        $
        some
        \begin{align*}
            \optimalCommitment{\pureStrategy_1}
            \in
            \argmax \{
                \,
                \utility_2 (\pureStrategy_1, \mixedStrategy_2)
                \ | \ 
                &
                    \mixedStrategy_2 \in \Delta(\allowedResponses(\pureStrategy_1)),
                    \\
                & \phantom{\mixedStrategy_2 \in }
                    \utility_1(\pureStrategy_1, \mixedStrategy_2) \geq \maxminV
                \,
            \}
        \end{align*}
        -- that is, an optimal mixed response to $\pureStrategy_1$ which still leaves \Plone{} at least as well off as their maxmin value.
        (Note that
            the maximum is actually attained when the utility is \textit{equal} to $\maxminV$.
            Otherwise, we would arrive at a contradiction with the definition of $\maxminV$.)

    Because of the assumption that \Pltwo{} does not use Pareto-dominated strategies,
    the only Nash equilibria of the game $\game$ -- i.e., without simulation -- are those where
        \Plone{} plays some $\maxminStrategy$ for which
            $
                \utility_1(\maxminStrategy, \optimalResponse{\maxminStrategy})
                =
                \max \, \left\{
                    \utility_1(\pureStrategy_1, \optimalResponse{\pureStrategy_1})
                    \mid
                    \pureStrategy_1 \in \PureStrategies^{\game_0}_1
                \right\}
            $
        and \Pltwo{} responds by $\optimalResponse{\maxminStrategy}$,
        resulting in utility $\utility_1 ( \pureStrategy_1, \optimalResponse{\pureStrategy_1}) = \maxminV$ for \Plone{}.
    To prove the theorem,
        we will show that for any $\mixedMetaStrategy_1 \in \MixedStrategies^{\msimgame}_1$ and any best-response $\mixedMetaStrategy_2$ in $\msimgame$,
        \Plone{}'s utility satisfies
        $
            \utility_1(\mixedMetaStrategy_1, \mixedMetaStrategy_2)
            \leq \maxminV - \mixedMetaStrategy_1(\mSim) \cdot \simcost
        $.
    Since \Plone{} can always achieve $\maxminV$ utility by playing $\maxminStrategy$,
        this will show that
        there cannot exist a Nash equilibrium of $\msimgame$ with $\mixedMetaStrategy_1(\mSim) > 0$.

    Let
        $\mixedMetaStrategy_1 = \simProb \cdot \mSim + (1-\simProb) \cdot \mixedStrategy_1$,
            $\mixedStrategy_1 \in \MixedStrategies^{\game_0}_1$,
            be a mixed strategy of \Plone{} in $\msimgame$
    and let
        $
            \responseStrategy :
                \PureStrategies^{\game_0}_1 \to \Delta(\PureStrategies^{\game_0}_2)
        $
        be a response strategy of \Pltwo{}
            (i.e., a pure-strategy of \Pltwo{} in $\msimgame$).
    First, if $\simProb = 0$,
        then \Pltwo{} can improve (not necessarily strictly) their utility
        by switching from $\responseStrategy$ to
        $
            \responseStrategyAlt : \pureStrategy_1 \mapsto \optimalResponse{\pureStrategy_1}
        $,
        bringing \Plone{}'s utilities $\utility_1(\pureStrategy_1, \responseStrategyAlt)$ down to $\maxminV$ or less.
    By the assumption of \Pltwo{} only using Pareto-optimal strategies,
        the only way for this improvement of \Pltwo{}'s utility to \textit{not} be strict
        is if we already had
        $
            \utility_1(\pureStrategy_1, \responseStrategy(\pureStrategy_1))
            =
            \utility_1(\pureStrategy_1, \optimalResponse{\pureStrategy_1})
            \leq
            \maxminV
        $
        for every $\pureStrategy_1$ with $\mixedStrategy_1(\pureStrategy_1) > 0$.
    
    Second, when $\simProb > 0$,
        \Pltwo{} can weakly improve their utility in $\msimgame$
        by switching from $\responseStrategy$ to the response strategy $\responseStrategyAlt$ defined as follows:
    Denote by $\pureStrategy_1^0$ some friendly best response of \Plone{} to $\responseStrategy$.
    Set
        $\responseStrategyAlt(\pureStrategy_1^0) := \optimalCommitment{\pureStrategy_1^0}$.
        For $\pureStrategy_1 \neq \pureStrategy_1^0$, set
            $\responseStrategyAlt(\pureStrategy_1) := \optimalResponse{\pureStrategy_1}$.
    
    \begin{claim}\label{cl:perfect_info_new_strategy_improves}
        For every $\pureStrategy_1 \in \PureStrategies^{\game_0}_1$, we have
        \begin{align*}
            \utility_2(\pureStrategy_1, \responseStrategyAlt(\pureStrategy_1))
            \geq
            \utility_2(\pureStrategy_1, \responseStrategy(\pureStrategy_1))
            .
        \end{align*}
    \end{claim}

    \noindent
    (\Cref{cl:perfect_info_new_strategy_improves} holds trivially for every $\pureStrategy_1$ with the exception of $\pureStrategy_1^0$.
    For $\pureStrategy_1^0$, suppose we had
        \begin{align*}
            \utility_2(\pureStrategy_1^0, \responseStrategyAlt(\pureStrategy_1^0))
            <
            \utility_2(\pureStrategy_1^0, \responseStrategy(\pureStrategy_1^0))
            .
        \end{align*}
        By the Pareto-optimality assumption, this is equivalent to
        \begin{align*}
            \utility_1(\pureStrategy_1^0, \responseStrategy(\pureStrategy_1^0))
            <
            \utility_1(\pureStrategy_1^0, \responseStrategyAlt(\pureStrategy_1^0))
            .
        \end{align*}
        However,
            by definition of $\responseStrategyAlt$ and $\optimalCommitment{\symbolPlaceholder}$,
            the right-hand side of this inequality is equal to the maxmin value for \Plone:
            \begin{align*}
                \utility_1(\pureStrategy_1^0, \responseStrategy(\pureStrategy_1^0))
                <
                \utility_1(\pureStrategy_1^0, \responseStrategyAlt(\pureStrategy_1^0))
                =
                \utility_1(\pureStrategy_1^0, \optimalCommitment{\pureStrategy_1^0})
                =
                \maxminV
                .
            \end{align*}
        Since
            $\utility_1(\br, \responseStrategy)$ cannot be lower than the maxmin value of the game, for any strategy of \Pltwo{},
            this would contradict the assumption that $\pureStrategy_1^0$ is a (friendly) best response to $\responseStrategy$.)

    \begin{claim}\label{cl:perfect_info_new_strategy_improves_for_sim}
        We have
        $
            \utility_2(\mSim, \responseStrategyAlt)
            \geq
            \utility_2(\mSim, \responseStrategy)
        $.
    \end{claim}

    \noindent
    (To see this, note that since $\pureStrategy_1^0$ is an opponent-friendly best response to $\responseStrategy$, we have
        $
            \utility_2(\mSim, \responseStrategy)
            =
            \utility_2(\pureStrategy_1^0, \responseStrategy(\pureStrategy_1^0))
        $.
    By \Cref{cl:perfect_info_new_strategy_improves}, we have
        % \begin{align*}
        $
            \utility_2(\pureStrategy_1^0, \responseStrategy(\pureStrategy_1^0))
            \leq
            \utility_2(\pureStrategy_1^0, \responseStrategyAlt(\pureStrategy_1^0))
            .
        $
        % \end{align*}
    By definition of $\responseStrategyAlt$, we have
        $
            \utility_1(\pureStrategy_1, \responseStrategyAlt(\pureStrategy_1))
            =
            \utility_1(\pureStrategy_1, \optimalCommitment{\pureStrategy_1})
            \leq \maxminV
        $
        for every $\pureStrategy_1 \neq \pureStrategy_1^0$
        and
        $
            \utility_1(\pureStrategy_1^0, \responseStrategyAlt(\pureStrategy_1^0))
            =
            \maxminV
        $.
    This shows that $\pureStrategy_1^0$ is \textit{a} best response to $\responseStrategyAlt$,
        even though it might not be the most opponent-friendly one.
    As a result, we have
        $
            \utility_2(\pureStrategy_1^0, \responseStrategyAlt(\pureStrategy_1^0))
            \leq
            \utility_2(\fbr, \responseStrategyAlt)
            =
            \utility_2(\mSim, \responseStrategyAlt)
        $.
        Putting all of these inequalities together gives \Cref{cl:perfect_info_new_strategy_improves_for_sim}.)

    By definition of $\responseStrategyAlt$, we get that
        \begin{align*}
            \utility_1(\mixedStrategy_1, \responseStrategyAlt) \leq \maxminV,
            \ \ 
            \utility_1(\mSim, \responseStrategyAlt) \leq \maxminV - \simcost
            .
        \end{align*}
    Moreover, the only way to have
        $
            \utility_2^{\msimgame}(\mixedMetaStrategy_1, \responseStrategy)
            =
            \utility_2^{\msimgame}(\mixedMetaStrategy_1, \responseStrategyAlt)
        $
    would be for these inequalities to hold already for $\responseStrategy$.
    In summary, this shows that
    \begin{claim}
        For any $\mixedMetaStrategy_1 \in \MixedStrategies^{\msimgame}_1$ and $\responseStrategy \in \br(\mixedMetaStrategy_1)$, we have
        \begin{align*}
            \utility^{\msimgame}_1(\mixedMetaStrategy_1, \responseStrategy)
            \leq
            \maxminV - \mixedMetaStrategy_1(\mSim) \cdot \simcost
            .
        \end{align*}
    \end{claim}

    Since $\utility_1(\maxminStrategy, \responseStrategy(\maxminStrategy) \geq \maxminV$ for every $\responseStrategy$,
    this implies that $\msimgame$ admits no Nash equilibrium where \Plone{} puts non-zero probability on $\mSim$.
    This concludes the proof of \Cref{thm:negative}.
\end{proof}

\begin{figure}
    \centering
    \begin{NiceTabular}{rccc}[cell-space-limits=3pt]
        & $\Cooperate$     & $\Defect$ & $\ \ \ $\\
        $\FullTrust$ & \Block[hvlines]{3-2}{}  $20, 20$    & $-100, 100$ & $\ \ \ $ \\
        $\PartialTrust$ &                         $10, 10$    & $\ \, -25, 25 \ $ & $\ \ \ $ \\
        $\WalkOut$ &                         $0, 0$      & $\ \ \ 0, 0 \ $ & $\ \ \ $ 
    \end{NiceTabular}
    \\
    \bigskip
    \begin{NiceTabular}{rccc}[cell-space-limits=3pt]
            & $\Cooperate$     & $\Defect$ & $\ \ \ $\\
        $\FullTrust$ & \Block[hvlines]{3-2}{}  $\Good\Full_1, \Good\Full_2$    & $\Bad\Full_1, \Awesome\Full_2$ & $\ \ \ $ \\
        $\PartialTrust$ &                         $\Good\Partial_1, \Good\Partial_2$            & $\Bad\Partial_1, \Awesome\Partial_2$ & $\ \ \ $ \\
        $\WalkOut$ &                         $\Neutral_1, \Neutral_2$      & $\Neutral_1, \Neutral_2$ & $\ \ \ $
    \end{NiceTabular}
    \caption{A concrete and a parameterised version of the Partial-Trust Game from \Cref{fig:PTG}.
        The names of the constants are meant as mnemonics for (fully- and partially-) Bad, Neutral, Good, and Awesome.
        Correspondingly, we assume that
            % \begin{align*}
            $
                \Bad\Full_1 < \Bad\Partial_1 < \Neutral_1 < \Good\Partial_1 < \Good\Full_1
            $
            and 
            $\Neutral_2 < \Good\Partial_2 < \Awesome\Partial_2$,
            $\Neutral_2 < \Good\Full_2 < \Awesome\Full_2$,
            $\Good\Partial_2 < \Good\Full_2$, and
            $\Awesome\Partial_2 < \Awesome\Full_2$
            (though the last one is not strictly necessary).
        (The relationship between $\Good\Full_2$ and $\Awesome\Partial_2$ is not important.)
    }
    \Description{A concrete and a parameterised version of a generalised partial-trust game.}
    \label{fig:parametrised_PTG}
\end{figure}

\section{Proofs \texorpdfstring{for \Cref{sec:sub:positive_partial_trust} (Partial Trust)}{Related to Partial Trust}}\label{sec:app:generalised_partial_trust}

In this section, we present the proofs related to partial-trust games.
Recall that we defined these as follows:
\defGPTG*

As an immediate observation, we can prove the first two claims of \Cref{lem:gPTG_properties}:

\begin{lemma}[Restating \Cref{lem:gPTG_properties}\,(i)-(ii)]
    Let $\game$ be a generalised PTG.
    \begin{enumerate}[label=(\roman*)]
        \item For any $\mixedStrategy \in \NE(\game)$, $\mixedStrategy_1(\WO) = 1$.
        \item The unique pure-commitment equilibrium of $\game$ is $(\FT, \C)$, where
            $\{ \FT \}  = \argmax \left\{ \utility_1(\T, \C) \mid \T \in \PureStrategies^\game_1 \right\}$.\\
            In particular, $\game$ is a generalised trust game.
    \end{enumerate}
\end{lemma}

\begin{proof}
    (i): This immediately follows from the property (3).
        Indeed,
            if $\mixedStrategy_1(\WO) > 0$, then \Pltwo{}'s only best response is $\D$ by (3),
            to which \Plone{}'s only best response is $\WO$ by (3).

    (ii): Since \Pltwo{} only has two actions, the only candidates for \Pltwo{}'s optimal pure commitment are $\C$ and $\D$.
        By (3), \Plone{} will respond to $\C$ with some action $\T \neq \WO$,
            which means that \Pltwo{}'s optimal pure commitment is $\C$.
        By (4), \Plone{}'s best response to $\C$ is unique.
        
    To show that $\game$ is a generalised trust game, we need to show
        that any pure-commitment equilibrium of $\game$ is a strict Pareto-improvement over any \NE{} of $\game$.
        By (3), $(\FT, \C)$ is a strict Pareto-improvement over both $(\WO, \D)$ and $(\WO, \C)$.
        In combination with (i), this shows that $\game$ is a generalised trust game.
\end{proof}

When dealing with mixed strategies of \Pltwo{} in a generalised PTGs,
    it will be useful to consider the following strategies:
    \begin{notation}
        Let $\game$ be a generalised PTG
            and $\T \in \PureStrategies^\game_1$ a strategy that
                is a best response to some $\mixedStrategy_2 \in \MixedStrategies^\game_2$.
        We denote
        \begin{align*}
            \lowerD{\T}
            & := \min \{ 
                    \delta \in [0, 1]
                \mid
                    \T \in \br(\delta \cdot \D + (1-\delta) \cdot \C)
                 \}
            \\ 
            \upperD{\T}
            & := \max \{ 
                    \delta \in [0, 1]
                \mid
                    \T \in \br(\delta \cdot \D + (1-\delta) \cdot \C)
                 \}
             .
        \end{align*}
        We also use $\incentivise{\T}$ to denote we the \textbf{optimal strategy of \Pltwo{} that still incentivises $\T$ as a best response}:
        \begin{align*}
            \incentivise{\T}
            & :=
                \delta_\T \cdot \D + (1-\delta_\T) \cdot \C
            .
        \end{align*}
    \end{notation}

We will also use the following notation for $\game$'s payoffs.

    \begin{notation}
        Let $\game$ be a generalised partial-trust game and $\T \neq \WO$ a strategy of \Plone{} in game.
        Then we denote
        \begin{align*}
            (\Good^\T_1, \Good^\T_2) & := \utility(\T, \C)
            \\
            (\Bad^\T_1, \Awesome^\T_2) & := \utility(\T, \D)
            \\
            (\Neutral_1, \Neutral_2) & := \utility(\WO, \C) = \utility(\WO, \D)
            .
        \end{align*}
    \end{notation}
    \noindent
    In line with the assumptions made in \Cref{def:gPTG},
        these numbers are meant to evoke the words \textbf{G}ood, \textbf{B}ad, \textbf{A}wesome, and \textbf{N}eutral.
    While \Cref{def:gPTG} posits that $\Neutral_1 = \Neutral_2 = 0$,
        our proofs will use the constants $\Neutral_\pl$ instead,
        to make it clearer where the ``zeroes'' come from.

As a first step in analysing generalised partial-trust games,
    we note that any such game can be reduced as follows.

\begin{lemma}\label{lem:gPTG_reduction_to_Ti}
    Let $\game$ be a generalised partial-trust game.
    \begin{enumerate}[label=(\roman*)]
        \item There exists some $n \geq 0$ and strategies
                \begin{align*}
                    \T_0 := \FT,
                    \T_1,
                    \dots,
                    \T_n,
                    \WO =: \T_{n+1}
                    \in \PureStrategies^\game_1
                    ,
                \end{align*}
            which satisfy
                \begin{align*}
                    \inverseBR(\T_i)
                    =
                    \left\{
                        \delta \cdot \D + (1\shortminus \delta) \cdot \C
                    \mid
                        \delta \in [\lowerD{\T_i}, \upperD{\T_i}]
                    \right\}
                \end{align*}
            and
                \begin{align*}
                    & 0 = \lowerD{\FT} < \upperD{\FT}
                    = \lowerD{\T_1} < \upperD{\T_1}
                    = \lowerD{\T_2} < \dots
                    \\
                    & \phantom{buga}
                    \dots < \upperD{\T_{n-1}}
                    = \lowerD{\T_n} < \upperD{\T_n}
                    = \lowerD{\WO} < \upperD{\WO} = 1
                    ,
                \end{align*}
            such that the game $\msimgame$ can,
                without loss of generality (in the sense of \Cref{prop:finite_str_space}),
                be replaced by the subgame $(\game')_\simSubscriptMixed$ given by
                \begin{align*}
                    \PureStrategies^{\game'}_1
                        & :=
                        \{ \FT, \T_1, \dots, \T_n, \WO \}
                    \\
                    \PureStrategies^{\game'}_2
                        & :=
                        \PureStrategies^{\game}_2
                        =
                        \{ \C, \D \}
                    .
                \end{align*}
        \item The game $(\game')_\simSubscriptMixed$ can,
            without loss of generality (in the sense of \Cref{prop:finite_str_space}),
            be further replaced by the subgame $\msimgame''$ given by 
            \begin{align*}
                \PureStrategies^{\msimgame''}_1
                    & :=
                    \PureStrategies^{(\game')_\simSubscriptMixed}_1
                    =
                    \{ \mSim, \FT, \T_1, \dots, \T_n, \WO \}
                \\
                \PureStrategies^{\msimgame''}_2
                    & :=
                    \{ \incentivise{\FT}, \incentivise{\T_1}, \dots, \incentivise{\T_n}, \D \}
                .
            \end{align*}
        \item In (i), we have $n > 0$ if and only if
            \begin{align*}
                \frac{\Good^\FT_1 - \Neutral_1}{\Neutral_1 - \Bad^\FT_1}
                \neq
                \max_{\T \in \PureStrategies^\game_1}
                    \frac{\Good^\T_1 - \Neutral_1}{\Neutral_1 - \Bad^\T_1}
                .
            \end{align*}
        \item The ratios $\frac{\upperD{\T_i}}{1-\upperD{\T_i}}$ satisfy 
            \begin{align*}
                \upperD{\FT} : (1 - \upperD{\FT})
                    & =
                    (\Good^{\FT}_1 - \Good^{\T_1}_1) : (\Bad^{\T_1}_1 - \Bad^{\FT}_1)
                \\
                \upperD{\T_i} : (1 - \upperD{\T_i})
                    & =
                    (\Good^{\T_i}_1 - \Good^{\T_{i+1}}_1) : (\Bad^{\T_{i+1}}_1 - \Bad^{\T_i}_1)
                \\
                \upperD{\T_n} : (1 - \upperD{\T_n})
                    & =
                    (\Good^{\T_n}_1 - \Neutral_1) : (\Neutral_1 - \Bad^{\T_n}_1)
                .
            \end{align*} 
    \end{enumerate}
\end{lemma}

\begin{proof}[Proof of \Cref{lem:gPTG_reduction_to_Ti}]
    First, observe that
        because $\game$ only has two actions for \Pltwo{},
        of which $\D$ is strictly dominant against anything except for $\WO$,
    best-responses of \Plone{} in $\game$ have the following properties.

    \begin{claim}\label{cl:gPTG_br_properties}
        For any generalised PTG, we have the following.
        \begin{enumerate}[label=(\roman*)]
            \item For every $\T \in \PureStrategies^\game_1$,
                \begin{align*}
                    \inverseBR(\T)
                    =
                    \{
                        \mixedStrategy_2 \in \MixedStrategies^\game_2
                    \mid
                        \br(\mixedStrategy_2) \ni \T
                    \}
                \end{align*}
                is either empty
                or equal to
                    \begin{align*}
                        \inverseBR(\T)
                        =
                        \{
                            \delta \cdot \D + (1-\delta) \cdot \C
                        \mid
                            \delta \in [\lowerD{\T}, \upperD{\T}]
                            \,
                        \}
                    \end{align*}
                for some $\lowerD{\T} \leq \upperD{\T}$.\footnotemark{}
                    \footnotetext{
                        The star in $\upperD{\T}$ is meant to indicate that this is the \textit{optimal} defection probability
                            -- from the perspective of \Pltwo{} --
                            corresponding to \Plone{} still being incentivised to best-respond by $\T$.
                    }
            \item Let $\T, \T' \in \bigcup \{ \br(\mixedStrategy_2) \mid \mixedStrategy_2 \in \MixedStrategies^\game_2 \}$ be s.t.
                the intervals $[\lowerD{\T}, \upperD{\T}]$ and $[\lowerD{\T'}, \upperD{\T'}]$ intersect.
                Then
                    either the intersection consists of a single point
                    or $[\lowerD{\T}, \upperD{\T}] = [\lowerD{\T'}, \upperD{\T'}]$.
            \item Let $\T \in \PureStrategies^\game_1$ be a strategy for which $[\lowerD{\T}, \upperD{\T}]$ is not a single point.
                Then, because of the assumption (4) in \Cref{def:gPTG},
                the case $[\lowerD{\T}, \upperD{\T}] = [\lowerD{\T'}, \upperD{\T'}]$ does not occur for any $\T' \neq \T$.
            \item Suppose that $\Tless, \Tmid, \Tmore \in \bigcup_{\mixedStrategy_2 \in \MixedStrategies^\game_2} \br(\mixedStrategy_2)$
                satisfy
                \begin{align*}
                    \lowerD{\Tmore} < \upperD{\Tmore}
                    = \upperD{\Tmid} = \delta^* = \lowerD{\Tmid}
                    = \lowerD{\Tless} < \upperD{\Tless}
                \end{align*}
                and denote $\mixedStrategy^*_2 := \delta^* \cdot \D + (1-\delta^*) \cdot \C$.
                Then we have
                \begin{enumerate}[label=(\alph*)]
                    \item
                        $
                            \fbr(\mixedStrategy^*_2)
                            =
                            \{ \Tmore \}
                        $;
                    \item
                        $
                            \utility_2(\Tmore, \mixedStrategy^*_2 )
                            >
                            \utility_2(\Tmid, \mixedStrategy^*_2 )
                            >
                            \utility_2(\Tless, \mixedStrategy^*_2 )
                        $.
                \end{enumerate}
            \item $\FT$ and $\WO$ belong to $\bigcup_{\mixedStrategy_2 \in \MixedStrategies^\game_2} \br(\mixedStrategy_2)$
                and the corresponding intervals satisfy:
                \begin{enumerate}[label=(\roman*)]
                    \item $0 = \lowerD{\FT} < \upperD{\FT} \leq \lowerD{\WO} < \upperD{\WO} = 1$;
                    \item $\upperD{\FT} < \lowerD{\WO}$ holds unless
                        $
                            \frac{\Good^\FT_1 - \Neutral_1}{\Neutral_1 - \Bad^\FT_1}
                            =
                            \max_{\T \in \PureStrategies^\game_1}
                                \frac{\Good^\T_1 - \Neutral_1}{\Neutral_1 - \Bad^\T_1}
                        $.
                \end{enumerate}
            \item The set $\bigcup_{\mixedStrategy_2 \in \MixedStrategies^\game_2} \fbr(\mixedStrategy_2)$
                consists precisely of the strategies whose sets $\inverseBR(\T)$ correspond to non-degenerate intervals.
        \end{enumerate}
    \end{claim}

    \begin{proof}[Proof of \Cref{cl:gPTG_br_properties}]
        (i) follows from the (general) fact that the set $\inverseBR(\pureStrategy_1)$ is always closed and convex.

        (ii) is a consequence of the general observation that
            when $\mixedStrategy_2$ is a non-trivial convex combination of $\mixedStrategyAlt_2$ and $\mixedStrategyAlt_2'$
            and $\pureStrategy_1, \pureStrategyAlt_1 \in \PureStrategies^\game_1$,
            exactly one of the following cases must be true:
                \begin{align*}
                    \utility_1(\pureStrategy_1, \mixedStrategyAlt_2) = \utility_1(\pureStrategyAlt_1, \mixedStrategyAlt_2)
                    & \ \ \ \& \ \ \ 
                    \utility_1(\pureStrategy_1, \mixedStrategyAlt_2') = \utility_1(\pureStrategyAlt_1, \mixedStrategyAlt_2')
                    \\
                    \utility_1(\pureStrategy_1, \mixedStrategyAlt_2) > \utility_1(\pureStrategyAlt_1, \mixedStrategyAlt_2)
                    & \ \ \ \& \ \ \ 
                    \utility_1(\pureStrategy_1, \mixedStrategyAlt_2') < \utility_1(\pureStrategyAlt_1, \mixedStrategyAlt_2')
                    \\
                    \utility_1(\pureStrategy_1, \mixedStrategyAlt_2) < \utility_1(\pureStrategyAlt_1, \mixedStrategyAlt_2)
                    & \ \ \ \& \ \ \ 
                    \utility_1(\pureStrategy_1, \mixedStrategyAlt_2') > \utility_1(\pureStrategyAlt_1, \mixedStrategyAlt_2')
                    .
                \end{align*}
            (This observation is a straightforward consequence of linearity of utility functions.)
            Indeed, the observation implies that if the two intervals intersect at two distinct points,
                $\T$ and $\T'$ would have to satisfy
                    $\utility_1(\T, \C) = \utility_1(\T', \C)$
                    and
                    $\utility_1(\T, \D) = \utility_1(\T', \D)$,
                which would imply that the two intervals coincide.
            
        (iii) follows as corollary of the proof of (ii),
            since \Cref{def:gPTG} assumes that different actions of \Plone{}
            cannot, in fact, have identical utilities $\utility_1(\symbolPlaceholder, \C)$.

        (iv-a):
            From (iii), it follows that $\Tmore$ must in fact be the only strategy of \Plone{}
                that corresponds to a non-degenerate interval
                and satisfies
                    $
                        \upperD{\T}
                        =
                        \delta^*
                    $.
            This implies that,
                among all strategies $\T \in \br(\mixedStrategy^*_2)$,
                $\Tmore$ must be the one with (strictly) highest $\utility_1(\T, \C)$ and (strictly) lowest $\utility_1(\T, \D)$.
            By the assumption (3) in the definition of a generalised PTG (\Cref{def:gPTG}),
                this implies that among all strategies $\T \in \br(\mixedStrategy^*_2)$,
                $\Tmore$ is the one with (strictly) highest utilities $\utility_2(\T, \C)$ and $\utility_2(\T, \D)$.
            As a result, $\Tmore$ is the unique favourable best response to $\mixedStrategy^*_2$.

        (iv-b) follows from the same argument as (iv-a).

        (v-a) is a straightforward consequence of definition of a general PTG (and $\FT$).

        (v-b): It is not difficult to verify that for every $\T \neq \WO$,
            the point when \Plone{} becomes indifferent between $\T$ and $\WO$ is
            \Pltwo{}'s strategy
                $
                    \mixedStrategy_2
                    =
                    \delta \cdot \D + (1-\delta) \cdot \C
                $
            satisfies
            \begin{align}\label{eq:gPTG_properties_lemma_trivility}
                \frac{\delta}{1-\delta}
                =
                \frac{\Good^\T_1 - \Neutral_1}{\Neutral_1 - \Bad^\T_1}
                .
            \end{align}
            Moreover, any $\T \notin \{ \FT, \WO \}$ will be \textit{strictly} dominated
                by $\FT$ when $\delta = 0$
                and by $\WO$ when $\delta = 1$.
            From these observations, it follows that when
                \begin{align*}
                    \frac{\Good^\T_1 - \Neutral_1}{\Neutral_1 - \Bad^\T_1}
                    \leq
                    \frac{\Good^\FT_1 - \Neutral_1}{\Neutral_1 - \Bad^\FT_1}
                    ,
                \end{align*}
                the point when the expected utility
                    $\utility_1(\T, \delta \cdot \D + (1-\delta) \cdot \C)$
                becomes equal to $\Neutral_1$ will come before
                    -- i.e., for higher $\delta$ --
                    the point when the same happens to
                    $\utility_1(\FT, \delta \cdot \D + (1-\delta) \cdot \C)$.
                This implies that the strategy $\T$ will always be (at least weakly) dominated by
                    either $\FT$, or $\WO$, or both.
            When the inequality \eqref{eq:gPTG_properties_lemma_trivility} holds for \textit{every} $\T \notin \{ \FT, \WO \}$,
                the strategies $\FT$ and $\WO$ will still both be best-responses
                even for the value of $\delta$ that makes \Plone{} indifferent between $\FT$ and $\WO$.

        (vi) follows from (iv-a) and the $\FT$ part of (v-i).            
    \end{proof}

\textit{    The proof of \Cref{lem:gPTG_reduction_to_Ti}\,(i), part 1:
}        Equipped with \Cref{cl:gPTG_br_properties},
            we are now ready to identify the advertised actions $\T_i$
            and to observe some of their basic properties.
        Denote
            \begin{align*}
                \mc T
                := \{
                    \pureStrategy_1 \in \PureStrategies^\game_1
                \mid
                    \exists \mixedStrategy_2 \in \MixedStrategies^\game_2 :
                        \fbr(\mixedStrategy_2) \ni \pureStrategy_1
                \}
            .
            \end{align*}
        By \Cref{cl:gPTG_br_properties}\,(vi),
            the set $\inverseBR(\T)$ for each $\T \in \mc T$
            corresponds to a non-degenerate interval.
        Moreover, the corresponding intervals necessarily
            cover the full range $[0, 1]$
            and
            only intersect at endpoints.
        By \Cref{cl:gPTG_br_properties}\,(v-i),
            we have $\mc T \supseteq \{ \FT, \WO \}$,
            so we can write $\mc T = \mc T' \cup \{ \FT, \WO \}$.
        By enumerating the elements of $\mc T$ from lowest to highest $\upperD{\T}$,
            we can thus write
                $
                    \mc T
                    =
                    \{ \T_0 := \WO, \T_1, \dots, \T_n, \T_{n+1} := \WO \}
                $
            for some $n := | \mc T' | \geq 0$.

    \textit{The proof of \Cref{lem:gPTG_reduction_to_Ti}\,(iii):}
        By \Cref{cl:gPTG_br_properties}\,(v-ii),
            we have $n > 0$
            if and only if
            $
                \frac{\Good^\FT_1 - \Neutral_1}{\Neutral_1 - \Bad^\FT_1}
                \neq
                    \max_{\T \in \PureStrategies^\game_1}
                        \frac{\Good^\T_1 - \Neutral_1}{\Neutral_1 - \Bad^\T_1}
            $.

    \textit{The proof of \Cref{lem:gPTG_reduction_to_Ti}\,(i), continued:}
        Recall that in $\game'$,
            the set of strategies of \Pltwo{} remains unchanged
            and \Plone{}'s set of strategies is given by
                \begin{align*}
                    \PureStrategies^{\game'}_1
                    & :=
                    \{ \FT, \T_1, \dots, \T_n, \WO \}
                    .
                \end{align*}
        Since the set of strategies of \Pltwo{} remains unchanged,
            we only need to verify the assumptions of \Cref{lem:extremal_points_only}
            for strategies of \Plone{}.
        To this end, consider a pure strategy $\pureMetaStrategy_1$ of \Plone{} in $\msimgame$
            that is a best-response to some mixed (meta-) strategy of \Pltwo{} in $\msimgame$.
        To verify the assumptions of \Cref{lem:extremal_points_only},
            we need to find a mixed (meta-) strategy $\mixedMetaStrategy'_1$ of \Plone{} in $(\game')_\simSubscriptMixed$
            that satisfies
                \begin{align*}
                    \utility_1(\mixedMetaStrategy'_1, \mixedStrategy_2)
                        & \geq
                        \utility_1(\pureMetaStrategy_1, \mixedStrategy_2)
                    \\
                    \utility_2(\mixedMetaStrategy'_1, \mixedStrategy_2)
                        & =
                        \utility_2(\pureMetaStrategy_1, \mixedStrategy_2)
                \end{align*}
                for every $\mixedStrategy_2 \in \PureStrategies^{(\game')_\simSubscriptMixed}_2$.
        However, in this case, we will show that the condition on $\utility_1$ actually holds with identity.
        We need to consider two cases:
            $\pureMetaStrategy_1 = \mSim$
            and
            $\pureMetaStrategy_1 \in \PureStrategies^\game_1$.
        
        Suppose that $\pureMetaStrategy_1 = \mSim$.
            We will show that the condition holds with $\mixedMetaStrategy'_1 = \pureMetaStrategy_1 = \mSim$.
            By \Cref{cl:gPTG_br_properties}\,(vi),
                we know that
                    $
                        \fbr^{\game'}(\mixedStrategy_2)
                        =
                        \fbr^{\game}(\mixedStrategy_2)
                    $
                for every $\mixedStrategy_2 \in \MixedStrategies^\game_2$.
            From this, it follows that
                $
                    \utility^{\game'}(\mSim, \mixedStrategy_2)
                    =
                    \utility^{\game}(\mSim, \mixedStrategy_2)
                $
                for every $\mixedStrategy_2 \in \MixedStrategies^\game_2$
                (since the utilities corresponding to $\mSim$ are defined in terms of favourable best responses in the base game).

        Suppose that $\pureMetaStrategy_1 \in \PureStrategies^\game_1$.
        First, when $\pureMetaStrategy_1 \in \{ \FT, \T_1, \dots, \T_n, \WO \}$,
            the condition trivially holds with $\mixedMetaStrategy'_1 := \pureMetaStrategy_1$.
        Second, if $\inverseBR(\pureMetaStrategy_1) = \emptyset$,
            $\pureMetaStrategy_1$ would not be a best response to any $\mixedStrategy_2$,
            so there would not be anything to prove.
        Finally, suppose that $\pureMetaStrategy_1 =: \T$ does not belong to the set  $\{ \FT, \T_1, \dots, \T_n, \WO \}$,
            but $\inverseBR(\T)$ is non-empty.
            By \Cref{cl:gPTG_br_properties}\,(ii), this means that
                $
                    \lowerD{\T}
                    =
                    \upperD{\T}
                $.
            By construction of the strategies $\T_i$
                -- and \Cref{cl:gPTG_br_properties}\,(ii) --
                we have
                    $
                        \upperD{\T_i}
                        =
                        \upperD{\T}
                        =
                        \lowerD{\T_{i+1}}
                    $
                for some $i \in \{ 0, \dots, n \}$.
            This means that we can find some $\lambda \in (0, 1)$
                such that the convex combination
                    $
                        \mixedStrategy_1
                        :=
                        \lambda \cdot \T_i + (1-\lambda) \cdot \T_{i+1}
                    $
                satisfies
                both
                    $\utility_1(\T, \C) = \utility_1(\mixedStrategy_1, \C)$
                and
                    $\utility_1(\T, \D) = \utility_1(\mixedStrategy_1, \D)$.
                (and thus for every $\mixedStrategy_2 \in \MixedStrategies^\game_2$).
            However,
                by the assumption \ref{ass:gPTG_definition_no_tiebreaking} from \Cref{def:gPTG},
                such convex combination must in fact satisfy
                    $\utility(\T, \mixedStrategy_2) = \utility(\mixedStrategy_1, \mixedStrategy_2)$
                    for every $\mixedStrategy_2$
                (i.e., for \textit{both} players).
            This shows that removing such $\T$ does not have a strategic impact on $\game$,
                and thus concludes the proof of \Cref{lem:gPTG_reduction_to_Ti}\,(i).

    \textit{The proof of \Cref{lem:gPTG_reduction_to_Ti}\,(ii):}
        We need to show that the game $(\game')_\simSubscriptMixed$
        can be reduced to the subgame $\msimgame''$
            where \Plone{}'s strategy space remains unchanged
            and \Pltwo{} only has access to pure (meta-) strategies
            \begin{align*}
                \PureStrategies^{\msimgame'}_2
                & :=
                \{ \incentivise{\FT}, \incentivise{\T_1}, \dots, \incentivise{\T_n}, \D \}
                .
            \end{align*}
        
        First, note that
            when $\delta \in (\lowerD{\T_i}, \upperD{\T_i})$,
            then the strategy $\delta \cdot \D + (1-\delta) \cdot C$ is
                (a) \textit{weakly} dominated by $\incentivise{\T_i}$
                    against \textit{any} strategy of \Plone{} in $\msimgame''$
                and
                (b) strongly dominated against any \Plone{} strategy in $\msimgame''$ with the exception of $\WO$
                    (i.e., including against $\mSim$).
            (This trivially holds because $\D$ strictly dominated $\C$ against any strategy in $\game$ except for $\WO$
                (and ties against $\WO$)
                and all strategies corresponding to a defection probability in $(\lowerD{\T_i}, \upperD{\T_i}]$
                have the same favourable best response, $\T_i$.)
        This implies that any \NE{} of $\msimgame''$ is a \NE{} of $(\game')_\simSubscriptMixed$.

        It remains to show that
            for any $\metaNE \in \NE((\game')_\simSubscriptMixed)$,
            there is $\mixedMetaStrategy' \in \NE(\msimgame'')$
            that satisfies $\utility(\mixedMetaStrategy') = \utility(\metaNE)$.
        When $\metaNE_1(\WO) = 1$,
            the suitable witness is $(\mixedMetaStrategy' := (\WO, \D)$
            (which clearly exists as a \NE{} in $\msimgame''$ as well as in $(\game')_\simSubscriptMixed$).
        When $\metaNE_1(\WO) < 1$,
            the observation (b) from the previous paragraph implies that
            \Pltwo{} can only use the strategies from within the set $\PureStrategies^{\msimgame'}_2$.
        And because \Plone{}'s set of strategies is identical in both games,
            any such \NE{} of $(\game')_\simSubscriptMixed$ is a NE of $\msimgame''$ as well.
        This concludes the proof of \Cref{lem:gPTG_reduction_to_Ti}\,(ii).

    \textit{The proof of \Cref{lem:gPTG_reduction_to_Ti}\,(iv):}
        This is immediately follows from the requirement that
        \begin{align*}
            \utility_1 (\T_i, \, \upperD{\T_i} \C \!+\! (1 \shortminus \upperD{\T_i}) \D )
            \!=\!
            \utility_1 (\T_{i+1}, \, \upperD{\T_i} \C \!+\! (1 \shortminus \upperD{\T_i}) \D )
            ,
        \end{align*}
        which is equivalent to
        \begin{align*}
            \upperD{\T_i} \Bad^{\T_i}_1 + (1 \shortminus \upperD{\T_i}) \Good^{\T_i}_1
            =
            \upperD{\T_i} \Bad^{\T_{i+1}}_1 + (1 \shortminus \upperD{\T_i}) \Good^{\T_{i+1}}_1
            .
        \end{align*}
\end{proof}

As a particular corollary of \Cref{lem:gPTG_reduction_to_Ti}, we can now complete the proof of \Cref{lem:gPTG_properties}.

\begin{lemma}[Restating \Cref{lem:gPTG_properties}\,(iii)-(iv)]
    Let $\game$ be
        a generalised partial-trust game.
    If $\utility_2(\FT, \C)$ is sufficiently high relative to other payoffs, then:
    \begin{enumerate}[label=(\roman*)]
        \setcounter{enumi}{2}
        \item $(\FT, \incentivise{\FT})$ is the unique \SE{} of $\game$.
        \item $(\FT, \incentivise{\FT})$ is a strict Pareto-improvement over the original \NE{} $(\WO, \D)$
            if and only if
            there is $\T \in \PureStrategies^\game_1$ s.t.
                $
                    \frac{
                        \Good^{\T}_1 - \Neutral_1
                    }{
                        \Neutral_1 - \Bad^{\T}_1
                    }
                    >
                    \frac{
                        \Good^{\FT}_1 - \Neutral_1
                    }{
                        \Neutral_1 - \Bad^{\FT}_1
                    }
                $.
    \end{enumerate}
\end{lemma}

\begin{proof}
    (iii): From \Cref{lem:gPTG_reduction_to_Ti}\,(ii),
        it follows that \Pltwo{}'s Stackelberg strategy
        must be one of the strategies $\incentivise{\T}$, $\T \in \PureStrategies^\game_1$.
        Moreover, we know that
            the probabilities $\delta_\T$ are a function of \Plone{}'s payoffs
            and 
            of the utilities $\utility_2(\fbr(\incentivise{\T}), \incentivise{\T})$,
                only $\utility_2(\fbr(\incentivise{\FT}), \incentivise{\FT})$
                depends on $\utility_2(\FT, \C)$.
        As a result,
            if we hold fixed all utilities except for $\utility_2(\FT, \C)$
                (and $\utility_2(\FT, \D)$, which must be larger than $\utility_2(\FT, \C)$
                for $\game$ to qualify as a generalised PTG),
            then for sufficiently large $\utility_2(\FT, \C)$,
            \Pltwo{}'s SE strategy can only be equal to $\incentivise{\FT}$.

    (iv): First, note that in the game where \Plone{} only has actions $\{ \FT, \WO \}$,
        the Stackelberg equilibrium necessarily yields $\utility_1 = \utility_1(\WO, \D)$,
        and thus does not constitute a strict Pareto improvement over $(\WO, \D)$.
        Note also that any \SE{} that does not involve \Plone{} putting 100\% probability on $\WO$ will constitute a strict improvement in $\utility_2$ over $(\WO, \D)$.
    
        If we have
            $
                \frac{
                    \Good^{\T}_1 - \Neutral_1
                }{
                    \Neutral_1 - \Bad^{\T}_1
                }
                \leq
                \frac{
                    \Good^{\FT}_1 - \Neutral_1
                }{
                    \Neutral_1 - \Bad^{\FT}_1
                }
            $
            for every $\T \notin \{ \FT, \D \}$,
        every $\T \notin \{ \FT, \WO \}$ will be weakly dominated by the strategy $\mixedStrategy_1$ of the form
            $\mixedStrategy_1 = \alpha \cdot \FT + (1-\alpha) \cdot \WO$,
            $
                \alpha
                :=
                \frac{\Good^{\T}_1 - \Neutral_1}{\Good^{\FT}_1 - \Neutral_1}
            $.
        Since we have $\utility_2(\FT, \mixedStrategy_2) > \utility_2(\T, \mixedStrategy_2)$
                for any $\T \neq \FT$ and $\mixedStrategy_2$,
            this implies that \Plone{}'s part of the Stackelberg equilibrium will necessarily consist of best-responding by $\FT$.
        By the earlier observation, a SE of this form cannot give \Plone{} more utility than the \NE{} $(\WO, \D)$.

        Conversely, if we have
            $
                \frac{
                    \Good^{\T}_1 - \Neutral_1
                }{
                    \Neutral_1 - \Bad^{\T}_1
                }
                >
                \frac{
                    \Good^{\FT}_1 - \Neutral_1
                }{
                    \Neutral_1 - \Bad^{\FT}_1
                }
            $
            for some $\T$,
        then by \Cref{lem:gPTG_reduction_to_Ti}\,(iii),
            the \SE{} of $\game$ will have \Plone{} being indifferent between $\FT$ and some $\T_i \neq \WO$.
        By \Cref{cl:gPTG_br_properties} (and the construction of $\T_i$), this means that $\WO$ will not be a best-response to
            this strategy of \Pltwo{}.
        In particular, this means that this SE yields a payoff higher than $\utility_1(\WO, \D)$.
\end{proof}

\thmGenPTG*

\noindent
The proof of \Cref{thm:generalised_PTG} includes the proofs of Claims~\ref{cl:gPTG_remove_FT_WO_C}-\ref{cl:gPTG_subgame_NE_is_NE}.
We first give a high-level idea of the proof, and then proceed with the proof itself.

\begin{proof}[Proof sketch for \Cref{thm:generalised_PTG}]
    A key insight
        -- encompassed by \Cref{lem:gPTG_reduction_to_Ti} --
    is that as \Pltwo{}'s defection probability $\delta$ increases from $0$ to $1$,
        different strategies of \Plone{} will be a best response to $\delta \cdot \D + (1-\delta) \cdot C$.
    This creates a natural way for \Plone{} to match their level of trust to the expected rate of defection,
        with $\FT$ corresponding to maximum trust,
        $\WO$ corresponding to minimum trust,
        and with a range of partial-trust actions in between.
        (Note that some actions might not be the most appropriate response to \textit{any} value of $\delta$.)
        
    The proof then consists of showing that there is an \NE{} where
        \Pltwo{} mostly selects the highest possible rate of defection that still results in \Plone{} responding by $\FullTrust$ when simulating, denoted $\dFT$,
        but they sometimes try to exploit \Plone{} by defecting with probability $100\%$.
    In response,
        \Plone{} typically uses the \textit{second} highest level of trust.
        However, they also sometimes simulate,
            which results either in \FullTrust\ or Walking Out,
            depending on whether \Pltwo{}'s defection probability is $\dFT$ or $100\,\%$.

    The assumption on \Pltwo{}'s payoffs ensures that
        \Pltwo{} cannot improve their utility by switching to some intermediate level of defection.
    The implicit assumption that $\simcost$ is low ensures that
        the probability that \Pltwo{} puts on $100\% \,\D$
        -- which is proportional to $\simcost$ --
        is low enough that 
        \Pltwo{}'s overall strategy still incentivises the \textit{second} highest level of trust as a best response.
\end{proof}

\begin{proof}
    Let $\game$ be a generalised PTG.
    By applying \Cref{lem:gPTG_reduction_to_Ti}, we reduce $\msimgame$ to the subgame $\msimgame''$ given by
        \begin{align*}
            \PureStrategies^{\game''}_1
                & :=
                    \{ \FT, \T_1, \dots, \T_n, \WO \}
                \\
            \PureStrategies^{\msimgame''}_2
                & :=
                \{ \incentivise{\FT}, \incentivise{\T_1}, \dots, \incentivise{\T_n}, \D \}
            .
        \end{align*}
    Before analysing the equilibria of this game, we will reduce it even further,
        by assuming that \Plone{} only plays $\FT$ and $\WO$ as a result of simulation.

    \begin{claim}\label{cl:gPTG_remove_FT_WO_C}
        Let $\msimgame'''$ be the subgame of $\msimgame$ given by
            \begin{align*}
                \PureStrategies^{\msimgame'''}_1
                & :=
                    \{ \T_1, \dots, \T_n, \mSim \}
                    = \PureStrategies^{\msimgame''}_1 \setminus \{ \FT, \WO \}
                \\
                \PureStrategies^{\msimgame'''}_2
                & :=
                    \{ \incentivise{\FT}, \incentivise{\T_1}, \dots, \incentivise{\T_n}, \D \}
                    =
                    \PureStrategies^{\msimgame''}_2
            \end{align*}
        (where $\mSim$ still gives the same utilities as in $\msimgame''$).
        Then:
        \begin{enumerate}[label=(\roman*)]
            \item $(\WO, \D) \in \NE(\msimgame'') \setminus \NE(\msimgame''')$.
                Moreover, every $\mixedMetaStrategy \in \NE(\msimgame'')$ with $\mixedMetaStrategy_1(\WO) > 0$
                has $\utility_1(\mixedMetaStrategy) = \Neutral_1$.
            \item Any $\mixedMetaStrategy \in \NE(\msimgame'')$ with $\mixedMetaStrategy_1(\WO) = 0$ is a \NE{} of $\msimgame'''$.
            \item For any $\mixedMetaStrategy \in \NE(\msimgame''')$, $\mixedMetaStrategy$ is a \NE{} of $\msimgame''$
                if and only if $\utility_1(\mixedMetaStrategy) \geq \Neutral_1$.
            \item In particular, mixed-str. sim. \iPiN{} in $\game$
                if and only if
                $
                    \exists \mixedMetaStrategy \in \NE(\msimgame''') :
                        \utility_1(\mixedMetaStrategy) > \Neutral_1
                $.
        \end{enumerate}
    \end{claim}

    \begin{proof}[Proof of \Cref{cl:gPTG_remove_FT_WO_C}]
        (i): This holds trivially,
            since $\utility_1(\WO, \C) = \utility_1(\WO, \D) = \Neutral_1$
            and $\utility_1(\T, \D) < \Neutral_1$ for every $\T \neq \WO$.
            (The second part follows from the fact that any \NE{} with $\mixedMetaStrategy_1(\WO) > 0$ has
                $\utility_1(\mixedMetaStrategy) = \utility_1(\WO, \mixedMetaStrategy_2)$.)

        (ii): We will prove this by showing that
            there is no \NE{} of $\msimgame''$ where \Plone{} puts positive probability on $\FT$.
            (Since $\msimgame'''$ is a subgame of $\msimgame''$,
                any \NE{} of $\msimgame''$ that only uses actions that are also available in $\msimgame'''$
                will be a \NE{} of $\msimgame'''$ as well.)
            First, observe that there is no \NE{} of $\msimgame''$ where $\mixedMetaStrategy_2(\incentivise{\FT}) = 1$.
                (To see why, note that $\br(\incentivise{\FT}) = \{\FT, \T_1\}$,
                    so in any such \NE{}, \Plone{} could only use strategies $\FT$ or $\T_1$.
                However, this would in turn mean that \Plone{} could strictly increase their utility by deviating to $\D$.)
            However, this means that
                \Pltwo{}'s overall defection probability $\mesa{\mixedMetaStrategy}_2(\D)$,
                for any $\mixedMetaStrategy \in \NE(\msimgame'')$,
                is \textit{strictly higher} than $\incentivise{\FT}_2(\D) = \upperD{\FT}$.
            As a result, any $\mixedMetaStrategy \in \NE(\msimgame'')$ makes $\FT$ strictly dominated by $\T_1$,
                which means that no \NE{} of $\msimgame''$ puts non-zero probability on $\FT$.

        (iii): To prove (iii), suppose that
            $\mixedMetaStrategy \in \NE(\msimgame''')$ satisfies $\utility_1(\mixedMetaStrategy) \geq \Neutral_1$.
            To show that $\mixedMetaStrategy$ is a \NE{} of $\msimgame''$,
                we only need to verify that \Plone{} cannot improve their utility by deviating to $\FT$ or $\WO$.
            For $\WO$, this immediately follows from the assumption that $\utility_1(\mixedMetaStrategy) \geq \Neutral_1$.
            For $\FT$, this follows from the fact that for all of the strategies of \Pltwo{},
                $\T_1$ weakly dominates $\FT$
                (in fact, it strongly dominates it for all except $\incentivise{\FT}$; though this is not needed here).

        (iv): To prove this, recall that any \NE{} of $\game$ has $\utility(\mixedStrategy) = (\Neutral_1, \Neutral_2)$.
            The ``only if'' part of (iv) trivially follows from the fact that
                if $\utility_1(\mixedMetaStrategy) \leq \Neutral_1$,
                then $\mixedMetaStrategy$ cannot be a (strict) Pareto-improvement over $\Neutral_1, \Neutral_2)$.
            The ``if'' part of (iv) follows from the fact that \textit{any} strategy profile $\mixedMetaStrategy$ in $\msimgame'''$,
                with the exception of $(\mSim, \D)$,
                satisfies $\utility_2(\mixedMetaStrategy) > \Neutral_2$.
                (And $(\mSim, \D)$ is not a \NE{} of $\msimgame'''$, since \Pltwo{} could profitably deviate to $\incentivise{\FT}$.)
    \end{proof}

    To conclude the proof of \Cref{thm:generalised_PTG},
        we will describe an equilibrium $\metaNE$ of the subgame given by $\{\T_1, \mSim\}$ and $\{\incentivise{\FT}, \D \}$
        and demonstrate that none of the actions $\T_2$, $\dots$, $\T_n$ and $\incentivise{\T_1}$, $\dots$, $\incentivise{\T_n}$ constitute a profitable deviation,
        and show that $\utility_1(\metaNE) > \Neutral_1$.
    
    \begin{claim}\label{cl:gPTG_subgame_NE}
        Suppose that $\simcost < \Neutral_1 \setminus \Bad^{\T_1}_1$.
        Then the subgame
            $(
                \{ \T_1, \mSim \},
                \{ \incentivise{\FT}, \D\},
                \utility
            )$
        has a unique \NE{}.
        This \NE{} is of the form
        \begin{align*}
            \metaNE_1 & = \simProb \cdot \mSim + (1-\simProb) \cdot \T_1
            \\
            \metaNE_2 & = \evilProb \cdot \D + (1-\evilProb) \cdot \incentivise{\FT}
            ,
        \end{align*}
        where
        \begin{align*}
            \simProb : (1 - \simProb)
                & =
                \frac{
                    (1-\dFT) (\Awesome^{\T_1}_2 - \Good^{\T_1}_2)
                }{
                    \dFT (\Awesome^{\FT}_2 - \Good^{\FT}_2)
                    +
                    (\Good^{\FT}_2 - \Neutral_2)
                }
            \\
            \evilProb
                & =
                \frac{\simcost}{\Neutral_1 - \Bad^{\T_1}_1}
            .
        \end{align*}
    \end{claim}

    \begin{proof}[Proof of \Cref{cl:gPTG_subgame_NE}]     
        As the first step, we calculate the utilities corresponding to the outcomes
            $(\T_1, \incentivise{\FT})$, $(\mSim, \incentivise{\FT})$,
            $(\T_1, \D)$, and $(\mSim, \D)$.
        For $\D$, the unique best response of \Plone{} is $\WO$, so we have
            \begin{align*}
                \utility_1(\T_1, \D)
                    & =
                    \Bad^{\T_1}_1
                    \\
                \utility_2(\T_1, \D)
                    & =
                    \Awesome^{\T_1}_2
                \\
                \utility_1(\mSim, \D)
                    & =
                    \Neutral_1 - \simcost
                \\
                \utility_1(\mSim, \D)
                    & =
                    \Neutral_2
                .
            \end{align*}
        For $\incentivise{\FT}$, \Cref{lem:gPTG_reduction_to_Ti} gives
            \begin{align*}
                \upperD{\FT} : (1 - \upperD{\FT})
                & =
                (\Good^{\FT}_1 - \Good^{\T_1}_1) : (\Bad^{\T_1}_1 - \Bad^{\FT}_1)
            \end{align*}
        and by definition of $\mSim$
            -- and the fact that $\incentivise{\FT}$ makes \Plone{} indifferent between $\FT$ and $\T_1$ --
        we have
            \begin{align*}
                \utility_1(\T_1, \incentivise{\FT})
                & =
                    \upperD{\FT} \cdot \Bad^{\T_1}_1
                    +
                    (1\shortminus\upperD{\FT}) \cdot \Good^{\T_1}_1
                \\
                \utility_2(\T_1, \incentivise{\FT})
                & =
                    \upperD{\FT} \cdot \Awesome^{\T_1}_2
                    +
                    (1\shortminus\upperD{\FT}) \cdot \Good^{\T_1}_2
                \\
                \utility_1(\mSim, \incentivise{\FT})
                & =
                    \utility_1(\T_1, \D) - \simcost
                    \\
                    & =
                    \upperD{\FT} \cdot \Bad^{\T_1}_1
                    +
                    (1\shortminus\upperD{\FT}) \cdot \Good^{\T_1}_1
                    - \simcost
                \\
                \utility_2(\mSim, \incentivise{\FT})
                & =
                    \utility_2(\FT, \incentivise{\FT})
                    \\
                    & =
                    \upperD{\FT} \cdot \Awesome^{\FT}_2
                    +
                    (1\shortminus\upperD{\FT}) \cdot \Good^{\FT}_2
                    .
            \end{align*}

        From this, it follows that this $2 \times 2$ game cannot have a pure \NE{}.
            (For $(\T_1, \D)$,
                the assumption that $\simcost < \Neutral_1 - \Bad^{\T_1}_1$ implies that \Plone{} would have a profitable deviation to $\mSim$.
            For $(\T_1, \incentivise{\FT})$,
                \Pltwo{} would deviate to $\D$.
            For $(\mSim, \incentivise{\FT})$,
                \Plone{} would deviate to $\T_1$.
            For $(\mSim, \D)$,
                \Pltwo{} would deviate to $\incentivise{\FT}$.)
        
        Since the game has no pure \NE{},
            it must have a unique \NE{},
            and this \NE{} will be mixed.
        The values of $\simProb$ and $\evilProb$
            are obtained by a straightforward calculation,
            starting from the requirement that each player is indifferent between their two actions:
            \begin{align*}
                & \utility_1 \left( \mSim, \ \evilProb \cdot \D + (1-\evilProb) \cdot \incentivise{\FT} \right) =
                    \\
                    & \phantom{blaaaahblaah}
                    \utility_1 \left( \T_1, \ \evilProb \cdot \D + (1-\evilProb) \cdot \incentivise{\FT} \right),
                \\
                \\
                & \utility_2 \left( \simProb \cdot \mSim + (1-\simProb) \cdot \T_1 , \ \D \right)
                    \\
                    & \phantom{blaaaahbla}
                    =
                    \utility_2 \left( \simProb \cdot \mSim + (1-\simProb) \cdot \T_1 , \ \incentivise{\FT} \right)
                .
            \end{align*}
        For completeness, the derivations are as follows.
        For $\evilProb$, we calculate:
            \begin{align*}
                &
                    \utility_1 \left( \mSim, \ \evilProb \cdot \D + (1-\evilProb) \cdot \incentivise{\FT} \right)
                    \\
                    &
                    =
                    \utility_1 \left( \T_1, \ \evilProb \cdot \D + (1-\evilProb) \cdot \incentivise{\FT} \right)
                \\
                & \iff
                    \evilProb \cdot \Neutral_1
                    +
                    (1-\evilProb) \cdot \utility_1(\T_1, \incentivise{\FT}) - \simcost
                    \\
                    & \phantom{\iff \ }
                    =
                    \evilProb \cdot \Bad^{\T_1}_1
                    +
                    (1-\evilProb) \cdot \utility_1(\T_1, \incentivise{\FT})
                \\
                & \iff
                    \evilProb \cdot \Neutral_1 - \simcost
                    =
                    \evilProb \cdot \Bad^{\T_1}_1
                \\
                & \iff
                    \evilProb
                    =
                    \frac{\simcost}{\Neutral_1 - \Bad^{\T_1}_1}
                .
            \end{align*}
        For $\simProb$, we calculate:
            \begin{align*}
                &
                    \utility_2 \left( \simProb \cdot \mSim + (1-\simProb) \cdot \T_1 , \ \D \right)
                    \\ &
                    =
                    \utility_2 \left( \simProb \cdot \mSim + (1-\simProb) \cdot \T_1 , \ \incentivise{\FT} \right)
                \\
                & \iff
                    \simProb \cdot \utility_2 \left( \mSim , \D \right)
                    + (1-\simProb) \cdot \utility_2 \left( \T_1 , \D \right)
                    \\ & \phantom{\iff}
                    =
                    \simProb \utility_2 \left( \mSim , \incentivise{\FT} \right)
                    + (1 \shortminus \simProb) \utility_2 \left( \T_1 , \incentivise{\FT} \right)
                \\
                & \iff
                    \simProb \cdot \Neutral_2
                    + (1-\simProb) \cdot \Awesome^{T_1}_2
                    \\ & \phantom{\iff}
                    =
                    \simProb \cdot \left( 
                            \upperD{\FT} \cdot \Awesome^{\FT}_2
                            +
                            (1\shortminus\upperD{\FT}) \cdot \Good^{\FT}_2
                        \right)
                    \\ & \phantom{\iff} \phantom{aa}
                    + (1-\simProb) \cdot \left(
                            \upperD{\FT} \cdot \Awesome^{\T_1}_2
                            +
                            (1\shortminus\upperD{\FT}) \cdot \Good^{\T_1}_2
                        \right)
                \\
                & \iff
                    \simProb \cdot \Neutral_2
                    + (1-\simProb) \cdot \Awesome^{T_1}_2
                    \\ & \phantom{\iff}
                    =
                    \simProb \cdot \left( 
                            \upperD{\FT} \cdot ( \Awesome^{\FT}_2 - \Good^{\FT}_2 )
                            +
                            \Good^{\FT}_2
                        \right)
                    \\ & \phantom{\iff} \phantom{aa}
                    + (1-\simProb) \cdot \left(
                            \upperD{\FT} \cdot ( \Awesome^{\T_1}_2 - \Good^{\T_1}_2 )
                            +
                            \Good^{\T_1}_2
                        \right)
                \\
                & \iff
                    (1-\simProb) \cdot ( \Awesome^{T_1}_2 - \Neutral_2 )
                    \\ & \phantom{\iff}
                    =
                    \simProb \cdot \left[
                            \upperD{\FT} \cdot ( \Awesome^{\FT}_2 - \Good^{\FT}_2 )
                            +
                            ( \Good^{\FT}_2 - \Neutral_2 )
                        \right]
                    \\ & \phantom{\iff} \phantom{aa}
                    + (1-\simProb) \cdot \left[
                            \upperD{\FT} \cdot ( \Awesome^{\T_1}_2 - \Good^{\T_1}_2 )
                            +
                            ( \Good^{\T_1}_2 - \Neutral_2 )
                        \right]
                \\
                & \iff
                    (1 \shortminus \simProb) \cdot
                        \Big[
                            ( \Awesome^{T_1}_2 \shortminus \Neutral_2 )
                            \ 
                            -
                            \\
                            & \phantom{asdfsfasdfasdasdfas}
                            -
                            \upperD{\FT} \cdot ( \Awesome^{\T_1}_2 \shortminus \Good^{\T_1}_2 )
                            \shortminus
                            ( \Good^{\T_1}_2 \shortminus \Neutral_2 )
                        \Big]
                    \\ & \phantom{\iff}
                    =
                    \simProb \cdot \left[
                            \upperD{\FT} \cdot ( \Awesome^{\FT}_2 \shortminus \Good^{\FT}_2 )
                            +
                            ( \Good^{\FT}_2 \shortminus \Neutral_2 )
                        \right]
                \\
                & \iff
                    (1 \shortminus \simProb) \cdot
                        (1 - \upperD{\FT} ) ( \Awesome^{T_1}_2 \shortminus \Good^{\T_1} )
                    \\ & \phantom{\iff}
                    =
                    \simProb \cdot \left[
                            \upperD{\FT} \cdot ( \Awesome^{\FT}_2 \shortminus \Good^{\FT}_2 )
                            +
                            ( \Good^{\FT}_2 \shortminus \Neutral_2 )
                        \right]
                \\
                & \iff
                    \frac{
                        (1 - \upperD{\FT} )
                        ( \Awesome^{T_1}_2 \shortminus \Good^{\T_1}_2 )
                    }{
                        \upperD{\FT} \cdot ( \Awesome^{\FT}_2 \shortminus \Good^{\FT}_2 )
                        +
                        ( \Good^{\FT}_2 \shortminus \Neutral_2 )
                    }
                    =
                    \frac{\simProb}{1-\simProb}
                .
            \end{align*}
        This concludes the proof of \Cref{cl:gPTG_subgame_NE}.
    \end{proof}

    By combining \Cref{cl:gPTG_subgame_NE} with by \Cref{cl:gPTG_remove_FT_WO_C}\,(iv),
        we get that to conclude the proof of \Cref{thm:generalised_PTG},
            it suffices to prove the following claim.

    \begin{claim}\label{cl:gPTG_subgame_NE_is_NE}
        Let $\metaNE$ be the strategy from \Cref{cl:gPTG_subgame_NE}.
        Then:
        \begin{enumerate}[label=(\roman*)]
            \item There exists $c_0 > 0$ such that if $\simcost < c_0$, then
                (a) \Plone{} cannot increase their utility by deviating to $\T_2, \dots, \T_n$,
                and
                (b) $\utility_1(\metaNE) > \Neutral_1$.
            \item There exists $K \in \R$,
                    which does not depend on the payoffs corresponding to $\FT$,
                such that if $\Awesome^{\FT}_2 > \Good^{\FT}_2 \geq K$,
                then \Pltwo{} cannot increase their utility by deviating to $\incentivise{\T_1}, \incentivise{\T_2}, \dots,$ or $\incentivise{\T_n}$.
            % \item In particular, $\metaNE$ is an \NE{} of $\msimgame$ and it constitutes a strict Pareto-improvement over any NE of $\game$.
        \end{enumerate}
    \end{claim}

    \begin{proof}[Proof of \Cref{cl:gPTG_subgame_NE_is_NE}]
        % (iii): By applying \Cref{cl:gPTG_remove_FT_WO_C}\,(i), we immediately get that (iii) follows from (i) and (ii).
        \textit{(i):}
            By \Cref{lem:gPTG_reduction_to_Ti}\,(i),
            each $\T_i$ is the unique best response
                (among the strategies $\T_1, \dots, \T_n$ and $\T_{n+1} = \WO$)
            to any strategy $\mixedMetaStrategy_2$
                whose total defection probability
                    $
                        \mesa{\mixedMetaStrategy}_2(\D)
                        :=
                        \int \mixedStrategy_2(\D) \, \textnormal{d} \mixedMetaStrategy_2(\mixedStrategy_2)
                    $
                lies in $(\lowerD{\T_i}, \upperD{\T_i})$.
            To prove both (i-a) and (i-b), it thus suffices to show that
                $\mesa{\mixedMetaStrategy}^*_2 \in (\upperD{\FT}, \upperD{\T_1})$.
            By \Cref{cl:gPTG_subgame_NE}, we have
            \begin{align*}
                \mesa{\mixedMetaStrategy}^*_2
                & =
                \evilProb \cdot 1 + (1-\evilProb) \cdot \upperD{\FT}
                \\
                & =
                \upperD{\FT} + \evilProb \cdot (1-\upperD{\FT})
                \in (\upperD{\FT}, \upperD{\FT} + \evilProb)
                .
            \end{align*}
            It thus remains to verify that $\upperD{\FT} + \evilProb \leq \upperD{\T_1}$.
            
            Plugging in the values of
                $\evilProb$ from \Cref{cl:gPTG_subgame_NE},
                and $\upperD{\FT}$, $\upperD{\T_1}$ from \Cref{lem:gPTG_reduction_to_Ti}\,(iv),
                we get
            \begin{align*}
                &
                    \upperD{\FT} + \evilProb \leq \upperD{\T_1}
                    \\
                & \iff
                    \evilProb \leq \upperD{\T_1} - \upperD{\FT}
                    \\
                & \iff
                    \frac{\simcost}{\Neutral_1 - \Bad^{\T_1}_1}
                    \leq
                        \frac{\Good^{\T_1}_1 - \Good^{\T_2}_1}{\Bad^{\T_2}_1 - \Bad^{\T_1}_1}
                        -
                        \frac{\Good^{\FT}_1 - \Good^{\T_1}_1}{\Bad^{\T_1}_1 - \Bad^{\FT}_1}
                    \\
                & \iff
                    \simcost
                    \leq
                    c_0
                    \, ,
                    \\
                & \textnormal{where }
                    \simcost
                    :=
                    \left(
                    \frac{\Good^{\T_1}_1 \shortminus \Good^{\T_2}_1}{\Bad^{\T_2}_1 \shortminus \Bad^{\T_1}_1}
                    \shortminus
                    \frac{\Good^{\FT}_1 \shortminus \Good^{\T_1}_1}{\Bad^{\T_1}_1 \shortminus \Bad^{\FT}_1}
                    \right)
                    \left(
                        \Neutral_1 \shortminus \Bad^{\T_1}_1
                    \right)
                .
            \end{align*}
            By \Cref{lem:gPTG_reduction_to_Ti},
                $c_0$ is a positive constant
                (which only depends on the payoffs of \Plone{}),
            which shows that (i) holds.

        \textit{The proof of \Cref{cl:gPTG_subgame_NE_is_NE}\,(ii):}
            We need to show that, for sufficiently high $\Good^\FT_2$,
                we have
                $
                    \utility_2(\metaNE_1, \incentivise{\FT})
                    >
                    \utility_2(\metaNE_1, \incentivise{\T_i})
                $
                every $i = 1, \dots, n$.
            Plugging in the value of $\metaNE_1$ from \Cref{cl:gPTG_subgame_NE},
                it is straightforward (though somewhat lengthy) to calculate that
                \begin{align*}
                    \utility_2(\metaNE_1, \incentivise{\FT})
                    >
                    \utility_2(\metaNE_1, \incentivise{\T_i})
                \end{align*}
                holds if and only if
                \begin{align}\label{eq:gPTG_sufficiently_high_characterisation}
                        \frac{
                           1 - \upperD{\T_i}
                        }{
                            1 - \upperD{\FT}
                        }
                        >
                        \frac{
                           \Good^{\T_i}_2 \!\shortminus\! \Neutral_2 + \upperD{\T_i} ( \Awesome^{\T_i}_2 \!\shortminus\! \Good^{\T_i}_2) 
                        }{
                            \Good^\FT_2 \!\shortminus\! \Neutral_2 + \upperD{\FT} ( \Awesome^\FT_2 \!\shortminus\! \Good^\FT_2) 
                        }
                    .
                \end{align}

            For completeness, here is the derivation of this equivalence.
            Proceeding as in the proof of \Cref{cl:gPTG_subgame_NE}, we observe that
                for any $i$, $\utility_2(\metaNE_1, \incentivise{\T_i})$ can be expressed
                \begin{align*}
                    \utility_2(\metaNE_1, \incentivise{\T_i})                    
                    & =
                        \simProb \cdot \left[
                            \upperD{\T_i} \cdot ( \Awesome^{\T_i}_2 - \Good^{\T_i}_2 )
                            +
                            \Good^{\T_i}_2
                        \right]
                        \\ & \phantom{=}
                        + (1-\simProb) \cdot \left[
                                \upperD{\T_i} \cdot ( \Awesome^{\T_1}_2 - \Good^{\T_1}_2 )
                                +
                                \Good^{\T_1}_2
                            \right]
                    .
                \end{align*}
            From this, it follows that
                \begin{align}
                    &
                    \nonumber
                        \utility_2(\metaNE_1, \incentivise{\FT})
                        >
                        \utility_2(\metaNE_1, \incentivise{\T_i})
                        \iff
                    \\
                    & \iff
                        \frac{\simProb}{1-\simProb}
                        >
                        \\
                        \nonumber
                        &
                        >
                        \frac{
                            \left[
                                \upperD{\T_i} \cdot ( \Awesome^{\T_1}_2 \shortminus \Good^{\T_1}_2 )
                                \! + \!
                                \Good^{\T_1}_2
                            \right]
                            \shortminus
                            \left[
                                \upperD{\FT} \cdot ( \Awesome^{\T_1}_2 \shortminus \Good^{\T_1}_2 )
                                \! + \!
                                \Good^{\T_1}_2
                            \right]
                        }{
                            \left[
                                \upperD{\FT} \cdot ( \Awesome^{\FT}_2 \shortminus \Good^{\FT}_2 )
                                \! + \!
                                \Good^{\FT}_2
                            \right]
                            \shortminus
                            \left[
                                \upperD{\T_i} \cdot ( \Awesome^{\T_i}_2 \shortminus \Good^{\T_i}_2 )
                                \! + \!
                                \Good^{\T_i}_2
                            \right]
                        }
                    \\
                    & \iff
                        \frac{
                            (1-\dFT) (\Awesome^{\T_1}_2 - \Good^{\T_1}_2)
                        }{
                            \dFT (\Awesome^{\FT}_2 - \Good^{\FT}_2)
                            +
                            (\Good^{\FT}_2 - \Neutral_2)
                        }
                        >
                        \\
                        \nonumber
                        &
                        \frac{
                            \left( \upperD{\T_i} - \upperD{\FT} \right)
                            \left( \Awesome^{\T_1}_2 \shortminus \Good^{\T_1}_2 \right)
                        }{
                            \upperD{\FT} ( \Awesome^{\FT}_2 \! \shortminus \! \Good^{\FT}_2 )
                            \! + \! 
                            ( \Good^{\FT}_2 \! \shortminus \! \Neutral_2 )
                        \ \shortminus \ 
                            \upperD{\T_i} ( \Awesome^{\T_i}_2 \! \shortminus \! \Good^{\T_i}_2 )
                            \! \shortminus \! 
                            ( \Good^{\T_i}_2 \! \shortminus \! \Neutral_2 )
                        }
                    \\
                    & \iff
                        \label{eq:gPTG_K_high_enough_condition}
                        \frac{
                            (1-\dFT)
                        }{
                            \dFT (\Awesome^{\FT}_2 - \Good^{\FT}_2)
                            +
                            (\Good^{\FT}_2 - \Neutral_2)
                        }
                        >
                        \\
                        \nonumber
                        &
                        \frac{
                            ( 1 - \upperD{\FT} ) - ( 1 - \upperD{\T_i} )
                        }{
                            \upperD{\FT} ( \Awesome^{\FT}_2 \! \shortminus \! \Good^{\FT}_2 )
                            \! + \! 
                            ( \Good^{\FT}_2 \! \shortminus \! \Neutral_2 )
                        \ \shortminus \ 
                            \upperD{\T_i} ( \Awesome^{\T_i}_2 \! \shortminus \! \Good^{\T_i}_2 )
                            \! \shortminus \! 
                            ( \Good^{\T_i}_2 \! \shortminus \! \Neutral_2 )
                        }
                    .
                \end{align}
            (Note that for sufficiently high $\Good^\FT_2$,
                the denominator of the right-hand side of \eqref{eq:gPTG_K_high_enough_condition} will always be positive.)
            We can further simplify this inequality by observing that \eqref{eq:gPTG_K_high_enough_condition} is of the form
                $
                    \frac{a}{x}
                    >
                    \frac{a - b}{x - y}
                $,
            which is equivalent to
                $ ax - ay > ax - bx $,
            which is equivalent to
                $
                    \frac{b}{a}
                    >
                    \frac{y}{x}
                $.
            Applying this observation to \eqref{eq:gPTG_K_high_enough_condition}, we get that
                \begin{align}
                    &
                        \nonumber
                        \utility_2(\metaNE_1, \incentivise{\FT})
                        >
                        \utility_2(\metaNE_1, \incentivise{\T_i})
                    \\
                    & \iff
                        \label{eq:gPTG_K_high_enough_condition_two}
                        \frac{
                            1 - \upperD{\T_i}
                        }{
                            1 - \upperD{\FT}
                        }
                        >
                        \frac{
                            \upperD{\T_i} ( \Awesome^{\T_i}_2 \shortminus \Good^{\T_i}_2 )
                            +
                            ( \Good^{\T_i}_2 \shortminus \Neutral_2 )
                        }{
                            \upperD{\FT} ( \Awesome^{\FT}_2 \shortminus \Good^{\FT}_2 )
                            + 
                            ( \Good^{\FT}_2 \shortminus \Neutral_2 )
                        }
                    .
                \end{align}
            
            \noindent
            (Note that $\upperD{\FT}$ is strictly lower than $\upperD{\T_i}$ by \Cref{lem:gPTG_reduction_to_Ti},
                so for the inequality to hold, the right-hand side denominator must be strictly larger than the numerator.
                This implies that if this inequality holds,
                    the transformation from \eqref{eq:gPTG_K_high_enough_condition} to \eqref{eq:gPTG_K_high_enough_condition_two}
                    was justified.)
            
            Finally, we observe that the inequality \eqref{eq:gPTG_K_high_enough_condition_two} holds
                for any sufficiently high $\Good^\FT_2 := \utility_2(\FT, \C)$
                (and
                    $\Awesome^\FT_2 := \utility_2(\FT, \D) > \Good^\FT_2$,
                    which is required for $\game$ to qualify as a generalised PTG).
            To see this, note that
                by \Cref{lem:gPTG_reduction_to_Ti}, the probabilities $\upperD{\T}$ only depend on the payoffs of \Plone{}.
            This means that the left hand-side of \eqref{eq:gPTG_K_high_enough_condition_two} is a positive constant,
                and the right-hand side can be made arbitrarily small by choosing $\Good^\FT_2$
                    (and $\Awesome^\FT_2 \geq \Good^\FT_2$)
                that is sufficiently high (relative to the other payoffs).
            As a concrete bound, it clearly suffices if $\Good^\FT_2$ satisfies
                \begin{align*}
                    \Good^\FT_2 - \Neutral_2
                    \geq
                    \max_{i = 1, \dots, n}
                        \frac{\Awesome^{\T_i}_2 - \Neutral_2}{1 - \upperD{\T_i}}
                    .
                \end{align*}
            This concludes the proof of \Cref{cl:gPTG_subgame_NE_is_NE}.
    \end{proof}

    \vk{A corollary of \Cref{cl:gPTG_subgame_NE_is_NE} is that
        (1) The strategies $\T_1, \dots, \T_n$ are strictly dominated against \emph{any} convex combination of $\incentivise{\FT}$ and $\D$, not just the one corresponding to the $\evilProb$ from \Cref{cl:gPTG_subgame_NE}.
        (2) The strategies $\incentivise{\T_1}, \dots, \incentivise{\T_n}$ are strictly dominated against \emph{any} convex combination of $\incentivise{\T_1}$ and $\mSim$, not just the one corresponding to the $\simProb$ from \Cref{cl:gPTG_subgame_NE}.
        However, by being more general about the proof, we could try to show
            --- or rather, explore the conditions necessary for ---
            the non-existence of other equilibria, where the players use other strategies.
    }

    This concludes the proof of \Cref{thm:generalised_PTG}.
\end{proof}

To simplify potential re-use of this result,
    the following technical corollary (of the proof of \Cref{thm:generalised_PTG})
    gathers the key definitions and requirements into one place.

\begin{corollary}[The form of the ``second-most trust'' simulation equilibrium of a generalised partial-trust game]\label{cor:characterisation_of_PTG_NE}
    Let $\game$ be a generalised PTG
    with payoffs
        \begin{align*}
            (\Neutral_1, \Neutral_2) & := \utility(\WO, \C) = \utility(\WO, \D) \\
            (\Good^\T_1, \Good^\T_2) & := \utility(\T, \C) \\
            (\Bad^\T_1, \Awesome^\T_2) & := \utility(\T, \D)
            .
        \end{align*}
    Denote
        $\FT = \argmax_{\T \in \PureStrategies^\game_1} \, \utility_1(\FT, \C)$
    and suppose that the set of all non-redundant actions of \Plone{}
        (in the sense of \Cref{lem:gPTG_reduction_to_Ti})
    is
        \begin{align*}
            \PureStrategies^{\game'}_1 := \{ \FT, \T_1, \dots, \T_n, \WO \}
            .
        \end{align*}
    Let $\upperD{\T} \in (0, 1)$ be the numbers whose ratios $\frac{\upperD{\T}}{1-\upperD{\T}}$ satisfy
        \begin{align*}
            \upperD{\FT} : (1 - \upperD{\FT})
                & :=
                (\Good^{\FT}_1 - \Good^{\T_1}_1) : (\Bad^{\T_1}_1 - \Bad^{\FT}_1)
            \\
            \upperD{\T_i} : (1 - \upperD{\T_i})
                & :=
                (\Good^{\T_i}_1 - \Good^{\T_{i+1}}_1) : (\Bad^{\T_{i+1}}_1 - \Bad^{\T_i}_1)
            \\
            \upperD{\T_n} : (1 - \upperD{\T_n})
                & :=
                (\Good^{\T_n}_1 - \Neutral_1) : (\Neutral_1 - \Bad^{\T_n}_1)
        \end{align*}
    and denote
        \begin{align*}
            \incentivise{\T} & := \upperD{\T} \cdot \D + (1-\upperD{\T}) \cdot \C
            .
        \end{align*}
    Recall that by \Cref{lem:gPTG_reduction_to_Ti}, we have
        \begin{align*}
            \utility_2(\T, \incentivise{\T})
            =
            \max \, \{
                \utility_2(\T, \mixedStrategy_2)
                \mid
                \mixedStrategy_2 \in \MixedStrategies^\game_2
                , \,
                \br(\mixedStrategy_2) \ni \T
            \}.
        \end{align*}
    Finally, let $\metaNE$ be the strategy in $\msimgame$ given by
        \begin{align*}
            \metaNE_1 & := \simProb \cdot \mSim + (1-\simProb) \cdot \T_1
            \\
            \metaNE_2 & := \evilProb \cdot \D + (1-\evilProb) \cdot \incentivise{\FT}
            \ ,
        \end{align*}
    where
        \begin{align*}
            \simProb : (1 - \simProb)
                & :=
                \frac{
                    (1-\dFT) (\Awesome^{\T_1}_2 - \Good^{\T_1}_2)
                }{
                    \dFT (\Awesome^{\FT}_2 - \Good^{\FT}_2)
                    +
                    (\Good^{\FT}_2 - \Neutral_2)
                }
            \\
            \evilProb
                & :=
                \frac{\simcost}{\Neutral_1 - \Bad^{\T_1}_1}
                .
        \end{align*}

    Then $\metaNE$ is a simulation equilibrium of $\msimgame$
    if and only if
        \begin{align*}
            & \forall i \in \{ 1, \dots, n \} : \\
            & \phantom{aasdfaaa}
                \frac{
                    \upperD{\T_i} ( \Awesome^{\T_i}_2 \shortminus \Good^{\T_i}_2 )
                    +
                    ( \Good^{\T_i}_2 \shortminus \Neutral_2 )
                }{
                    \upperD{\FT} ( \Awesome^{\FT}_2 \shortminus \Good^{\FT}_2 )
                    + 
                    ( \Good^{\FT}_2 \shortminus \Neutral_2 )
                }
                \leq
                \frac{
                    1 - \upperD{\T_i}
                }{
                    1 - \upperD{\FT}
                }
            \ .
        \end{align*}
\end{corollary}

\begin{proof}
    This is an immediate corollary of the proof of \Cref{thm:generalised_PTG}.
\end{proof}

\section{Proofs \texorpdfstring{for \Cref{sec:sub:positive_coordination} (Coordination Games)}{Related to Coordination Games}}\label{sec:app:coordination}

In this section, we present the proofs related to claims about coordination games.

\lemmaCoordinationGameProperties*

\begin{proof}
    This is standard, and follows, for example, from the proof of Proposition~16 in \citet{kovarik2023game}.
    The cited proof shows that $
            \supp (\mixedStrategy^{\actionSubset}_\pl)
            % = \supp (\mixedStrategy^{*,\actionSubset}_2)
            = \{(a_k,b_k) : k \in S\}
        $ and also that
    $
        \utility_\pl(\mixedStrategy^{\actionSubset})
        =
        \left(
            \sum_{k \in \actionSubset} \frac{1}{\utility_\pl(\actionPl^k, \actionOpp^k)}
        \right)^{-1}
    $,
    from which the second part of the lemma immediately follows.
\end{proof}

\propCoordinationGames*
\begin{proof}
    \textit{The proof of (i):}
    We will prove (i) by separately considering two cases:
        (1) The case where \Pltwo{} has at least two distinct optimal pure-commitments.
        (2) The case where \Pltwo{} has a unique optimal pure-commitment.

    (1) Suppose that
            $
                \actionSubset^*
                :=
                \argmax
                    \{
                    \utility_2(\actionPl^k, \actionOpp^k)
                    \mid
                        k \in \{ 1, \dots, \nOfActions \}
                    \}
            $
        has at least two different elements.
        We will show that $\msimgame$ has some
            $\metaNE$ with $\utility_2(\metaNE) = \max_k \utility_2 (\actionPl^k, \actionOpp^k)$;
        To find $\metaNE$, pick any two distinct elements $k^*, l^*$ of $\actionSubset^*$
            and let
                $
                    \metaNE_2
                    :=
                        \frac{1}{2} \cdot \hat \actionOpp^{k^*}
                        +
                        \frac{1}{2} \cdot \hat \actionOpp^{l^*}
                $.
        Clearly, the only candidates for \Plone{}'s best-response to $\metaNE_2$ are $\actionPl^{k^*}$, $\actionPl^{k^*}$, $\mSim$
            (since all other actions yield $\utility_1(\actionPl^k, \metaNE_2) = 0$).
        Comparing the utilities corresponding to $\actionPl^{k^*}$ (or $\actionPl^{l*}$) and $\mSim$, we get
            \begin{align*}
                & \utility_1(\mSim, \metaNE_2) > \utility_1(\actionPl^{k^*}, \metaNE_2)
                \\
                & \iff
                    \frac{1}{2} \cdot \utility_1(\actionPl^{k^*}, \actionOpp^{k^*})
                    +
                    \frac{1}{2} \cdot \utility_1(\actionPl^{l^*}, \actionOpp^{l^*})
                    - \simcost
                    >
                    \\
                    & \phantom{\iff}
                    > \ \ 
                    \frac{1}{2} \cdot \utility_1(\actionPl^{k^*}, \actionOpp^{k^*})
                    +
                    \frac{1}{2} \cdot \Bad_1
                \\
                & \Longleftarrow
                    \frac{1}{2} \cdot \min_{k\leq \nOfActions}
                        \utility_1(\actionPl^k, \actionOpp^k)
                    - \simcost
                    >
                    \frac{1}{2} \cdot \Bad_1
                \\
                & \iff
                    \simcost
                    < \frac{1}{2} \left(
                        \min_{k\leq \nOfActions}
                            \utility_1(\actionPl^k, \actionOpp^k)
                        -
                        \Bad_1
                    \right)
                .
            \end{align*}
        Since the right-hand side of the latest inequality is a positive constant,
            this shows that for any sufficiently low $\simcost$,
            $\msimgame$ has a simulation equilibrium that satisfies $\utility_2(\metaNE) = \max_{k\leq \nOfActions} \utility_2(\actionPl^k, \actionOpp^k)$.

    (2) Suppose that $\game$ as a unique optimal pure-commitment for \Pltwo{}, denoted $\actionOpp^{k_2}$.
        Pick some $k_1 \neq k_2$.
        We claim that $\msimgame$ has an \NE{} $\metaNE$ of the form
            \begin{align*}
                \metaNE_1
                & =
                \simProb \cdot \mSim + (1-\simProb) \actionPl^{k_1}
                \\
                \metaNE_2
                & =
                \evilProb \cdot \actionOpp^{k_2} + (1-\evilProb) \actionOpp^{k_1}
                ,
            \end{align*}
            where
            \begin{align*}
                \simProb
                    & =
                    \frac{\utility_2(\actionPl^{k_1}, \actionOpp^{k_1})}{\utility_2(\actionPl^{k_2}, \actionOpp^{k_2})}
                \\
                \evilProb
                    & =
                    \frac{\simcost}{\utility_1(\actionPl^{k_2}, \actionOpp^{k_2})}
                .
            \end{align*}
        To see that $\metaNE$ is an \NE{}, note that:
            \Pltwo{} cannot have any profitable deviation from $\metaNE$
                because for $k \notin \{k_1, k_2\}$, we have
                \begin{align*}
                    \utility_2(\metaNE_1, \actionOpp^k)
                    & =
                    \simProb \cdot \utility_2(\actionPl^k, \actionOpp^k) + (1-\simProb) \cdot 0
                    \\
                    & <
                    \simProb \cdot \utility_2(\actionPl^{k_2}, \actionOpp^{k_2}) + (1-\simProb) \cdot 0
                    \\
                    & =
                    \utility_2(\metaNE_1, \actionOpp^{k_2})
                    =
                    \utility_2(\metaNE)
                    .
                \end{align*}
            For \Plone{},
                the only $\actionPl^k \notin \supp{\metaNE_1}$ which can give $\utility_1(\actionPl^k, \metaNE_2) > 0$
                is $\actionPl^{k_2}$.
                However, the corresponding utility satisfies
                    \begin{align*}
                        \utility_1(\actionPl^{k_2}, \metaNE_2)
                        & =
                        \evilProb \cdot \utility_1(\actionPl^{k_2}, \actionOpp^{k_2})
                        \\
                        & =
                            \frac{\simcost}{\utility_1(\actionPl^{k_2}, \actionOpp^{k_2})}
                            \cdot
                            \utility_1(\actionPl^{k_2}, \actionOpp^{k_2})
                        = \simcost
                        .
                    \end{align*}
                We see that
                    \begin{align*}
                        \utility_1(\metaNE)
                        & =
                            \utility_1(\actionPl^{k_1}, \metaNE_2)
                        \\
                        & =
                            (1-\evilProb) \cdot
                                \utility_1(\actionPl^{k_1}, \actionOpp^{k_1})
                            + \evilProb \cdot 0
                        \\
                        & =
                            (1-\frac{\simcost}{\utility_1(\actionPl^{k_2}, \actionOpp^{k_2})}) \cdot
                                \utility_1(\actionPl^{k_1}, \actionOpp^{k_1})
                        ,
                    \end{align*}
                    which shows that $\utility_1(\metaNE)$ converges to $\utility_1(\actionPl^{k_1}, \actionOpp^{k_1})$
                    as $\simcost \to 0_+$.
                This implies that for sufficiently low $\simcost$,
                    $\actionPl^{k_2}$ is not a profitable deviation for \Plone{},
                and thus $\metaNE$ is a simulation equilibrium of $\msimgame$.
            This concludes the proof of the case (2) of (i),
            and thus concludes the whole proof of (i).

    \textit{The proof of (ii):}
        Let $\metaNE$ be some simulation equilibrium of $\msimgame$.
        First, note that the inequality
            \begin{align*}
                \utility_2(\metaNE)
                \leq
                \max_{k\leq n} \utility_2(\actionPl^k, \actionOpp^k)
            \end{align*}
            holds trivially,
                since $\max_{k\leq n} \utility_2(\actionPl^k, \actionOpp^k)$ is the maximum that \Pltwo{} can gain in the whole game.
            The inequality
            \begin{align*}
                \utility_1(\metaNE)
                <
                \max_{k\leq n} \utility_1(\actionPl^k, \actionOpp^k)
            \end{align*}
            holds for the same reason
                (since in any simulation equilibrium, \Plone{} is putting a positive probability on $\mSim$,
                    which decreases their utility by $\simcost > 0$).

        Second,
            observe that when $\metaNE_1(\mSim) > 0$,
            \Pltwo{}'s mixed-meta strategy $\metaNE_2$ can only mix over \textit{pure} strategies.
            (This is because if we had $\metaNE_2(\mixedStrategy_2)>0$ for some non-pure $\mixedStrategy_2$,
                \Pltwo{} could strictly improve their utility by replacing $\mixedStrategy_2$
                by its constituent strategies
                -- i.e., by
                    $
                        \sum_{k=1}^\nOfActions
                            \mixedStrategy_2(\actionOpp^k) \meta{\actionOpp}^k
                    $.)

        From this, we can derive the inequality
            \begin{align}\label{eq:coordination_plone_gains}
                \utility_1(\metaNE)
                >
                \utility_1(\mixedStrategy^{\{1, \dots, \nOfActions\}})
                .
            \end{align}
            To see this, note that because $\metaNE$ is a \NE{} that puts non-zero probability on $\mSim$, we have
                \begin{align*}
                    \utility_1(\metaNE)
                    & =
                    \utility_1(\mSim, \metaNE)
                    \\
                    & =
                    \sum_{k=1}^\nOfActions
                        \mixedStrategy_2(\actionOpp^k) \utility_1(\mSim, \meta{\actionOpp}^k)
                    \\
                    & =
                    \sum_{k=1}^\nOfActions
                        \mixedStrategy_2(\actionOpp^k) \utility_1(\actionPl^k, \actionOpp^k)
                        - \simcost
                    \\
                    & \geq
                    \min_{k\leq \nOfActions}
                        \utility_1(\actionPl^k, \actionOpp^k)
                        - \simcost
                    .
                \end{align*}
                By (the proof of) \Cref{lem:coordination_game_properties},
                    we have
                    \begin{align*}
                        \utility_\pl(\mixedStrategy^{\actionSubset})
                        & =
                            \frac{1}{
                                \sum_{k = 1}^\nOfActions \frac{1}{\utility_\pl(\actionPl^k, \actionOpp^k)}
                            }
                        \\
                        & >
                            \frac{1}{
                                \frac{1}{
                                    \min_{k\leq \nOfActions} \utility_\pl(\actionPl^k, \actionOpp^k)
                                }
                            }
                        \\
                        & =
                            \min_{k\leq \nOfActions} \utility_\pl(\actionPl^k, \actionOpp^k)
                        .
                    \end{align*}
                This shows that for sufficiently small $\simcost$,
                    the inequality \eqref{eq:coordination_plone_gains} holds.

        To finish the proof of (ii),
            it remains to show that
            \begin{align}\label{eq:coordination_pltwo_gains}
                \utility_2(\metaNE)
                >
                \utility_2(\mixedStrategy^{\{1, \dots, \nOfActions\}})
                .
            \end{align}
            This follows from the fact that
                    $\mixedStrategy^{\{1, \dots, \nOfActions\}}_1$
                is in fact the strategy which minimises $\utility_2( \mixedStrategy_1, \br(\mixedStrategy_1) )$.
                (From the proof of \Cref{lem:coordination_game_properties}, we know that
                    \begin{align*}
                        \utility_2( \mixedStrategy^{\{1, \dots, \nOfActions\}})
                        =
                        \min_{\actionSubset \subseteq \{1, \dots, \nOfActions\}}
                            \utility_2(\mixedStrategy^{\actionSubset})
                        .
                    \end{align*}
                    Moreover, since $\mixedStrategy^{\{1, \dots, \nOfActions\}}$ is a \NE{},
                        $\mixedStrategy^{\{1, \dots, \nOfActions\}}_1$ makes \Pltwo{} is indifferent between all actions $\actionOpp^k$,
                        and thus increasing the probability that \Plone{} puts on any $\actionPl^k$
                        would strictly increase the expected utility corresponding to $\actionOpp^k$.)
            However, we already saw that any simulation must have \Plone{} only mixing between pure strategies,
                so we have
                    \begin{align*}
                        \utility_2(\mSim, \metaNE_2)
                        \geq
                        \min_{k\leq \nOfActions} \utility_2(\actionPl^k, \actionOpp^k)
                        >
                        \utility_2( \mixedStrategy^{\{1, \dots, \nOfActions\}} )
                        .
                    \end{align*}
            Since any \textit{simulation} equilibrium $\metaNE$ plays $\mSim$ with a non-zero probability,
                it follows that any such $\metaNE$ must satisfy
                    $
                        \utility_2(\metaNE)
                        >
                        \utility_2(\mixedStrategy^{\{1, \dots, \nOfActions\}})
                    $.
            This concludes the proof of (ii).

    \textit{The proof of (iii):}
        In the case (1) of the proof of (i),
            we saw that there exists a simulation equilibrium where
                \begin{align*}
                    \utility_2(\metaNE)
                    =
                    \max_{k\leq \nOfActions}
                        \utility_2(\actionPl^k, \actionOpp^k)
                    .
                \end{align*}
        This shows that the ``unless'' clause in the statement is necessary.
        For the ``$<$'' part of (iii), suppose that \Pltwo{} only has a unique optimal pure-commitment, denoted $\actionOpp^{k_2}$.
        Clearly, we can only have
            \begin{align*}
                \utility_2(\metaNE)
                =
                    \max_{k\leq \nOfActions}
                        \utility_2(\actionPl^k, \actionOpp^k)
                =
                    \utility_2(\actionPl^{k_2}, \actionOpp^{k_2})
            \end{align*}
        if $\metaNE_2(\actionOpp^{k_2}) = 1$.
        However, such $\metaNE_2$ would have a unique best-response for \Plone{}
            -- $\actionPl^{k_2}$ --
            which would imply that $\metaNE$ would be the pure equilibrium $(\actionPl^{k_2}, \actionOpp^{k_2})$.
        This shows that no simulation equilibrium of $\msimgame$ satisfy
            $
                \utility_2(\metaNE)
                =
                \max_{k\leq \nOfActions}
                    \utility_2(\actionPl^k, \actionOpp^k)
            $,
        which concludes the proof of (iii).

    To finish the whole proof, it remains to show that (iii) holds in this case as well.
    This holds because
        \begin{align*}
            \utility_2(\metaNE)
            & =
                \utility_2(\metaNE_1, \actionOpp^{k_0})
            \\
            & =
                \simProb \cdot \utility_2(\mSim, \actionOpp^{k_0})
                +
                (1-\simProb) \cdot \utility_2(\actionPl^{k_0}, \actionOpp^{k_0})
            \\
            & =
                \utility_2(\actionPl^{k_0}, \actionOpp^{k_0})
            \\
            & <
                \utility_2(\actionPl^{k_1}, \actionOpp^{k_1})
            =
                \max_{k \leq n} \utility_2(\actionPl^{k}, \actionOpp^{k})
            .
        \end{align*}
\end{proof}

\begin{figure}
    \centering
    \begin{NiceTabular}{rccccc}[cell-space-limits=3pt]
                    & $\actionOpp^1$     & $\actionOpp^2$ & $\actionOpp^3$ & $\OO$ & $\ \ \ $\\
        $\actionPl^1$ & \Block[hvlines]{4-4}{}
                        $\game_1$           & $\Bad_1, \Bad_2$  & $\Bad_1, \Bad_2$  & $\Bad_1, \Bad_2 \!\!+\!\epsilon$ \\
        $\actionPl^2$ &   $\Bad_1, \Bad_2$    & $\game_2$         & $\Bad_1, \Bad_2$  & $\Bad_1, \Bad_2 \!\!+\!\epsilon$ \\
        $\actionPl^3$ &   $\Bad_1, \Bad_2$    & $\Bad_1, \Bad_2$  & $\game_3$         & $\Bad_1, \Bad_2 \!\!+\!\epsilon$ \\
        $O$ &         $\Bad_1 \!\!+\! \epsilon, \Bad_2$
                                            & $\Bad_1 \!\!+\! \epsilon, \Bad_2$
                                                                & $\Bad_1 \!\!+\! \epsilon, \Bad_2$
                                                                    & $\Bad_1 \!\!+\! \epsilon, \Bad_2 \!\!+\! \epsilon$
    \end{NiceTabular}
    \begin{NiceTabular}{rccc}[cell-space-limits=3pt]
        $\game_k$            & $\Cooperate$     & $\Defect$ & $\ \ \ $\\
        $\Trust$    & \Block[hvlines]{2-2}{}
                        $\Good^k_1, \Good^k_2$          & $\Horrible^k_1, \Awesome^k_2$ & $\ \ \ $ \\
        $\WalkOut$  &   $\Neutral^k_1, \Horrible^k_2$   & $\Neutral^k_1, \Horrible^k_2$ & $\ \ \ $ 
    \end{NiceTabular}
    \caption{A generalised form of a trust-and-coordination game.
        The general version of the game has $n+1$ actions for each player,
            with the joint action $(\actionPl^k, \actionOpp^k)$ leading into the subgame $\game_k$.
        We assume that
            $\Horrible^k_1 < \Bad_1 < \Neutral^k_1 < \Good^k_1$
            and
            $\Horrible^k_2 < \Bad_2 < \Good^k_1 < \Awesome^k_1$
        holds for every $k$
        and $\epsilon > 0$ is small enough to not affect this ordering.
        Note that each of the subgames $\game_k$ is a trust game, with equilibrium $(\WalkOut, \Defect)$.
        As a result, the only NE of the large game is for each of the players to take the $\OptOut$.
    }
    \Description{A generalised form of a trust-and-coordination game.}
    \label{fig:gDTG}
\end{figure}

\thmGenDTG*

We first give a high-level sketch of the proof, and then present the formal proof itself.

\begin{proof}[Proof sketch]
    Because the $(\WalkOut, \Defect)$ equilibria of the trust subgames are undesirable for both players,
        the only non-simulation NE of a trust-and-coordination game is for both players to opt out.
    With simulation, this dynamic changes:
    if the players successfully coordinated and always played into the same subgame,
        there would be no simulation equilibrium
        because Alice would not gain any information by simulating.
    
    However, when Bob mixes between SE of two different subgames,
        simulation can be worthwhile for Alice.
    This can happen either when two subgames are equally good for Bob,
        in which case he can randomise between them just to force Alice to simulate.
    
    Alternatively, it can happen when Alice and Bob have different preferences over subgames,
        in which case the equilibrium can consist of
            Bob mostly playing the SE of \textit{Alice's} favourite subgame,
            but sometimes deviating and playing the SE of some game that favours him more.
        If the frequency of deviation is just right, Alice will be indifferent
            between ``blindly'' playing her favourite subgame and trusting Bob there
            and simulating to find out which subgame Bob plays.
    Note that for this equilibrium to exist,
        Bob must not find it profitable to try exploiting Alice by defecting in her favourite subgame
        -- this makes it important that the $(\WalkOut, \Defect)$ outcome is highly undesirable for Bob.
    
    Note also that,
        unlike the simulation equilibrium of the partial-trust game from \Cref{fig:PTG}
        this simulation equilibrium is not evolutionarily stable.
        (This is because if Bob randomised too much,
            Alice would switch to simulating all the time,
            incentivising Bob to in turn switch to playing into his favourite subgame all the time.
        And at this point,
            Alice has no reason to simulate anymore,
            the simulation equilibrium collapses,
            and both players $\OptOut$.)
\end{proof}

% \noindent
% The proof of \Cref{thm:generalised_DTG} includes \vk{nothing, actually. Nevermind.}

\begin{proof}
    Let $\game$ be a generalised trust-and-coordination game as in \Cref{ex:generalised_trust_and_coordination},
        corresponding to subgames $\game_1$, \dots, $\game_n$.
    For any $k \in \{ 1, \dots, n \}$, let $(\T, \stackelberg{k})$ denote the Stackelberg equilibrium of the subgame $\game_k$.
    Pick some 
        $
            \deviateIndex
            \in
            \argmax_k \utility^{\game_k}_2(\T, \stackelberg{k})
        $
        and
        $
            \baselineIndex
            \in
            \utility^{\game_k}_1(\T, \stackelberg{k})
        $.
        (By the assumption of the theorem, we have
            $
                \utility^{\game_\deviateIndex}_1(\T, \stackelberg{\deviateIndex})
                <
                \max_k
                    \utility^{\game_k}_1(\T, \stackelberg{k})
            $
            and
            $
                \utility^{\game_\baselineIndex}_2(\T, \stackelberg{\baselineIndex})
                <
                \max_k
                    \utility^{\game_k}_2(\T, \stackelberg{k})
            $.)
    We will show that $\msimgame$ allows a Nash equilibrium where:
        \begin{itemize}
            \item \Plone{} mixes between playing $(\actionPl^\baselineIndex, \Trust)$ and $\mSim$.
            \item \Pltwo{} mixes between playing
                $(\actionOpp^\baselineIndex, \stackelberg{\baselineIndex})$ and 
                $(\actionOpp^\deviateIndex, \stackelberg{\deviateIndex})$.
        \end{itemize}
    For sufficiently low $\simcost$, this NE will automatically constitute a strict Pareto-improvement over any NE of $\game$.
        (Since the only NE of the original game is for both players to take the opt-out action, resulting in utilities $\Bad_1$ and $\Bad_2$.)

    Denote by $\metaNE$ the strategy profile of the form
    \begin{align*}
        \metaNE_1
        & =
        \simProb \cdot \mSim + (1 - \simProb) \cdot (\actionPl^\baselineIndex, \T)
        \\
        \metaNE_2
        & =
        \evilProb \cdot (\actionOpp^\deviateIndex, \stackelberg{\deviateIndex})
        + (1 - \evilProb) \cdot (\actionOpp^\baselineIndex, \stackelberg{\baselineIndex})
        ,
    \end{align*}
    where $\simProb, \evilProb \in (0, 1)$
    are such that
        \begin{align*}
            \utility_1(\mSim, \metaNE_2)
            & =
            \utility_1((\actionPl^\baselineIndex, \T), \metaNE_2),
            \\
            \utility_2(\metaNE_1, (\actionOpp^\baselineIndex, \stackelberg{\baselineIndex}))
            & =
            \utility_2(\metaNE_1, (\actionOpp^\deviateIndex, \stackelberg{\deviateIndex}))
            .
        \end{align*}
    (Such $\simProb$ and $\evilProb$ exist since
        the Stackelberg values corresponding to the subgames $\game_k$ satisfy
        $\SEvalue{\baselineIndex}{1} > \SEvalue{\deviateIndex}{1}$
        and
        $\SEvalue{\deviateIndex}{2} > \SEvalue{\baselineIndex}{2}$.)
    We will prove that $\metaNE$ is an \NE{} of $\msimgame$.

    We will start by showing that \Plone{} has no profitable deviations from $\metaNE$.
    Note that $\metaNE$ gives \Plone{} utility
    \begin{align}
        \utility_1(\metaNE)
            & =
            \utility_1(\mSim, \metaNE_2)
            \label{eq:C_and_T_u1}
        \\
            & =
            \evilProb \cdot \SEvalue{\deviateIndex}{1} + (1-\evilProb) \cdot \SEvalue{\baselineIndex}{1} - \simcost
            .
            \nonumber
    \end{align}
    In contrast,
        taking any of the actions $\actionPl^k$,
            $k \notin \{ \deviateIndex, \baselineIndex \}$,
        will yield utility $\Bad_1$,
        and similarly for the opt-out action.
    As a result, the deviations worth considering for \Plone{} must involve starting with either $\actionPl^\deviateIndex$ and $\actionPl^\baselineIndex$.
    Moreover, within each of the subgames $\game_\deviateIndex$, $\game_\baselineIndex$,
        \Pltwo{} plays the \SE{}
        and $\T$ is the corresponding best response by \Plone{}.
    Since $(\actionPl^\baselineIndex, \T)$ is already one of the strategies that \Plone{} uses in $\metaNE$,
        the only possible remaining deviation is $(\actionPl^\deviateIndex, \T)$.

    To see that $(\actionPl^\deviateIndex, \T)$ is not a profitable deviation for \Plone{}, against $\metaNE_2$,
    note that
    \begin{align}
        \evilProb
        =
        \frac{\simcost}{\SEvalue{\deviateIndex}{1} - \Bad_1}
        \label{eq:C_and_T_deviate_prob}
    \end{align}
    (we derive this from requiring that $\utility_1(\mSim, \metaNE_2) = \utility_1( (\actionPl^\baselineIndex, \T), \metaNE_2)$).
    Comparing $\utility_1(\metaNE)$ (from eq. \ref{eq:C_and_T_u1}) and $\utility_1( (\actionPl^\deviateIndex, \T), \metaNE_2)$, we get
    \begin{align*}
        \utility_1(\metaNE)
        & =
            \evilProb \cdot \SEvalue{\deviateIndex}{1}
            +
            (1-\evilProb) \cdot \SEvalue{\baselineIndex}{1}
            \\
            & \phantom{ = \evilProb \cdot \SEvalue{\deviateIndex}{1} \ }
            -
            \simcost
        \\
        & =
            \evilProb \cdot \SEvalue{\deviateIndex}{1}
            +
            (1-\evilProb) \cdot \SEvalue{\baselineIndex}{1}
            \\
            & \phantom{ = \evilProb \cdot \SEvalue{\deviateIndex}{1} \ }
            -
            \evilProb \cdot ( \SEvalue{\deviateIndex}{1} - \Bad_1 ) 
        \\
        \utility_1( (\actionPl^\deviateIndex, \T), \metaNE_2)
        & =
            \evilProb \cdot \SEvalue{\deviateIndex}{1}
            +
            (1-\evilProb) \cdot \Bad_1
        .
    \end{align*}
    It follows that the former is strictly higher
        if and only if
        $
            (1-\evilProb) \cdot ( \SEvalue{\deviateIndex}{1} - \Bad_1 )
            >
            \evilProb \cdot ( \SEvalue{\deviateIndex}{1} - \Bad_1 )
        $,
        which holds if and only if $\evilProb < \frac{1}{2}$.
    By \eqref{eq:C_and_T_deviate_prob}, this holds if and only if
        $
            \simcost
            <
            \frac{1}{2}
            \cdot
            ( \SEvalue{\deviateIndex}{1} - \Bad_1 )
        $.
    In other words, we see that for sufficiently low $\simcost$, \Plone{} has no profitable deviation from $\metaNE$.

    Second, we show that \Pltwo{} has no profitable deviation from $\metaNE$.
    To start with, note that $\metaNE$ gives \Pltwo{} utility
        \begin{align*}
            \utility_2(\metaNE)
            =
            \utility_2(\metaNE_1, (\actionOpp^\baselineIndex, \stackelberg{\baselineIndex}) )
            =
            \SEvalue{\baselineIndex}{2}
            .
        \end{align*}
    This means that \Pltwo{} cannot benefit from using the opt-out action.
    Moreover, if \Pltwo{} deviates to any $\actionOpp^k$ with $k\neq \baselineIndex$,
        they can only benefit when \Plone{} takes the $\mSim$ action
        (otherwise they reach the miscoordination outcome, yielding utility $\Bad_2$ to \Pltwo{}).
    It follows that in any game other than $\game_\baselineIndex$,
        \Pltwo{} only needs to consider the Stackelberg strategy $\stackelberg{k}$
        (using any other strategy yields
            $
                \utility^{\game_k}_2(\fbr, \mixedStrategy'_2)
                <
                \SEvalue{k}{2}
            $,
            and
            is strictly dominated by using $\stackelberg{k}$).
    Moreover, in the game $\game_\baselineIndex$, the only strategy worth considering for \Pltwo{}
        (other than $\stackelberg{\baselineIndex}$, which they are already using)
        is defecting with probability 100\%.
        (This is because
            $\stackelberg{\baselineIndex}$ is dominant among the strategies which induce the best response $\T$
            and $100\% \, \D$ is dominant among all other strategies, which induce the best response $\WO$.)

    Note that among the potential deviations of the form
        $(\actionOpp^k, \stackelberg{k})$
        for $k \neq \baselineIndex$,
    no deviation can be strictly more profitable than the strategy 
        $(\actionOpp^\deviateIndex, \stackelberg{\deviateIndex})$.
    This is because
        \begin{align*}
            \utility_2(\metaNE_1, (\actionOpp^k, \stackelberg{k}) )
            =
                \simProb \cdot \SEvalue{k}{2}
                +
                (1-\simProb) \cdot \Bad_1
        \end{align*}
        and $\game_\deviateIndex$ was chosen as one of the subgames with highest value of $\SEvalue{k}{2}$.\footnotemark{}
            \footnotetext{
                This would no longer be true if the miscoordination payoff
                $\utility_2(\actionPl^k, \actionOpp^l)$,
                $k\neq l$,
                had non-trivial dependence on $k$ and $l$.
                In such a more general version of the trust-and-coordination game, the simulation equilibria might not necessarily involve the subgame with the highest Stackelberg value for \Pltwo{}.
            }
    However, since the strategy $(\actionOpp^\deviateIndex, \stackelberg{\deviateIndex})$ is already a part of $\metaNE_2$, it cannot constitute a profitable deviation for \Pltwo{}.
    
    To prove that $\metaNE$ is an \NE{} of $\msimgame$, it thus only remains to check that
        $(\actionOpp^\baselineIndex, \D)$ is not a profitable deviation for \Pltwo{}.
    To verify this, note first that
    \begin{align}\label{eq:C_and_T_sim_probability}
        \simProb : (1-\simProb)
        =
            (\SEvalue{\baselineIndex}{2} - \Bad_2)
            :
            (\SEvalue{\deviateIndex}{2} - \Bad_2)
    \end{align}
        (this follows from the requirement that
            \Pltwo{} should be indifferent between
            $(\actionOpp^\baselineIndex, \stackelberg{\baselineIndex})$
            and
            $(\actionOpp^\deviateIndex, \stackelberg{\deviateIndex})$).
    The NE utility for \Pltwo{} is
    \begin{align*}
        \utility_2(\metaNE)
        & =
            \utility_2(\metaNE_1, \metaNE_2)
        \\
        & =
            \simProb \cdot \utility_2(\br, (\actionOpp^\baselineIndex, \stackelberg{\baselineIndex}))
            \\
            & \phantom{ = \ \ }
            +
            (1-\simProb) \cdot
                \utility^{\game_\baselineIndex}_2(\T, \stackelberg{\baselineIndex})
        \\
        & =
            \simProb \cdot \SEvalue{\baselineIndex}{2}
            +
            (1-\simProb) \cdot \SEvalue{\baselineIndex}{2}
        \\
        & =
            \SEvalue{\baselineIndex}{2}
        .
    \end{align*}
    In comparison, the deviation utility for \Pltwo{} is
    \begin{align*}
        \utility_2(\metaNE_1, (\actionOpp^\baselineIndex, \D))
        & =
            \simProb \cdot \utility^{\game_\baselineIndex}_2(\WO, \D)
            \\
            & \phantom{ = \ \ }
            +
            (1-\simProb) \cdot \utility^{\game_\baselineIndex}_2(\T, \D)
        \\
        & =
            \simProb \cdot \Horrible^\baselineIndex_2
            + (1-\simProb) \cdot \Awesome^\baselineIndex_2
        .
    \end{align*}
    It follows that $(\actionOpp^\baselineIndex, \D)$ yields strictly lower utility than $\metaNE_2$ against $\metaNE_1$ if and only if
    \begin{align*}
        \simProb \cdot \Horrible^\baselineIndex_2
            + (1-\simProb) \cdot \Awesome^\baselineIndex_2
        <
        \SEvalue{\baselineIndex}{2}
        ,
    \end{align*}
    which is equivalent to
        \begin{align}\label{eq:C_and_T_horrible_condition}
            \Horrible^\baselineIndex_2
            <
            \frac{1}{\simProb} \left(
                \SEvalue{\baselineIndex}{2}
                -
                (1-\simProb) \Awesome^\baselineIndex_2
            \right)
            .
        \end{align}
    From \eqref{eq:C_and_T_sim_probability}, we know that $\simProb$
        is a function of
            $\SEvalue{\baselineIndex}{2}$,
            $\SEvalue{\deviateIndex}{2}$, and
            $\Bad_2$,
        but not of
            $\Horrible^\baselineIndex_2$.
    It follows that for sufficiently low $\Horrible^\baselineIndex_2$,
    \Pltwo{} has no profitable deviation from $\metaNE_2$.

    This shows that $\metaNE$ is an \NE{} of $\msimgame$,
    and concludes the proof of \Cref{thm:generalised_DTG}.
\end{proof}

\begin{figure}
    \centering
    \begin{NiceTabular}{rccc}[cell-space-limits=3pt]
        & $\actionOpp^1$    & $\actionOpp^2$    & $\OO$ \\
        $\actionPl^1$ & \Block[hvlines]{3-3}{}
            \begin{tabular}{@{}c@{}} $20, 20\phantom{\shortminus}$ \ \ $\shortminus99, 40$ \\ $\phantom{\shortminus9}9, \shortminus99$ \ \ $\phantom{\shortminus9}9, \shortminus99$ \end{tabular}
            & $0, 0$
            & $0, 1$ \\
        $\actionPl^2$
            & $0, 0$
            & \begin{tabular}{@{}c@{}} $20, 20\phantom{\shortminus}$ \ \ $\shortminus99, 40$ \\ $\phantom{\shortminus}10, \shortminus99$ \ \ $\phantom{\shortminus}10, \shortminus99$ \end{tabular}
            & $0, 1$ \\
        $\OO$
            & $1, 0$
            & $1, 0$
            & $1, 1$
    \end{NiceTabular}
    \caption{\Cref{fig:DTG}, repeated for convenience.
        An example of a trust-and-coordination game,
        where coordinating on a joint action $(\actionPl^k, \actionOpp^k)$
        leads the players to a trust subgame.
    }
    \Description{An example of a trust-and-coordination game.}
    \label{fig:DTG_repeated}
\end{figure}

\corrollaryDTG*

\begin{proof}
    Since the game $\game$ from \Cref{fig:DTG} is clearly an instance of a trust-and-coordination game,
        it remains to verify that $\game$ satisfies the assumptions of \Cref{thm:generalised_DTG}.
    The first assumption holds because we have
        $\argmax \{ \SEvalue{k}{1} \,|\, k = 1, 2 \} = \{ 1 \} $,
        $\argmax \{ \SEvalue{k}{2} \,|\, k = 1, 2 \} = \{ 2 \} $.
    To verify that $\Horrible^k_2$ are ``sufficiently low relative to the other payoffs in $\game$'',
    we need to make sure that the proof of \Cref{thm:generalised_DTG} goes through with the specific payoffs in \Cref{fig:DTG}.
    Inspecting the proof, we see that the only condition which actually needs verifying
        is \eqref{eq:C_and_T_horrible_condition}, with $\baselineIndex = 1$:
        \begin{align}
            \Horrible^1_2
            <
            \frac{1}{\simProb} \left(
                \SEvalue{1}{2}
                -
                (1-\simProb) \Awesome^1_2
            \right)
            .
            \label{eq:horrible_condition_specific}
        \end{align}

    First, the SE of $\game_1$ and $\game_2$ correspond to
        the strategies that make \Plone{} indifferent between $\T$ and $\WO$.
        As a result, they are equal to
            \begin{align*}
                \stackelberg{1}
                & =
                    \left(
                        \frac{20\shortminus10}{20\shortminus10+10\shortminus(\shortminus99)},
                        \frac{10\shortminus(\shortminus99)}{20\shortminus10+10\shortminus(\shortminus99)}
                    \right)
                        \\
                    & =
                    \left( \frac{10}{119}, \frac{109}{119} \right)
                    =
                    ( \stackelberg{1}(\D), \stackelberg{1}(\C) )
                \\
                \stackelberg{2}
                & =
                    \left(
                        \frac{20\shortminus9}{20\shortminus9)+(9\shortminus(\shortminus9))},
                        \frac{9\shortminus(\shortminus99)}{20\shortminus9+9\shortminus(\shortminus99)}
                    \right)
                        \\
                    & =
                    \left( \frac{11}{119}, \frac{108}{119} \right)
                    =
                    ( \stackelberg{2}(\D), \stackelberg{2}(\C) )
                .
            \end{align*}

    This implies that
        \begin{align*}
            \SEvalue{1}{2}
            & =
                \frac{10}{119} \cdot 40
                +
                \frac{109}{119} \cdot 20
            = \frac{129}{119} \cdot 20
            \\
            \SEvalue{2}{2}
            & =
                \frac{11}{119} \cdot 40
                +
                \frac{108}{119} \cdot 20
            = \frac{130}{119} \cdot 20
        \end{align*}
    (and, unsurprisingly,
        $\SEvalue{1}{1} = 10$
        and
        $\SEvalue{2}{1} = 9$).
    By \eqref{eq:C_and_T_sim_probability}, we get
        \begin{align*}
            \simProb : (1-\simProb)
            & =
                \left( \frac{129}{119} \cdot 20 - 0 \right)
                :
                \left( \frac{130}{119} \cdot 20 - 0 \right)
            \\
            & =
                129 : 130
            ,
        \end{align*}
        implying that
        $
            \simProb = \frac{129}{259}
        $.
    Finally, we can put these values into \eqref{eq:horrible_condition_specific},
        obtaining
        \begin{align*}
            \shortminus 99
            & <
                \frac{259}{129}
                \left(
                    \frac{129}{119} \cdot 20
                    -
                    \frac{130}{259} \cdot 40
                \right)
            \\
            & =
                \frac{259}{119} \cdot 20
                -
                \frac{130}{129} \cdot 40
            \doteq
                1.7
            .
        \end{align*}
    Since this condition holds, the proof is complete.
\end{proof}

\section{Proofs \texorpdfstring{for \Cref{sec:sub:privacy} (Privacy)}{Related to Privacy}}\label{sec:app:password_guessing}

The purpose of this section
    is to present the ``password-guessing'' construction discussed in \Cref{sec:sub:privacy}
    and use it to prove \Cref{theorem:password_guessing}.

\thmMsimCanBeBetterThanPsim*

Before proceeding, we introduce the following notation:

\begin{definition}[Game with opt-out]\label{def:opt_out}
    By saying that $\game$ is a \textbf{game with an option to opt-out}, we mean that
    each player $\pl$ has access to some pure strategy $\OptOut$ ($\OO$) which satisfies
    $
        \forall \pureStrategy_\opp \in \PureStrategies^\game_\opp
        :
        \utility(\OO, \pureStrategy_\opp)
            =
            ( \Baseline_1, \Baseline_2 )
    $
    ,
    where
        \begin{align*}
            \Baseline_{i}
            <
            \min \left\{
                \utility_\pl(\pureStrategy_1, \pureStrategy_2)
            \mid
                \pureStrategy_j \in \PureStrategies^\game_j \setminus \{ \OO \}
                \textnormal{ for }
                j = 1, 2
            \right\}
        \end{align*}
    are some \textbf{B}aseline payoffs.
\end{definition}

\noindent
Note that this definition could be extended in various ways,
    for example by allowing each player to have multiple opt-out strategies,
    allowing the opt-out payoffs to depend on the opponent's actions, etc.
To simplify the exposition, this section considers this more straightforward version of \Cref{def:opt_out}.

% The password-guessing construction is as follows:

The basic building block of the password-guessing construction is the following game:

\begin{example}[Password guessing]\label{ex:password_guessing}
    The \textbf{password guessing} game $\PG{N}{x}$
            with $N \in \N$ passwords and stakes $x > 0$        
        is meant to capture a situation where
        Bob puts $\$x$ into a password-protected account
        and Alice can attempt to steal the money, but risks punishment when detected.

    Formally, $\PG{N}{x}$ works as follows:
    Bob starts by selecting a password $\password \in \{1, \dots, N \}$.
    Afterwards, Alice can
        either do nothing ($\dontGuess$), terminating the game with payoffs $(0,0)$,
        or try to guess the password, selecting some $\guess \in \{1, \dots, N \}$.
    If $\guess = \password$,
        the game terminates with payoffs
            $\utility_A = x + 1$
                (Alice stealing the money and feeling smug)
            and
            $\utility_B = -x -1$
                (Bob losing the money and feeling sad).
    If $\guess \neq \password$,
        the game terminates with payoffs
            $\utility_A = -2(x+1)$
                (Alice getting punished)
            and
            $\utility_B = 0$.

    For $N \geq 3$, all (stable) NE of $\PG{N}{x}$ involve
        Bob selecting $\password$ sufficiently randomly
        that Alice does not attempt to guess the password.
        (There might be other NE where Bob randomises in such a way that
                Alice is indifferent between guessing the password and not doing so,
                but does not actually do so.
            However, these equilibria are not very interesting,
                because if Alice started guessing the password,
                Bob would switch to randomising his password more robustly.
            In particular, every NE with non-uniform randomisation for Bob
                gives the same expected utilities as some NE where Bob randomises uniformly.)
\end{example}

This allows us to define the password-guessing modification of an arbitrary base-game $\game$ with opt-out.

\begin{example}[Password guessing]\label{ex:password_guessing_modification}
    Let $\game$ be a game
        between Alice and Bob
        with opt-out and opt-out utilities $(\Baseline_A, \Baseline_B)$.
    Informally,
        \textbf{the password-guessing modification $\passwordGuessing$ of $\game$}
        corresponds to the situation where
            Alice and Bob play $\game$ as usual,
            except that at the end of the game,
                Bob now has to put all his profits into a password-protected account,
                giving Alice an opportunity for theft.
                
    Formally, $\passwordGuessing$ looks like $\game$, except that
        after $\game$ is played
            (with strategy profile $\pureStrategy \in \PureStrategies^\game$), 
        if $\utility_B(\pureStrategy) > \Baseline_B$,
            the players are forced to play a password-guessing game $\PG{3}{\utility_B(\pureStrategy) - \Baseline_B}$.

    Note that if we consider the strategy profiles in which Bob always selects the password uniformly and Alice does not attempt to guess it, then every NE of $\game$ corresponds to an \NE{} of $\passwordGuessing$.
\end{example}

As advertised in the main text,
    password guessing is harmful when Alice has access to pure-strategy simulation,
    but does not change the dynamics when she only has access to mixed-strategy simulation.
    (This holds because
        in the pure-strategy simulation game,
            Alice can always pay the simulation cost to predict Bob's password.
            Indeed, this will, at the minimum, give her a positive payoff from ``feeling smug'' about having guessed the password.
        In contrast, in the mixed-strategy simulation game,
            Alice will only learn that ``Bob selected his password randomly'',
            implying that any password-guessing attempts are not worth the risk for her.)

\begin{restatable}[Mixed-strategy simulation can be better than pure-strategy simulation]{proposition}{propPasswordGuessing}\label{prop:password_guessing}
    Let $\game$ be an arbitrary two-player game with opt-out.
    Then, for sufficiently low $\simcost$:
    \begin{enumerate}[label=(\roman*)]
        \item In the pure-strategy simulation game $\passwordGuessing_\simSubscriptPure$,
        any \NE{} will involve Bob using $\OptOut$ with probability at least $1 - \simcost$;
        \item In the mixed-strategy simulation game $\passwordGuessing_\simSubscriptMixed$,
        every NE of $\game$ is an \NE{} of $\passwordGuessing_\simSubscriptMixed$.
    \end{enumerate}
\end{restatable}

\begin{proof}
    (i):
        Suppose that Bob uses a strategy $\mixedStrategy_2$ with $\mixedStrategy_2(\OO) = 1 - \alpha < 1$.
        Denote by $\mixedStrategyAlt_2$ the strategy that Bob uses in $\game$ when not using opt-out
            (i.e., such that we can write
                $
                    \mixedStrategy_2
                    =
                    (1-\alpha) \cdot \OO
                    +
                    \alpha \cdot \mixedStrategy_2
                $).
        When Alice uses pure-strategy simulation $\pSim$,
            she will pay the simulation cost $\simcost$,
            after which she will
            receive the payoff $\Baseline_1$ with probability $1-\alpha$ (when Bob uses $\OptOut$),
            but with remaining probability $\alpha$,
                she will get to play (without Bob opting out) and will then be able to guess his password.
            This means that her overall payoff will be at least
                $\Baseline_1 - \simcost + \alpha \cdot 1$.
            This means that for $\alpha > \simcost$, Alice's unique best response to $\mixedStrategy_2$ will be $\pSim$.
        However, this will bring Bob's overall payoff to $\Baseline_2 + \alpha \cdot (\shortminus 1)$,
            which means that Bob will have a profitable deviation to $\OO$.

    (ii):
        Let $\mixedStrategy$ be a \NE{} of $\game$ that has $\mixedStrategy_1(\OO) = \mixedStrategy_2(\OO) = 0$.
        Denote by $\metaNE$ the strategy in $\msimgame$ where
            (a) Alice plays $\mixedStrategy_2$ (and does not simulate);
            (b) Bob plays $\mixedStrategy_2$ and then uses a uniformly random distribution over the passwords $\{1, \dots, N \}$.
        Then $\metaNE$ is clearly a \NE{} of $\msimgame$
            (since Alice does not learn any new information from using mixed-strategy simulation,
                and attempting to guess Bob's password has negative expected value when he chooses it uniformly).
\end{proof}

Finally, we use this construction to prove \Cref{theorem:password_guessing}.

\thmMsimCanBeBetterThanPsim*

\begin{proof}[Proof of \Cref{theorem:password_guessing}]
    Let \PTG{} be the partial-trust game \PTG{} from \Cref{fig:PTG}
        and suppose that $\simcost < 0.1$.
    This proof will rely on the following (already-established\arxivOnly{\footnotemark{}}) properties of \PTG{}:
        \begin{enumerate}[label=(\roman*)]
            \item Any \NE{} of \PTG{}
                has \Plone{} using the action $\WO$ with probability $1$,
                which leads to utilities $\utility(\WO, \C) = \utility(\WO, \D) = (0, 0)$.
            \item All payoffs in $\utility_\pl(\pureStrategy_1, \pureStrategy_2)$ in \PTG{}
                lie in the interval $[-100, 100]$.
            \item Mixed-strategy simulation \iPiN{} in \PTG{}.
        \end{enumerate}
        \arxivOnly{\footnotetext{
            Strictly speaking, \Cref{thm:generalised_PTG} only shows that \PTG{} would have a strictly Pareto-improving \NE{}
            if the payoff $\utility_2(\FT, \C)$ were sufficiently high relative to other payoffs in \PTG{}.
            However, the proof of \Cref{thm:generalised_PTG} does give an exact condition,
                Equation \ref{eq:gPTG_sufficiently_high_characterisation},
                which specifies the meaning of ``sufficiently high'',
            and the condition does hold in this specific case.
            (And even if it did not, we could prove \Cref{theorem:password_guessing} by
                replacing \PTG{} by some game with an appropriately higher $\utility_2(\FT, \C)$.)
        }}
    
    Denote by $\game$ the extension of $\PTG$ where
        each player additionally has access to a strategy $\OO$,
        and we have
            $
                \utility (\OO, \pureStrategy_2)
                =
                \utility (\pureStrategy_1, \OO)
                =
                (-200, -200)
            $
        for every $\pureStrategy_1 \in \{ \FT, \PT, \WO \}$, $\pureStrategy_2 \in \{ \C, \D \}$.
    We will show that $\passwordGuessing$ serves as a witness to the statement of the theorem.

    First, note that
        (a) in any $\mixedStrategy \in \NE(\passwordGuessing)$, we have
        $\utility_1(\mixedStrategy), \utility_2(\mixedStrategy) \leq 0$;
        and
        (b) $\passwordGuessing$ admits $\mixedStrategy \in \NE(\passwordGuessing)$ with $\utility(\mixedStrategy) = (0, 0)$.
    (This holds because
        -- since $\OO$ is strictly dominated, against all strategies except for $\OO$ --
        $\NE(\game)$ consist of
            the \NE{} of the original game \PTG{}
            and the opt-out equilibrium $(\OptOut, \OptOut)$.
        By \Cref{ex:password_guessing_modification},
        the \NE{} of $\passwordGuessing$ coincide with the \NE{} of $\game$.)

    Next, we observe that any \NE{} of the pure-strategy simulation game $(\passwordGuessing)_\simSubscriptPure$
        have utilities strictly lower than $0$.
    (By \Cref{prop:password_guessing}\,(i),
        any \NE{} of $(\passwordGuessing)_\simSubscriptPure$ has Bob using $\OptOut$ with probability at least $0.9$.
    This implies that
        for any $\mixedMetaStrategy \in \NE((\passwordGuessing)_\simSubscriptPure)$ and $\pl \in \{1, 2\}$,
        $
            \utility_\pl(\mixedMetaStrategy)
            <
            (-200) \cdot 0.9
            + (100) \cdot 0.1
            < 0
        $.)

    Finally, we note that $(\passwordGuessing)_\simSubscriptMixed$ has a \NE{} where
        both players have utilities strictly higher than $0$.
        (To see this, use the observation (c) above to
            get some $\metaNE$ be the simulation equilibrium of \PTG{} that satisfies
            $\utility_1(\metaNE), \utility_2(\metaNE) > 0$.
        Since the utilities in particular satisfy $\utility_1(\metaNE), \utility_2(\metaNE) > -200$,
            $\metaNE$ is also an \NE{} of $(\passwordGuessing)_\simSubscriptMixed$.)
    This concludes the proof of the theorem.            
\end{proof}

%% file: NPH_game_figures.tex
\begin{figure}[tb]
    \centering
    \begin{NiceTabular}{rccc}[cell-space-limits=3pt]
        & $b_j \in B$ & $a_j \in A$ & $\OO$ \\
        $a_i \in A$
            & \Block[hvlines]{3-3}{}
                \begin{tabular}{@{}c@{}}
                    $(1, 1)$ if $(a_i, b_j) \in E$ \\
                    $(0, 0)$ otherwise \phantom{bla}
                \end{tabular}
            &
                \begin{tabular}{@{}c@{}}
                    $(-k, k)$ if $i = j$ \phantom{b} \\
                    $(0, 0)$ otherwise
                \end{tabular}
            & $-1, 1$ \\
        $b_i \in B$
            &
                \begin{tabular}{@{}c@{}}
                    $(k, -k)$ if $i = j$ \phantom{b} \\
                    $(0, 0)$ otherwise
                \end{tabular}
            & $0, 0$
            & $-1, 1$ \\
        $\OO$
            & $1, -1$
            & $1, -1$
            & $0, 0$
    \end{NiceTabular}
    \\
    \bigskip
    \begin{NiceTabular}{rcc}
        % $\game^i$
            & $\C$     & $\D$     \\
        $\T$  & \Block[hvlines]{2-2}{} 1, 1  & -1, 2 \\
        $\WO$ & 1, -1 & 0, 0
    \end{NiceTabular}
    \phantom{$\WO b$}
    % \\
    % \bigskip
    % \phantom{blahblah}
    \begin{NiceTabular}{ccc}[cell-space-limits=3pt]
        \Block[hvlines]{3-3}{}
        $\game'$ & $0, 0$ & $0, 0$ \\
        $0, 0$ & $\game^1$ & $0, 0$ \\
        $0, 0$ & $0, 0$ & $\game^2$
    \end{NiceTabular}
    \caption{
        The games used in the proof of \Cref{thm:helps_is_NPH_v3}.
        % \\
        \emph{Top:}
            The payoff matrix of the subgame $\game'$.
            The row labelled ``$a^i \in A$'' represents multiple rows -- i.e., those corresponding to \Plone{}'s pure strategies $a \in A$
                (and similarly for the row labelled $b_i \in B$, and the columns labelled $a_j \in A$ and $b^j \in B$).
            In other words, the figure illustrates a matrix game of size $(|A \cup B|+1) \times (|A \cup B|+1)$.
        % \\
        \emph{Bottom Left:}
            Trust games $\game^1$ and $\game^2$.
        % \\
        \emph{Bottom Right:} The full game $\game$, which contains $\game'$, $\game^1$, and $\game^2$ as subgames.
    }
    \Description{The games used in the proof of \Cref{thm:helps_is_NPH_v3}.}
    \label{fig:hardness_graph_subgame}
        \label{fig:TG_for_NPH_result}
        \label{fig:hardness_game}
\end{figure}